\documentclass[letterpaper,12pt,oneside]{article}

\RequirePackage{amsfonts,amssymb,amsmath,amsthm}
\RequirePackage[hang,small,bf]{caption}
\RequirePackage{fancyhdr}
\RequirePackage{achicago}
\RequirePackage{graphicx,graphics,epsfig}
\RequirePackage{appendix}
\RequirePackage{times}

\RequirePackage[english]{babel}
\newtheorem{theorem}{Theorem}
\newtheorem{remark}{Remark}

\usepackage{achicago,boxedminipage}

\begin{document}

\begin{titlepage}
\begin{center}
\setlength{\fboxrule}{0.1cm} 
\setlength{\fboxsep}{1cm} 

\framebox[15.59cm]{
\begin{minipage}[21.94cm]{15.59cm}
\begin{center}
\parbox{13.59cm}{
\begin{center}
\renewcommand{\baselinestretch}{1.5}
{\small\Large\textbf{Combining cluster sampling and link-tracing sampling to estimate 
the size of a hidden population: asymptotic properties of the estimators}}
\normalsize
\end{center}
}

\vspace{1.0cm}

Mart{\'\i}n H. F\'elix Medina

\vspace{2.0cm}
 Technical report 
 
\vspace{1.0cm} 
 
Number: FCFM-UAS-2015-01
 
\vspace{0.5cm}
 Class: Research
 
\vspace{0.5cm}
 June 10, 2015
 
\vspace{2.5cm}

Facultad de Ciencias F{\'\i}sico-Matem\'aticas

\vspace{0.5cm}
 Universidad Aut\'onoma de Sinaloa
 
\vspace{0.5cm}
Ciudad Universitaria, Culiac\'an Sinaloa

\vspace{0.5cm}
M\'exico

\end{center}

\end{minipage}

}

\end{center}

\end{titlepage}

\title{{\Large\textbf{Combining cluster sampling and link-tracing sampling to estimate 
the size of a hidden population: asymptotic properties of the estimators}}}

\author{Mart{\'\i}n H. F\'elix Medina\thanks{mhfelix@uas.edu.mx}\\Facultad de Ciencias 
 F{\'\i}sico-Matem\'aticas de la\\Universidad Aut\'onoma de Sinaloa}

\date{}

\maketitle

\vspace{-1ex}

\begin{abstract}
F\'{e}lix-Medina and Thompson (2004) proposed a variant of link-tracing
sampling to estimate the size of a hidden population such as drug users,
sexual workers or homeless people. In their variant a sampling frame of
sites where the members of the population tend to gather is constructed.
The frame is not assumed to cover the whole population, but only a portion
of it. A simple random sample of sites is selected; the people in the 
sampled sites are identified and are asked to name other members of the 
population which are added to the sample. Those authors proposed maximum 
likelihood estimators of the population size which derived from a 
multinomial model for the numbers of people found in the sampled sites 
and a model that considers that the probability that a person is named 
by any element in a particular sampled site (link-probability) does not 
depend on the named person, that is, that the probabilities are homogeneous. 
Later, F\'{e}lix-Medina et al. (2015) proposed unconditional and conditional
maximum likelihood estimators of the population size which derived from a 
model that takes into account the heterogeneity of the link-probabilities. 
In this work we consider this sampling design and set conditions for a general 
model for the link-probabilities that guarantee the consistency and asymptotic 
normality of the estimators of the population size and of the estimators of 
the parameters of the model for the link-probabilities. In particular we showed
that both the unconditional and conditional maximum likelihood estimators of
the population size are consistent and have asymptotic normal distributions
which are different from each other.
\end{abstract}
 
\vspace{1ex}
 

\newpage

\section{Introduction}

Conventional sampling methods are not appropriate for sampling hidden or
hard-to-reach human populations, such as drug users, sexual-workers and
homeless people, because of the lack of suitable sampling frames. For this
reason, several specific sampling methods for this type of population have
been proposed. See Magnani et al. (2005) and Kalton (2009) for reviews of
some of them. One of this methods is snowball sampling, also known as
link-tracing sampling (LTS) or chain referral sampling. In LTS an initial
sample of members of the population is selected and the sample size is
increased by asking the people in the initial sample to name other members
of the populations. The named people who are not in the initial sample are
added to the sample and they are asked to name other members of the
population. The sampling process might continue in this way until a stopping
rule is satisfied. For reviews of several variants of LTS see Spreen (1992),
Thompson and Frank (2000) and Johnston and Sabin (2010).

F\'{e}lix-Medina and Thompson (2004) proposed a variant of link-tracing
sampling (LTS) to estimate the size of a hidden population. In their variant
they supposed that a sampling frame of sites where the members of the target
population tend to gather can be constructed. As a examples of sites are
public parks, bars and blocks. It is worth nothing that they do not supposed
that the frame covers the whole population, but only a portion of it. Then
an initial sample of sites is selected by a simple random sampling without
replacement design and the members of the population who belong to the
sampled sites are identified. Finally the people in the initial sample are
asked to named other members of the population and the named persons who are
not in the initial sample are included in the sample. Those authors proposed
models to describe the number of members of the population who belong to each
site in the frame and to describe the probability that a person is linked
to a sampled site, that is, that he or she was named by at least one person
who belongs to that site. From those models they derived maximum likelihood
estimators of the population size. In that work those authors considered
that the probability that a person is linked to a site (link-probability)
does not depend on the person, but does on the site, that is, they consider
homogeneous link-probabilities.

F\'{e}lix-Medina and Monjardin (2006) considered this same variant of LTS
and derived estimators of the population size using a Bayesian-assisted
approach, that is, they derived the estimators using the Bayesian approach,
but the inferences were made under a frequentist approach. Those authors
considered an homogeneous two-stage normal model for the logits of the
link-probabilities.

Later F\'{e}lix-Medina et al. (2015) extended the work by F\'{e}lix-Medina
and Thompson (2004) to the case in which the link-probabilities are
heterogeneous, that is, that they depend on the named people. Those authors
modeled the heterogeneity of the link-probabilities by means of a mixed
logistic normal model proposed by Coull and Agresti (1999) in the context of
capture-recapture studies. From this model they derived unconditional and
conditional maximum likelihood estimators of the population size.

In this work we consider the variant of the LTS proposed by F\'{e}lix-Medina
and Thompson (2004) and a general model for the link-probabilities from
which we derive the forms of the unconditional and conditional maximum
likelihood estimators of the population size. We state conditions that
guarantee the consistency and asymptotic normality of both types of
estimators, and we proposed estimators of the variances of the estimators of
the population size. It is worth noting that our work is based on that by
Sanathanan (1972) in which she derived asymptotic properties of both
unconditional and conditional maximum likelihood estimators of the size of a
multinomial distribution from an incomplete observation of the cell totals
which is a situation that occurs in capture-recapture studies. Thus, our
work is basically an adaptation of that by Sanathanan (1972) to the estimators
used in the sampling variant proposed by F\'{e}lix-Medina and Thompson
(2004).

The structure of this document is the following. In section 2 we describe
the variant of LTS proposed by F\'{e}lix-Medina and Thompson (2004). In
section 3 we present probability models that describe the numbers of people
that belong to the sites in the frame and the probabilities of links between
the members of the population and the sites. From these models we construct
the likelihood function that allows us to derive the unconditional and
conditional maximum likelihood estimators of the parameters of the assumed 
model for the link-probabilities and of the population size. In addition, 
we present conditions that guarantee the consistency of the proposed estimators. 
In section 4, which is the central part of this paper, we define the asymptotic
framework under which are derived the asymptotic properties of the proposed
estimators. In section 5 we proposed a method for estimating the variance-covariance
matrices of the estimators of the different vectors of parameters that appear in
the assumed models. Finally, in section 6 we discuss some points to be considered 
whenever the results of this paper want to be used in actual situations. 

\section{Link-tracing sampling design}

In this section we will describe the LTS variant proposed by F\'{e}%
lix-Medina and Thompson (2004). Thus, let $U$ be a finite population of $%
\tau $ people. Let $U_{1}$ be the portion of $U$ that is covered by a
sampling frame of $N$ sites $A_{1},\ldots ,A_{N}$, which are places where
members of the population tend to gather. We will assume that each one of
the\ $\tau _{1}$ persons who are in $U_{1}$ belongs to only one site $A_{i}$
in the frame. Notice that this does not imply that a person cannot be found
in distinct places, but that, as in ordinary cluster sampling, the
researcher has a criterion that allows him or her to assign a person to only
one site. Let $M_{i}$ be the number of people in $U_{1}$ that belong to the
site $A_{i}$, $i=1,\ldots ,N$. The previous assumption implies that $\tau
_{1}=\sum_{1}^{N}M_{i}$. Let $\tau _{2}=\tau -\tau _{1}$ be the number of
people that belong to the portion $U_{2}=U-U_{1}$ of $U$ that is not covered
by the sampling frame.

The sampling procedure is as follows. An initial simple random sample
without replacement (SRSWOR) $S_{A}$ of $n$ sites $A_{1},\ldots ,A_{n}$ is
selected from the frame and the members of the population who belong to each
sampled site are identified. Let $S_{0}$ be the set of people in the initial
sample. Notice that the size of $S_{0}$ is $M=\sum_{1}^{n}M_{i}$. Then from
each sampled site $A_{i}$, $i=1,\ldots ,n$, the people who belong to that
site are asked to name other members of the population. A person and a
sampled site are said to be linked if any of the persons who belong to that
site names that person. Let $S_{1}$ and $S_{2}$ be the sets of people in $%
U_{1}-S_{0}$ and in $U_{2}$, respectively, who are linked to at least one
site in $S_{A}$. Finally, from each named person the following information
is obtained: the portion of $U$ where that person is located, that is, $%
U_{1}-S_{0}$, $A_{i}\in S_{A}$ or $U_{2}$, and the subset of sampled sites
that are linked to him or her.

\section{Unconditional and conditional maximum likelihood estimators}

\subsection{Probability models}

As in F\'{e}lix-Medina and Thompson (2004), we will suppose that the numbers 
$M_{1},\ldots ,$ $M_{N}$ of people who belong to the sites $A_{1},\ldots
,A_{N}$ are independent Poisson random variables with mean $\lambda _{1}$.
Therefore, the joint conditional distribution of $(M_{1},\ldots ,M_{n},\tau
_{1}-M)$ given that $\sum_{1}^{N}M_{i}=\tau _{1}$ is multinomial with
probability mass function (pmf): 
\begin{equation}
f(m_{1},\ldots ,m_{n},\tau _{1}-m|\tau _{1})=\frac{\tau _{1}!}{%
\prod_{1}^{n}m_{i}!(\tau _{1}-m)!}\left( \frac{1}{N}\right) ^{m}\left( 1-%
\frac{n}{N}\right) ^{\tau _{1}-m}.  \label{mult}
\end{equation}%
To model the links between the members of the population and the sampled
sites we will define for person $j$ in $U_{k}-S_{0}$ the vector of
link-indicator variables $\mathbf{X}_{j}^{(k)}=(X_{1j}^{(k)},\ldots
,X_{nj}^{(k)})$, where $X_{ij}^{(k)}=1$ if person $j$ is linked to site $%
A_{i}$ and $X_{ij}^{(k)}=0$ otherwise. Notice that $\mathbf{X}_{j}^{(k)}$
indicates which sites in $S_{A}$ are linked to person $j$. We will suppose
that given $S_{A}$, and consequently the values $M_{i}$s of the sampled
sites, the $X_{ij}^{(k)}$s are Bernoulli random variables with means $%
p_{ij}^{(k)}$s and that the vectors $\mathbf{X}_{j}^{(k)}$ are independent.
Let $\Omega =\{(x_{1},\ldots ,x_{n}):x_{i}=0,1;$ $i=1,\ldots ,n\}$, that is,
the set of all the $n$-dimensional vectors such that each one of their
elements is $0$ or $1$. For $\mathbf{x}=(x_{1},\ldots ,x_{n})\in \Omega $ we
will denote by $\pi _{\mathbf{x}}^{(k)}$ the probability that the vector of
link-indicator variables associated with a randomly selected person from $%
U_{k}-S_{0}$ equals $\mathbf{x}$, that is, the probability that the person
is linked only to the sites $A_{i}$ such that the $i$-th element $x_{i}$ of $%
\mathbf{x}$ equals $1$. We will suppose that $\pi _{\mathbf{x}}^{(k)}$%
depends on a $q_{k}$-dimensional parameter $\boldsymbol{\theta }_{k}=(\theta
_{1}^{(k)},\ldots ,\theta _{q_{k}}^{(k)})\in \mathbf{\Theta }_{k}
\subseteq \mathbb{R}^{q_{k}}$, that is, $\pi _{\mathbf{x}}^{(k)}=
\pi_{\mathbf{x}}^{(k)}(\boldsymbol{\theta }_{k})$, $k=1,2$. In this work we will 
assume that $\boldsymbol{\theta }_{k}$ does not depend on the observed $M_{i}$s.

Similarly, for person $j$ in $A_{i}\in S_{A}$, we will define the vector of
link-indicator variables $\mathbf{X}_{j}^{(A_{i})}=(X_{1j}^{(A_{i})},\ldots
, $ $X_{i-1j}^{(A_{i})},X_{i+1j}^{(A_{i})},\ldots ,X_{nj}^{(A_{i})})$, where 
$X_{i^{\prime }j}^{(A_{i})}=1$ if person $j$ is linked to site $A_{i^{\prime
}}$, $i^{\prime }=1,\ldots ,n$, $i^{\prime }\neq i$ and $X_{i^{\prime
}j}^{(k)}=0$ otherwise. We will suppose that given $S_{A}$ the $X_{i^{\prime
}j}^{(A_{i})}$s are Bernoulli random variables with means $p_{i^{\prime
}j}^{(1)}$s and that the vectors $\mathbf{X}_{j}^{(A_{i})}$ are
independent. For each $A_{i}\in S_{A}$, let $\Omega _{-i}=\{(x_{1},\ldots
,x_{i-1},x_{i+1},\ldots ,x_{n}):x_{i^{\prime }}=0,1;$ $i^{\prime }\neq i,$ $%
i^{\prime }=1,\ldots ,n\}$, that is, the set of all $(n-1)$-dimensional
vectors obtained from the vectors in $\Omega $ by omitting their $i$-th
coordinate. For $\mathbf{x}=(x_{1},\ldots ,x_{i-1},x_{i+1},\ldots ,x_{n})\in
\Omega _{-i}$ we will denote by $\pi _{\mathbf{x}}^{(A_{i})}$ the
probability that the vector of link-indicator variables associated with a
randomly selected person from $A_{i}$ equals $\mathbf{x}$. We will suppose
that $\pi _{\mathbf{x}}^{(A_{i})}$depends on the $q_{1}$-dimensional
parameter $\boldsymbol{\theta }_{1}=(\theta _{1}^{(1)},\ldots ,\theta
_{q_{1}}^{(1)})\in \mathbf{\Theta }_{1}$, that is, $\pi _{\mathbf{x}%
}^{(A_{i})}=\pi _{\mathbf{x}}^{(A_{i})}(\boldsymbol{\theta }_{1})$, $i=1,\ldots
,n$.

For instance, F\'{e}lix-Medina and Monjardin (2006) modeled the
link-probability between person $j$ in $U_{k}-$ $A_{i}$ and site $A_{i}\in
S_{A}$ by $p_{ij}^{(k)}\!=\!\Pr\!\left(\! X_{ij}^{(k)}\! =\! 1|S_{A}\!\right)\!
=\!\exp\! \left(\!\alpha_{i}^{(k)}\!\right)\! /\!\left[\! 1+\!\exp\! \left(\! 
\alpha _{i}^{(k)}\!\right) %
\!\right] $, where the conditional distribution of $\alpha _{i}^{(k)}$ given $%
\psi _{k}$ is normal with mean $\psi _{k}$ and variance $\sigma _{k}^{2}$,
which we denote by $\alpha _{i}^{(k)}|\psi _{k}\sim N\left( \psi _{k},\sigma
_{k}^{2}\right) $ and $\psi _{k}\sim N\left( \mu _{k},\gamma _{k}^{2}\right) 
$. Thus, in this case $\boldsymbol{\theta }_{k}=\left( \mu _{k},\gamma
_{k},\sigma _{k}\right) \in \mathbf{\Theta }_{k}=\mathbb{R}
\times (0,\infty )\times (0,\infty )$, and
{\setlength\arraycolsep{1pt}
\begin{eqnarray*}
\pi _{\mathbf{x}}^{(k)}(\boldsymbol{\theta }_{k})&=&\left[ \int \int \frac{\exp
(\alpha )}{1+\exp (\alpha )}f_{k}(\alpha |\psi )f_{k}(\psi )d\alpha d\psi  %
\right] ^{t}  \\
&&\times\left[ \int \int \frac{1}{1+\exp (\alpha )}f_{k}(\alpha |\psi
)f_{k}(\psi )d\alpha d\psi \right] ^{n-t},
\end{eqnarray*}}%
where $\mathbf{x}=(x_{1},\ldots ,x_{n})\mathbf{\in \Omega }$, $%
t=\sum_{1}^{n}x_{i}$, and $f_{k}(\alpha |\psi )$ and $f_{k}(\psi )$ denote
the probability density functions of the distributions $N\left( \psi
_{k},\sigma _{k}^{2}\right) $ and $N\left( \mu _{k},\gamma _{k}^{2}\right) $%
, respectively. It is worth noting that those authors did not compute $\pi _{%
\mathbf{x}}^{(k)}(\boldsymbol{\theta }_{k})$ because they followed a Bayesian
approach and focused on computing the posterior distribution of the
parameters.

As another example, F\'{e}lix-Medina et al. (2015) modeled the
link-probability between person $j$ in $U_{k}-$ $A_{i}$ and site $A_{i}\in
S_{A}$ by the following Rasch model: $p_{ij}^{(k)}\!=\!\Pr \left(
X_{ij}^{(k)}=1|S_{A}\right)$ $=\exp \left( \alpha _{i}^{(k)}+\beta
_{j}^{(k)}\right) /\left[1+\exp \left( \alpha _{i}^{(k)}+\beta
_{j}^{(k)}\right)\right]$, where $\alpha _{i}^{(k)}$ is a fixed (not random)
effect associated with the site $A_{i}$ and $\beta _{j}^{(k)}$ is a normal
random effect with mean zero and variance $\sigma _{k}^{2}$ associated with
person $j$ in $U_{k}-$ $A_{i}$. Therefore%
\begin{equation*}
\pi _{\mathbf{x}}^{(k)}(\boldsymbol{\theta }_{k})=\int \prod_{i=1}^{n}\frac{\exp %
\left[ x_{i}\left( \alpha _{i}^{(k)}+\sigma _{k}z\right) \right] }{1+\exp
\left( \alpha _{i}^{(k)}+\sigma _{k}z\right) }\phi (z)dz,
\end{equation*}%
where $\mathbf{x}=(x_{1},\ldots ,x_{n})\mathbf{\in \Omega }$, $\boldsymbol{%
\theta }_{k}=(\alpha _{1}^{(k)},\ldots ,\alpha _{n}^{(k)},\sigma _{k})\in 
\mathbf{\Theta }_{k}=\mathbb{R}
^{n}\times (0,\infty )$ and $\phi (\cdot )$ denotes the probability density
function of the standard normal distribution. Those authors compute $\pi _{%
\mathbf{x}}^{(k)}(\boldsymbol{\theta }_{k})$ by means of Gaussian quadrature
formula.

Notice that in the first example the parameter $\boldsymbol{\theta }_{k}$ is
defined previously to the selection of the initial sample because the $%
\alpha _{i}^{(k)}$s are a random sample from a probability distribution
indexed by $\boldsymbol{\theta }_{k}$ and consequently this parameter does not
represent characteristics of the particular selected sample. On the other
hand, in the second example the parameter $\boldsymbol{\theta }_{k}$ is defined
once the initial sample of sites is selected because the $\alpha _{i}^{(k)}$%
s represent characteristics of the particular sites in $S_{A}$. Therefore,
as long as $\boldsymbol{\theta }_{k}$ does not depend on the $M_{i}$s the
results derived in this work are valid for both cases.

\subsection{Likelihood function}

To compute the likelihood function we will factorize it into different
components. One component, $L_{MULT}(\tau _{1})$, is given by the
probability of observing the particular sizes $m_{1},\ldots ,m_{n}$ of the
sites in $S_{A}$; therefore, it is specified by the multinomial distribution
(\ref{mult}). Two additional factors are given by the probabilities of the
configurations of the links between the people in $U_{k}-S_{0}$, $k=1,2$,
and the sites $A_{i}\in S_{A}$. To obtain those factors we will denote by $%
R_{\mathbf{x}}^{(k)}$, $\mathbf{x}=(x_{1},\ldots ,x_{n})\in \Omega $, the
random variable that indicates the number of distinct people in $U_{k}-S_{0}$
whose vectors of link-indicator variables are equal to $\mathbf{x}$, and by $%
R_{k}$ the random variable that indicates the number of distinct people in $%
U_{k}-S_{0}$ who are linked to at least one site $A_{i}\in S_{A}$. Notice
that $R_{k}=\sum_{\mathbf{x}\in \Omega -\{\mathbf{0}\}}R_{\mathbf{x}}^{(k)}$%
, where $\mathbf{0}$ denotes the $n$-dimensional vector of zeros.

Because of the assumptions we made about the vectors $\mathbf{X}_{j}^{(k)}$
of link-indicator variables we have that the conditional joint probability
distribution of the variables $\{R_{\mathbf{x}}^{(1)}\}_{\mathbf{x}\in
\Omega }$ given $S_{A}$ is a multinomial distribution with parameter of size 
$\tau _{1}-m$ and probabilities $\{\pi _{\mathbf{x}}^{(1)}(\boldsymbol{\theta }%
_{1})\}_{\mathbf{x}\in \Omega }$, whereas that of the variables $\{R_{%
\mathbf{x}}^{(2)}\}_{\mathbf{x}\in \Omega }$ is a multinomial distribution
with parameter of size $\tau _{2}$ and probabilities $\{\pi _{\mathbf{x}%
}^{(2)}(\boldsymbol{\theta }_{2})\}_{\mathbf{x}\in \Omega }$. Therefore, the
factors of the likelihood function associated with the probabilities of the
configurations of links between the people in $U_{k}-S_{0}$, $k=1,2$, and
the sites $A_{i}\in S_{A}$ are 
\begin{equation}
L_{1}(\tau _{1},\boldsymbol{\theta }_{1})=\frac{(\tau _{1}-m)!}{(\tau
_{1}-m-r_{1})!\prod_{\mathbf{x}\in \Omega }r_{\mathbf{x}}^{(1)}!}\prod_{%
\mathbf{x}\in \Omega }\left[ \pi _{\mathbf{x}}^{(1)}(\boldsymbol{\theta }_{1})%
\right] ^{r_{\mathbf{x}}^{(1)}}  \label{L1}
\end{equation}%
and 
\begin{equation*}
L_{2}(\tau _{2},\boldsymbol{\theta }_{2})=\frac{\tau _{2}!}{(\tau
_{2}-r_{2})!\prod_{\mathbf{x}\in \Omega }r_{\mathbf{x}}^{(2)}!}\prod_{%
\mathbf{x}\in \Omega }\left[ \pi _{\mathbf{x}}^{(2)}(\boldsymbol{\theta }_{2})%
\right] ^{r_{\mathbf{x}}^{(2)}}.
\end{equation*}%
Notice that $r_{\mathbf{0}}^{(1)}=\tau _{1}-m-r_{1}$ and $r_{\mathbf{0}%
}^{(2)}=\tau _{2}-r_{2}$.

The last factor of the likelihood function is given by the probability of
the configuration of links between the people in $S_{0}$ and the sites $%
A_{i}\in S_{A}$. To obtain this factor, we will denote by $R_{\mathbf{x}%
}^{(A_{i})}$, $\mathbf{x}=(x_{1},\ldots ,x_{i-1},x_{i+1},\ldots ,x_{n})\in
\Omega _{-i}$, the random variable that indicates the number of distinct
people in $A_{i}\in S_{A}$ such that their vectors of link-indicator
variables equal $\mathbf{x}$ and by $R^{(A_{i})}$ the random variable that
indicates the number of distinct people in $A_{i}\in S_{A}$ who are linked
to at least one site $A_{j}\in S_{A}$, $j\neq i$. Notice that $%
R^{(A_{i})}=\sum_{\mathbf{x}\in \Omega _{-i}-\{\mathbf{0}\}}R_{\mathbf{x}%
}^{(A_{i})}$, where $\mathbf{0}$ denotes the $(n-1)$-dimensional vector of
zeros and $R_{\mathbf{0}}^{(A_{i})}=m_{i}-R^{(A_{i})}$. Then, as in the
previous cases, the conditional joint probability distribution of the
variables $\{R_{\mathbf{x}}^{(A_{i})}\}_{\mathbf{x}\in \Omega _{-i}}$ given $%
S_{A}$ is a multinomial distribution with parameter of size $m_{i}$ and
probabilities $\{\pi _{\mathbf{x}}^{(A_{i})}(\boldsymbol{\theta }_{1})\}_{%
\mathbf{x}\in \Omega _{-i}}$. Therefore, the probability of the
configuration of links between the people in $S_{0}$ and the sites $A_{i}\in
S_{A}$ is given by the product of the previous multinomial probabilities
(one for each $A_{i}\in S_{A}$), and consequently the factor of the
likelihood function associated with that probability is 
\begin{equation*}
L_{0}(\boldsymbol{\theta }_{1})=\prod_{i=1}^{n}\frac{m_{i}!}{\prod_{\mathbf{x}%
\in \Omega }r_{\mathbf{x}}^{(A_{i})}!}\prod_{\mathbf{x}\in \Omega _{-i}}%
\left[ \pi _{\mathbf{x}}^{(A_{i})}(\boldsymbol{\theta }_{1})\right] ^{r_{\mathbf{%
x}}^{(A_{i})}}\left[ \pi _{\mathbf{0}}^{(A_{i})}(\boldsymbol{\theta }_{1})\right]
^{m_{i}-r^{(A_{i})}}.
\end{equation*}

From the previous results we have that the maximum likelihood function is
given by 
\begin{equation*}
L(\tau _{1},\tau _{2},\boldsymbol{\theta }_{1},\boldsymbol{\theta }%
_{2})=L_{(1)}(\tau _{1},\boldsymbol{\theta }_{1})L_{(2)}(\tau _{2},\boldsymbol{%
\theta }_{2}),
\end{equation*}%
where{\ 
\begin{eqnarray}
L_{(1)}(\tau _{1},\boldsymbol{\theta }_{1}) &=&L_{MULT}(\tau _{1})L_{1}(\tau
_{1},\boldsymbol{\theta }_{1})L_{0}(\boldsymbol{\theta }_{1})\quad \text{and}
\label{L_(1)} \\
L_{(2)}(\tau _{2},\boldsymbol{\theta }_{2}) &=&L_{2}(\tau _{2},\boldsymbol{\theta }%
_{2}).  \notag
\end{eqnarray}%
}

\subsection{Unconditional and conditional maximum likelihood estimators of $(%
\protect\tau _{k},\boldsymbol{\protect\theta }_{k}^{\ast })$}

In this section we will derive unconditional and conditional maximum
likelihood estimators of the parameters of the previously specified models.
Henceforth we will suppose that conditional on the initial sample $S_{A}$ of
sites the following \textquotedblleft regularity
conditions\textquotedblright\ are satisfied:

\begin{description}
\item[(1)] $\boldsymbol{\theta }_{k}^{\ast }$ is the true value of $\boldsymbol{%
\theta }_{k}$.

\item[(2)] $\boldsymbol{\theta }_{k}^{\ast }$ is an interior point of $\mathbf{%
\Theta }_{k}$.

\item[(3)] $\pi _{\mathbf{x}}^{(k)}(\boldsymbol{\theta }_{k}^{\ast })>0$, $%
\mathbf{x}\in \Omega $ and $\pi _{\mathbf{x}}^{(A_{i})}(\boldsymbol{\theta }%
_{1}^{\ast })>0$, $\mathbf{x}\in \Omega _{-i}$, $i=1,\ldots ,n.$

\item[(4)] $\partial \pi _{\mathbf{x}}^{(k)}(\boldsymbol{\theta }_{k})/\partial
\theta _{j}^{(k)}$, $\mathbf{x}\in \Omega $ and $\partial \pi _{\mathbf{x}%
}^{(A_{i})}(\boldsymbol{\theta }_{1})/\partial \theta _{j}^{(1)}$, $\mathbf{x}%
\in \Omega _{-i}$, $i=1,\ldots ,n$; $j=1,\ldots ,q_{k}$, exist at any $%
\boldsymbol{\theta }_{k}\in \mathbf{\Theta }_{k}$ and $\boldsymbol{\theta }_{1}\in 
\mathbf{\Theta }_{1}$, and are continuous in neighborhoods of $\boldsymbol{%
\theta }_{k}^{\ast }$ and $\boldsymbol{\theta }_{1}^{\ast }$, respectively.

\item[(5)] Given a $\delta _{1}>0$, it is possible to find an $\varepsilon
_{1}>0$ such that{\setlength\arraycolsep{1pt} 
\begin{eqnarray*}
&&\inf_{\left\Vert \boldsymbol{\theta }_{1}-\boldsymbol{\theta }_{1}^{\ast
}\right\Vert >\delta _{1}}\sum\limits_{\mathbf{x}\in \Omega -\{\mathbf{0}\}}%
\frac{\pi _{\mathbf{x}}^{(1)}(\boldsymbol{\theta }_{1}^{\ast })}{1-\pi _{\mathbf{%
0}}^{(1)}(\boldsymbol{\theta }_{1}^{\ast })}\ln \left\{ \frac{\pi _{\mathbf{x}%
}^{(1)}(\boldsymbol{\theta }_{1}^{\ast })/\left[ 1-\pi _{\mathbf{0}}^{(1)}(%
\boldsymbol{\theta }_{1}^{\ast })\right] }{\pi _{\mathbf{x}}^{(1)}(\boldsymbol{%
\theta }_{1})/\left[ 1-\pi _{\mathbf{0}}^{(1)}(\boldsymbol{\theta }_{1})\right] }%
\right\} \\
&&+\frac{1}{(N-n)\left[ 1-\pi _{\mathbf{0}}^{(1)}(\boldsymbol{\theta }_{1}^{\ast
})\right] }\sum_{i=1}^{n}\sum_{\mathbf{x}\in \Omega _{-i}}\pi _{\mathbf{x}%
}^{(A_{i})}(\boldsymbol{\theta }_{1}^{\ast })\ln \left[ \frac{\pi _{\mathbf{x}%
}^{(A_{i})}(\boldsymbol{\theta }_{1}^{\ast })}{\pi _{\mathbf{x}}^{(A_{i})}(%
\boldsymbol{\theta }_{1})}\right] \geq \varepsilon _{1}.
\end{eqnarray*}%
}

\item[(6)] Given a $\delta _{2}>0$, it is possible to find an $\varepsilon
_{2}>0$ such that%
\begin{equation*}
\inf_{\left\Vert \boldsymbol{\theta }_{2}-\boldsymbol{\theta }_{2}^{\ast
}\right\Vert >\delta _{2}}\sum_{\mathbf{x}\in \Omega -\{\mathbf{0}\}}\frac{%
\pi _{\mathbf{x}}^{(2)}(\boldsymbol{\theta }_{2}^{\ast })}{1-\pi _{\mathbf{0}%
}^{(2)}(\boldsymbol{\theta }_{2}^{\ast })}\ln \left\{ \frac{\pi _{\mathbf{x}%
}^{(2)}(\boldsymbol{\theta }_{2}^{\ast })/\left[ 1-\pi _{\mathbf{0}}^{(2)}(%
\boldsymbol{\theta }_{2}^{\ast })\right] }{\pi _{\mathbf{x}}^{(2)}(\boldsymbol{%
\theta }_{2})/\left[ 1-\pi _{\mathbf{0}}^{(2)}(\boldsymbol{\theta }_{2})\right] }%
\right\} \geq \varepsilon _{2}.
\end{equation*}
\end{description}

\begin{remark}
For a differentiable function $f:\mathbb{R}^{q}\rightarrow \mathbb{R}$, the
notation $\partial f(\mathbf{x}_{0})/\partial x_{j}$ represents $\partial f(%
\mathbf{x})/$ $\partial x_{j}|_{\mathbf{x}=\mathbf{x}_{0}}$.
\end{remark}

The regularity conditions ({\bf 1})-({\bf 4}) and ({\bf 6}) or conditions 
equivalent to them have been assumed by several authors such as Birch (1964), 
Rao (1973, Ch. 5), Bishop et al. (1975, Ch. 14), Sanathanan (1972) and Agresti 
(2002, Ch. 14), among others, in the context of deriving asymptotic properties 
of estimators of the parameters of models for the probabilities of a
multinomial distribution. The particular form of condition ({\bf 6}) comes from
Sanathanan (1972) who took it from the first edition of Rao (1973, Ch. 5) and
it is known as a \textit{strong identifiability} condition. Condition ({\bf 5}) is
a modification of ({\bf 6}) to meet the requirements of our particular sampling
design. In general, these conditions imply the existence and consistency of
the UMLEs and CMLEs of $\boldsymbol{\theta }_{1}^{\ast }$ and $\boldsymbol{\theta }%
_{2}^{\ast }$, and that they can be obtained deriving the likelihood
function with respect to $\boldsymbol{\theta }_{1}$ and $\boldsymbol{\theta }_{2}$.

\subsubsection{Unconditional and conditional maximum likelihood estimators
of $\protect\tau _{1}$ and $\boldsymbol{\protect\theta }_{1}^{\ast }$}

Let us firstly consider the unconditional maximum likelihood estimators
(UMLEs) $\hat{\tau}_{1}^{(U)}$ and $\boldsymbol{\hat{\theta}}_{1}^{(U)}$ of $%
\tau _{1}$ and $\boldsymbol{\theta }_{1}^{\ast }$. The log-likelihood function
of $\tau _{1}$ and $\boldsymbol{\theta }_{1}$ is%
\begin{eqnarray*}
l_{(1)}(\tau _{1},\boldsymbol{\theta }_{1}) &=&\ln [L_{(1)}(\tau _{1},\boldsymbol{%
\theta }_{1})] \\
&=&\ln (\tau _{1}!)-\ln [(\tau _{1}-m-r_{1})!]+\tau _{1}\ln (1-n/N) \\
&&+\sum\nolimits_{\mathbf{x}\in \Omega }r_{\mathbf{x}}^{(1)}\ln \left[ \pi
_{\mathbf{x}}^{(1)}(\boldsymbol{\theta }_{1})\right] +\sum\nolimits_{i=1}^{n}%
\sum\nolimits_{\mathbf{x}\in \Omega _{-i}}r_{\mathbf{x}}^{(A_{i})}\ln \left[
\pi _{\mathbf{x}}^{(A_{i})}(\boldsymbol{\theta }_{1})\right] +C,
\end{eqnarray*}%
where $C$ does not depend on $\tau _{1}$ and $\boldsymbol{\theta }_{1}$, and
recall that $r_{\mathbf{0}}^{(1)}=\tau _{1}-m-r_{1}$ and $r_{\mathbf{0}%
}^{(A_{i})}=m_{i}-r^{(A_{i})}$. Then, the UMLE $\boldsymbol{\hat{\theta}}%
_{1}^{(U)}$ of $\boldsymbol{\theta }_{1}^{\ast }$ is the solution to the
following equations:%
\begin{equation}
\frac{\partial l_{(1)}(\tau _{1},\boldsymbol{\theta}_{1})}{\partial\theta_{j}^{(1)}}\!
=\!\sum\limits_{\mathbf{x}\in \Omega }\frac{r_{\mathbf{x}}^{(1)}}{\pi _{%
\mathbf{x}}^{(1)}(\boldsymbol{\theta }_{1})}\frac{\partial \pi _{\mathbf{x}%
}^{(1)}(\boldsymbol{\theta }_{1})}{\partial \theta _{j}^{(1)}}%
+\sum\limits_{i=1}^{n}\sum\limits_{\mathbf{x}\in \Omega _{-i}}\frac{r_{%
\mathbf{x}}^{(A_{i})}}{\pi _{\mathbf{x}}^{(A_{i})}(\boldsymbol{\theta }_{1})}%
\frac{\partial \pi _{\mathbf{x}}^{(A_{i})}(\boldsymbol{\theta }_{1})}{\partial
\theta _{j}^{(1)}}\!=\!0,\text{ }j=1,\ldots,q_{1}.  \label{derivlnl1/thetaj}
\end{equation}

Since $\tau _{1}$ is an integer we will use the \textquotedblleft ratio
method\textquotedblright\ to maximize $L_{(1)}(\tau _{1},\boldsymbol{\theta }%
_{1})$. [See Feller (1968, Ch. 3).] Thus%
\begin{equation*}
\frac{L_{(1)}(\tau _{1},\boldsymbol{\theta }_{1})}{L_{(1)}(\tau _{1}-1,\boldsymbol{%
\theta }_{1})}=\frac{\tau _{1}(1-n/N)\pi _{\mathbf{0}}^{(1)}(\boldsymbol{\theta }%
_{1})}{(\tau _{1}-m-r_{1})}.
\end{equation*}%
Since this ratio is greater than or equal to $1$ if $\tau _{1}\leq $ $%
(m+r_{1})/\left[ 1-(1-n/N)\pi _{\mathbf{0}}^{(1)}(\boldsymbol{\theta }%
_{1})\right] $ and it is smaller than or equal to $1$ if $\tau _{1}$ is
greater than or equal to that quantity, it follows that $\hat{\tau}%
_{1}^{(U)} $ is given by%
\begin{equation}
\hat{\tau}_{1}^{(U)}=\left\lfloor \frac{M+R_{1}}{1-(1-n/N)\pi _{\mathbf{0}%
}^{(1)}\left( \boldsymbol{\hat{\theta}}_{1}^{(U)}\right) }\right\rfloor ,
\label{umlet1}
\end{equation}%
where $\left\lfloor x\right\rfloor $ denotes the largest integer not greater
than $x$. Notice that the right hand-side of (\ref{umlet1}) is not a closed
form for $\hat{\tau}_{1}^{(U)}$ since this expression depends on $\boldsymbol{%
\hat{\theta}}_{1}^{(U)}$. In fact, $\hat{\tau}_{1}^{(U)}$ and $\boldsymbol{\hat{%
\theta}}_{1}^{(U)}$ are obtained by simultaneously solving the set of
equations (\ref{derivlnl1/thetaj}) and (\ref{umlet1}), which is generally
done by numerical methods.

Let us now consider the conditional maximum likelihood estimators (CMLEs) $%
\hat{\tau}_{1}^{(C)}$ and $\boldsymbol{\hat{\theta}}_{1}^{(C)}$ of $\tau _{1}$
and $\boldsymbol{\theta }_{1}^{\ast }$. It is worth noting that this type of
estimators was proposed by Sanathanan (1972) in the context of estimating
the parameter of size of a multinomial distribution from an incomplete
observation of the cell frequencies. The approach we will follow to derive $%
\hat{\tau}_{1}^{(C)}$ and $\boldsymbol{\hat{\theta}}_{1}^{(C)}$ is an adaptation
of Sanathanan's (1972) approach to our case. Thus, from (\ref{L1}) we
have that{\setlength\arraycolsep{1pt}%
\begin{eqnarray}
L_{1}(\tau _{1},\boldsymbol{\theta }_{1}) &=&f\left( \{r_{\mathbf{x}}^{(1)}\}_{%
\mathbf{x}\in \Omega }|\{m_{i}\},\tau _{1},\boldsymbol{\theta }_{1}\right) 
\notag \\
&=&f\left( \{r_{\mathbf{x}}^{(1)}\}_{\mathbf{x}\in \Omega -\{\mathbf{0}%
\}}|r_{1},\{m_{i}\},\tau _{1},\boldsymbol{\theta }_{1}\right) f\left(
r_{1}|\{m_{i}\},\tau _{1},\boldsymbol{\theta }_{1}\right)  \notag \\
&=&\frac{r_{1}!}{\prod_{\mathbf{x}\in \Omega -\{\mathbf{0}\}}r_{\mathbf{x}%
}^{(1)}!}\prod_{\mathbf{x}\in \Omega -\{\mathbf{0}\}}\left[ \frac{\pi _{%
\mathbf{x}}^{(1)}(\boldsymbol{\theta }_{1})}{1-\pi _{\mathbf{0}}^{(1)}(\boldsymbol{%
\theta }_{1})}\right] ^{r_{\mathbf{x}}^{(1)}}  \notag \\
&&\times \frac{(\tau _{1}-m)!}{(\tau _{1}-m-r_{1})!r_{1}!}\left[ 1-\pi _{%
\mathbf{0}}^{(1)}(\boldsymbol{\theta }_{1})\right] ^{r_{1}}\left[ \pi _{\mathbf{0%
}}^{(1)}(\boldsymbol{\theta }_{1})\right] ^{\tau _{1}-m-r_{1}}  \notag \\
&=&L_{11}(\boldsymbol{\theta }_{1})L_{12}(\tau _{1},\boldsymbol{\theta }_{1})
\label{L_1}
\end{eqnarray}%
Notice that the first factor }$L_{11}(\boldsymbol{\theta }_{1})$ is given by the
joint pmf of the multinomial distribution with parameter of size $r_{1}$ and
probabilities $\left\{ \pi _{\mathbf{x}}^{(1)}(\boldsymbol{\theta }_{1})/\left[
1-\pi _{\mathbf{0}}^{(1)}(\boldsymbol{\theta }_{1})\right] \right\} _{\mathbf{x}%
\in \Omega -\{\mathbf{0}\}}$ and that this distribution does not depend on $%
\tau _{1}$. Note also that the second factor $L_{12}(\tau _{1},\boldsymbol{%
\theta }_{1})$ is given by the pmf of the binomial distribution with
parameter of size $\tau _{1}-m$ and probability $1-\pi _{\mathbf{0}}^{(1)}(%
\boldsymbol{\theta }_{1})$. Thus, the CMLE $\boldsymbol{\hat{\theta}}_{1}^{(C)}$ of $%
\boldsymbol{\theta }_{1}^{\ast }$ is the solution to the following system of
equations:%
\begin{eqnarray}
\frac{\partial }{\partial \theta _{j}^{(1)}}\ln [L_{11}(\boldsymbol{\theta }%
_{1})L_{0}(\boldsymbol{\theta }_{1})] &=&\sum\limits_{\mathbf{x}\in \Omega -\{%
\mathbf{0}\}}\frac{r_{\mathbf{x}}^{(1)}}{\pi _{\mathbf{x}}^{(1)}(\boldsymbol{%
\theta }_{1})}\frac{\partial \pi _{\mathbf{x}}^{(1)}(\boldsymbol{\theta }_{1})}{%
\partial \theta _{j}^{(1)}}+\frac{r_{1}}{1-\pi _{\mathbf{0}}^{(1)}(\boldsymbol{%
\theta }_{1})}\frac{\partial \pi _{\mathbf{0}}^{(1)}(\boldsymbol{\theta }_{1})}{%
\partial \theta _{j}^{(1)}}  \notag \\
&&+\sum\limits_{i=1}^{n}\sum\limits_{\mathbf{x}\in \Omega _{-i}}\frac{r_{%
\mathbf{x}}^{(A_{i})}}{\pi _{\mathbf{x}}^{(A_{i})}(\boldsymbol{\theta }_{1})}%
\frac{\partial \pi _{\mathbf{x}}^{(A_{i})}(\boldsymbol{\theta }_{1})}{\partial
\theta _{j}^{(1)}}=0,\text{ }j=1,\ldots ,q_{1}.  \label{derivlnL11L0}
\end{eqnarray}%
The CMLE $\hat{\tau}_{1}^{(C)}$ of $\tau _{1}$ is obtained by the ratio
method. Thus, since%
\begin{equation*}
\frac{L_{MULT}(\tau _{1})L_{12}\left( \tau _{1},\boldsymbol{\theta }_{1}\right) 
}{L_{MULT}(\tau _{1}-1)L_{12}\left( \tau _{1}-1,\boldsymbol{\theta }_{1}\right) }%
=\frac{\tau _{1}(1-n/N)\pi _{\mathbf{0}}^{(1)}\left( \boldsymbol{\theta }%
_{1}\right) }{(\tau _{1}-m-r_{1})},
\end{equation*}%
it follows that 
\begin{equation}
\hat{\tau}_{1}^{(C)}=\left\lfloor \frac{M+R_{1}}{1-(1-n/N)\pi _{\mathbf{0}%
}^{(1)}\left( \boldsymbol{\hat{\theta}}_{1}^{(C)}\right) }\right\rfloor .
\label{cmlet1}
\end{equation}%
Note that (\ref{cmlet1}) is a closed form for $\hat{\tau}_{1}^{(C)}$ since $%
\boldsymbol{\hat{\theta}}_{1}^{(C)}$ is firstly obtained from (\ref{derivlnL11L0}%
).

\subsubsection{Unconditional and conditional maximum likelihood estimators
of $\protect\tau_{2}$ and $\boldsymbol{\protect\theta }_{2}^{\ast }$}

By a similar analysis as that conducted in the previous subsection we have
that the UMLEs $\hat{\tau}_{2}^{(U)}$ and $\boldsymbol{\hat{\theta}}_{2}^{(U)}$
of $\tau _{2}$ and $\boldsymbol{\theta }_{2}^{\ast }$ are the solution to the
following equations:%
\begin{equation*}
\sum\limits_{\mathbf{x}\in \Omega }\frac{r_{\mathbf{x}}^{(2)}}{\pi _{%
\mathbf{x}}^{(2)}(\boldsymbol{\theta }_{2})}\frac{\partial \pi _{\mathbf{x}%
}^{(2)}(\boldsymbol{\theta }_{2})}{\partial \theta _{j}^{(2)}}=0,\text{ }%
j=1,\ldots ,q_{2}
\end{equation*}%
and%
\begin{equation}
\hat{\tau}_{2}^{(U)}=\left\lfloor \frac{R_{2}}{1-\pi _{\mathbf{0}%
}^{(2)}\left( \boldsymbol{\hat{\theta}}_{2}^{(U)}\right) }\right\rfloor .
\label{CMLEt2}
\end{equation}%
where recall that $r_{\mathbf{0}}^{(2)}=\tau _{2}-r_{2}$.

With respect to the conditional estimators, we have that the CMLE $\boldsymbol{%
\hat{\theta}}_{2}^{(C)}$ of $\boldsymbol{\theta }_{2}^{\ast }$ is the solution
to the following equations:%
\begin{equation*}
\sum\limits_{\mathbf{x}\in \Omega -\{\mathbf{0}\}}\frac{r_{\mathbf{x}}^{(2)}%
}{\pi _{\mathbf{x}}^{(2)}(\boldsymbol{\theta }_{2})}\frac{\partial \pi _{\mathbf{%
x}}^{(2)}(\boldsymbol{\theta }_{2})}{\partial \theta _{j}^{(2)}}+\frac{r_{2}}{%
1-\pi _{\mathbf{0}}^{(2)}(\boldsymbol{\theta }_{2})}\frac{\partial \pi _{\mathbf{%
0}}^{(2)}(\boldsymbol{\theta }_{2})}{\partial \theta _{j}^{(2)}}=0,\text{ }%
j=1,\ldots ,q_{2}.
\end{equation*}%
The CMLE $\hat{\tau}_{2}^{(C)}$ of $\tau _{2}$ is given by (\ref{CMLEt2}),
but replacing $\boldsymbol{\hat{\theta}}_{2}^{(U)}$ by $\boldsymbol{\hat{\theta}}%
_{2}^{(C)}$. Note that in this case (\ref{CMLEt2}) is a closed form for $%
\hat{\tau}_{2}^{(C)}$.

\subsubsection{Unconditional and conditional maximum likelihood estimators
of $\protect\tau =\protect\tau _{1}+\protect\tau _{2}$}

The UMLE and CMLE of $\tau =\tau _{1}+\tau _{2}$ are given by $\hat{\tau}%
^{(U)}=\hat{\tau}_{1}^{(U)}+\hat{\tau}_{2}^{(U)}$ and $\hat{\tau}^{(C)}=\hat{%
\tau}_{1}^{(C)}+\hat{\tau}_{2}^{(C)}$, respectively.

\section{Asymptotic properties of the unconditional and conditional maximum
likelihood estimators}

The structure of this section is as follows. Firstly we will define the
asymptotic framework under which we will derive the asymptotic properties of
the estimators. Next we will state and proof a theorem that guarantees the
asymptotic multivariate normal distribution of any estimator of $(\tau_{1},%
\boldsymbol{\theta}_{1}^{\ast})$ that satisfies the conditions expressed in
the theorem. Since not any estimator of $(\tau_{1},\boldsymbol{\theta }%
_{1}^{\ast})$ satisfies the conditions of the theorem, in particular the
CMLE does not, we will state and proof another theorem that guarantees the 
asymptotic multivariate normal distribution of any estimator of $\boldsymbol{\theta}%
_{1}^{\ast}$ that satisfies the conditions of that theorem. Then, we will
prove that the UMLE of $(\tau_{1},\boldsymbol{\theta }_{1}^{\ast})$
satisfies the conditions of the first theorem, whereas the CMLE of 
$\boldsymbol{\theta}_{1}^{\ast}$ satisfies those of the second one. In addition, we
will prove that in spite of that result, the CMLE $\hat{\tau}_{1}^{(C)}$
does have an asymptotic normal distribution although it is not the same as that of 
$\hat{\tau}_{1}^{(U)}$. After that we will consider the
asymptotic properties of estimators of $(\tau _{2},\boldsymbol{\theta}%
_{2}^{\ast})$. Since this problem is exactly the same as that considered by
Sanathanan (1972), we will only state a theorem that guarantees the
asymptotic multivariate normal distribution of any estimator of $(\tau_{2},%
\boldsymbol{\theta}_{2}^{\ast})$ that satisfies the conditions expressed in
the theorem, but we will omit its proof, as well as the proofs that both the
UMLE and the CMLE of $(\tau _{2},\boldsymbol{\theta}_{2}^{\ast})$ satisfy the
conditions of that theorem. Finally, we will obtain the asymptotic
properties of the estimators $\hat{\tau}^{(U)}$ and $\hat{\tau}^{(C)}$ of $%
\tau $.

\subsection{Basic assumptions}

To derive the asymptotic properties of the UMLEs and CMLEs of $\tau_{k}$ and
$\boldsymbol{\theta}_{k}^{\ast}$, $k=1,2$, we will make the following
assumptions:

\begin{description}
\item[A.] $\tau _{k}\rightarrow \infty $, $k=1,2.$

\item[B.] $\tau _{k}/\tau \rightarrow \alpha _{k}$, $0<\alpha _{k}<1$, $%
k=1,2 $.

\item[C.] $N$ and $n$ are fixed positive integer numbers.
\end{description}

For convenience of notation, we will put $\tau _{k}$ either as a subscript
or a superscript of every term that depends on $\tau _{k}$, $k=1,2$. In
addition, convergence in distribution will be denoted by $\overset{D}{%
\rightarrow }$ and convergence in probability by $\overset{P}{\rightarrow }$.

Notice that from (\ref{mult}) it follows that the conditional distribution
of $M_{i}^{(\tau _{1})}$ given $\tau _{1}$ is binomial with parameter of
size $\tau _{1}$ and probability $1/N$, that is $M_{i}^{(\tau _{1})}|\tau
_{1}\sim $ Bin$(\tau _{1},1/N)$; consequently $M_{i}^{(\tau _{1})}/\tau _{1}$
is stochastically bounded, that is, $M_{i}^{(\tau _{1})}=O_{p}(\tau _{1})$.
This means that the size of $U_{1}^{(\tau _{1})}$ is increased by increasing
the sizes of the clusters, even though their number $N$ is kept fixed. In
the same manner, the number of people in the initial sample $S_{0}^{(\tau
_{1})}$, given by $M^{(\tau _{1})}=\sum_{1}^{n}M_{i}^{(\tau _{1})}|\tau
_{1}\sim $ Bin$(\tau _{1},n/N)$, is increased because of the increasing of $%
M_{i}^{(\tau _{1})}$, $i=1,\ldots ,n$, even though $n$ is kept fixed. On the
other hand, since $\tau _{1}-M^{(\tau _{1})}|\tau _{1}\sim $ Bin$(\tau
_{1},1-n/N)$, $R_{1}^{(\tau _{1})}|S_{A}^{(\tau _{1})}\sim $ Bin$(\tau
_{1}-M^{(\tau _{1})},1-\pi _{\mathbf{0}}^{(1)})$ and $R_{2}^{(\tau
_{2})}|S_{A}^{(\tau _{1})}\sim $ Bin$(\tau _{2},1-\pi _{\mathbf{0}}^{(2)})$,
it follows that $R_{1}^{(\tau _{1})}|\tau _{1}\sim $ Bin$\left[ \tau
_{1},\left( 1-n/N\right) \left( 1-\pi _{\mathbf{0}}^{(1)}\right) \right] $
and $R_{2}^{(\tau _{2})}|\tau _{2}\sim $ Bin$(\tau _{2},1-\pi _{\mathbf{0}%
}^{(2)})$; therefore $R_{1}^{(\tau _{1})}=O_{p}(\tau _{1})$ and $%
R_{2}^{(\tau _{2})}=O_{p}(\tau _{2})$. Thus, the \ sizes of the sets $%
S_{1}^{(\tau _{1})}$ and $S_{2}^{(\tau _{2})}$ are increased because $\tau
_{1}$ and $\tau _{2}$ are increased even though the probabilities $\{\pi _{%
\mathbf{x}}^{(1)}\}_{\mathbf{x}\in \Omega }$ and $\{\pi _{\mathbf{x}%
}^{(2)}\}_{\mathbf{x}\in \Omega }$ are kept fixed.

We will end this subsection presenting the conditional and unconditional
distributions of the variables $R_{\mathbf{x}}^{(\tau _{1})}$, $R_{\mathbf{x}%
}^{(A_{i})}$ and $R_{\mathbf{x}}^{(\tau _{2})}$ which will be used later in
this work. Thus, from the multinomial distributions indicated in Subsection
3.1 it follows that $R_{\mathbf{x}}^{(\tau _{1})}|S_{A}^{(\tau _{1})}\sim $
Bin$(\tau _{1}-M^{(\tau _{1})},\pi _{\mathbf{x}}^{(1)})$, $R_{\mathbf{x}%
}^{(A_{i})}|M_{i}^{(\tau _{1})}\!\!\sim\! $ Bin$(M_{i}^{(\tau _{1})},\pi _{\mathbf{%
x}}^{(A_{i})})$ and $R_{\mathbf{x}}^{(\tau _{2})}|S_{A}^{(\tau _{1})}\!\!\sim\! $
Bin$(\tau _{2},\pi _{\mathbf{x}}^{(2)});$ therefore $R_{\mathbf{x}}^{(\tau
_{1})}|\tau _{1}$ $\sim$Bin$\left[ \tau _{1},\left( 1-n/N\right) \pi _{\mathbf{%
x}}^{(1)}\right] $, $R_{\mathbf{x}}^{(A_{i})}|\tau _{1}\sim $Bin$\left( \tau
_{1},\pi _{\mathbf{x}}^{(A_{i})}/N\right) $ and $R_{\mathbf{x}}^{(\tau
_{2})}|\tau _{2}\sim $Bin$(\tau _{2},\pi _{\mathbf{x}}^{(2)})$.

\subsection{Asymptotic multivariate normal distribution of estimators of $(%
\protect\tau_{1},\!\boldsymbol{\protect\theta}_{1}^{\ast})$}

\begin{theorem}
Let $\boldsymbol{\theta }_{1}^{\ast }=(\theta_{1}^{\ast },\ldots ,\theta
_{q_{1}}^{\ast })$ be the true value of $\boldsymbol{\theta }_{1}$. Let $\hat{%
\tau}_{1}^{(\tau _{1})}$ and $\boldsymbol{\hat{\theta}}_{1}^{(\tau _{1})}=(%
\hat{\theta}_{11}^{(\tau _{1})},\ldots ,$ $\hat{\theta}%
_{1q_{1}}^{(\tau _{1})})$ be estimators of $\tau _{1}$ and $\boldsymbol{\theta }%
_{1}^{\ast }$, such that

\begin{description}
\item[(i)] $\boldsymbol{\hat{\theta}}_{1}^{(\tau _{1})}\overset{P}{\rightarrow }%
\boldsymbol{\theta }_{1}^{\ast }.$

\item[(ii)] $\tau _{1}^{-1/2}\left\{ \hat{\tau}_{1}^{(\tau _{1})}-\left(
M^{(\tau _{1})}+R_{1}^{(\tau _{1})}\right) /\left[ 1-(1-n/N)\pi _{\mathbf{0}%
}^{(1)}\left( \boldsymbol{\hat{\theta}}_{1}^{(\tau _{1})}\right) \right]
\right\} \overset{P}{\rightarrow }0.$

\item[(iii)] $\tau _{1}^{-1/2}\left[ \frac{\partial }{\partial \theta
_{j}^{(1)}}l_{(1)}^{(\tau _{1})}(\hat{\tau}_{1}^{(\tau _{1})},\boldsymbol{\hat{%
\theta}}_{1}^{(\tau _{1})})\right] \overset{P}{\rightarrow }0$, $j=1,\ldots
,q_{1}.$
\end{description}

In addition, let $\mathbf{\Sigma }_{1}^{-1}$ be the $(q_{1}+1)\times
(q_{1}+1)$ matrix whose elements are
{\setlength\arraycolsep{0pt} 
\begin{eqnarray*}
\left[ \mathbf{\Sigma }_{1}^{-1}\right] _{1,1} &=&\left[ 1-(1-n/N)\pi _{%
\mathbf{0}}^{(1)}(\boldsymbol{\theta }_{1}^{\ast })\right] /\left[ (1-n/N)\pi _{%
\mathbf{0}}^{(1)}(\boldsymbol{\theta }_{1}^{\ast })\right] , \\
\left[ \mathbf{\Sigma }_{1}^{-1}\right] _{1,j+1} &=&\left[ \mathbf{\Sigma }%
_{1}^{-1}\right] _{j+1,1}=-\left[ 1/\pi _{\mathbf{0}}^{(1)}(\boldsymbol{\theta }%
_{1}^{\ast })\right] \left[ \partial \pi _{\mathbf{0}}^{(1)}(\boldsymbol{\theta }%
_{1}^{\ast })/\partial \theta _{j}^{(1)}\right] \text{, \ }j=1,\ldots ,q_{1},
\\
\left[ \mathbf{\Sigma }_{1}^{-1}\right] _{i+1,j+1} &=&\left[ \mathbf{\Sigma }%
_{1}^{-1}\right] _{j+1,i+1}\!=\!\left(1-\frac{n}{N}\right)\!\!\sum_{\mathbf{x}\in \Omega }\!%
\left[ 1/\pi _{\mathbf{x}}^{(1)}(\boldsymbol{\theta }_{1}^{\ast })\right]\!\! \left[
\partial \pi _{\mathbf{x}}^{(1)}(\boldsymbol{\theta }_{1}^{\ast })/\partial
\theta _{i}^{(1)}\right]\!\! \left[ \partial \pi _{\mathbf{x}}^{(1)}(\boldsymbol{%
\theta }_{1}^{\ast })/\partial \theta _{j}^{(1)}\right] \\
&&+\frac{1}{N}\sum\nolimits_{l=1}^{n}\sum\nolimits_{\mathbf{x}\in \Omega
_{-l}}\left[ 1/\pi _{\mathbf{x}}^{(A_{l})}(\boldsymbol{\theta }_{1}^{\ast })%
\right] \left[ \partial \pi _{\mathbf{x}}^{(A_{l})}(\boldsymbol{\theta }%
_{1}^{\ast })/\partial \theta _{i}^{(1)}\right] \left[ \partial \pi _{%
\mathbf{x}}^{(A_{l})}(\boldsymbol{\theta }_{1}^{\ast })/\partial \theta
_{j}^{(1)}\right] , \\
&&\hspace{50ex}i,j=1,\ldots ,q_{1},
\end{eqnarray*}}%
and which is assumed to be a non-singular matrix. Then%
\begin{equation*}
\left[ \tau _{1}^{-1/2}\left( \hat{\tau}_{1}^{(\tau _{1})}-\tau _{1}\right)
,\tau _{1}^{1/2}\left( \boldsymbol{\hat{\theta}}_{1}^{(\tau _{1})}-\boldsymbol{%
\theta }_{1}^{\ast }\right) \right] \overset{D}{\rightarrow }%
N_{q_{1}+1}\left( \mathbf{0},\mathbf{\Sigma }_{1}\right) ,
\end{equation*}%
where $\mathbf{\Sigma }_{1}$ is the inverse of $\mathbf{\Sigma }_{1}^{-1}$
and $\mathbf 0=(0,\ldots,0)\in\mathbb{R}^{q_{1}+1}$.
\end{theorem}

\begin{proof}
Evaluating equation (\ref{derivlnl1/thetaj}) at $(\hat{\tau}_{1}^{(\tau
_{1})},\boldsymbol{\hat{\theta}}_{1}^{(\tau _{1})})$ we get%
{\setlength\arraycolsep{0.0pt} 
\begin{eqnarray}
&&\frac{\partial }{\partial \theta _{j}^{(1)}}l_{(1)}^{(\tau _{1})}\left( 
\hat{\tau}_{1}^{(\tau _{1})},\boldsymbol{\hat{\theta}}_{1}^{(\tau _{1})}\right) 
\notag \\
&=&\sum\limits_{\mathbf{x}\in \Omega -\mathbf{0}}\frac{R_{\mathbf{x}%
}^{(\tau _{1})}}{\pi _{\mathbf{x}}^{(1)}(\boldsymbol{\hat{\theta}}_{1}^{(\tau
_{1})})}\frac{\partial \pi _{\mathbf{x}}^{(1)}(\boldsymbol{\hat{\theta}}%
_{1}^{(\tau _{1})})}{\partial \theta _{j}^{(1)}}+\frac{\hat{\tau}_{1}^{(\tau
_{1})}-M^{(\tau _{1})}-R_{1}^{(\tau _{1})}}{\pi _{\mathbf{0}}^{(1)}(\boldsymbol{%
\hat{\theta}}_{1}^{(\tau _{1})})}\frac{\partial \pi _{\mathbf{0}}^{(1)}(%
\boldsymbol{\hat{\theta}}_{1}^{(\tau _{1})})}{\partial \theta _{j}^{(1)}}  \notag
\\
&&+\sum\limits_{l=1}^{n}\sum\limits_{\mathbf{x}\in \Omega _{-l}}\frac{R_{%
\mathbf{x}}^{(A_{l},\tau _{1})}}{\pi _{\mathbf{x}}^{(A_{l})}(\mathbf{\hat{%
\theta}}_{1}^{(\tau _{1})})}\frac{\partial \pi _{\mathbf{x}}^{(A_{l})}(%
\boldsymbol{\hat{\theta}}_{1}^{(\tau _{1})})}{\partial \theta _{j}^{(1)}}  \notag
\\
&=&\sum\limits_{\mathbf{x}\in \Omega }\frac{R_{\mathbf{x}}^{(\tau _{1})}}{%
\pi _{\mathbf{x}}^{(1)}(\boldsymbol{\hat{\theta}}_{1}^{(\tau _{1})})}\frac{%
\partial \pi _{\mathbf{x}}^{(1)}(\boldsymbol{\hat{\theta}}_{1}^{(\tau _{1})})}{%
\partial \theta _{j}^{(1)}}+\frac{\left[ \hat{\tau}_{1}^{(\tau
_{1})}\!-\! M^{(\tau _{1})}\!-\!R_{1}^{(\tau _{1})}\right] -\left[ \tau _{1}\!-M^{(\tau
_{1})}\!-\!R_{1}^{(\tau _{1})}\right] }{\pi _{\mathbf{0}}^{(1)}(\mathbf{\hat{%
\theta}}_{1}^{(\tau _{1})})}\frac{\partial \pi _{\mathbf{0}}^{(1)}(\boldsymbol{%
\hat{\theta}}_{1}^{(\tau _{1})})}{\partial \theta _{j}^{(1)}}  \notag \\
&&+\sum\limits_{l=1}^{n}\sum\limits_{\mathbf{x}\in \Omega _{-l}}\frac{R_{%
\mathbf{x}}^{(A_{l},\tau _{1})}}{\pi _{\mathbf{x}}^{(A_{l})}(\mathbf{\hat{%
\theta}}_{1}^{(\tau _{1})})}\frac{\partial \pi _{\mathbf{x}}^{(A_{l})}(%
\boldsymbol{\hat{\theta}}_{1}^{(\tau _{1})})}{\partial \theta _{j}^{(1)}}  \notag
\\
&=&\sum\limits_{\mathbf{x}\in \Omega }\frac{R_{\mathbf{x}}^{(\tau _{1})}}{%
\pi _{\mathbf{x}}^{(1)}(\boldsymbol{\hat{\theta}}_{1}^{(\tau _{1})})}\frac{%
\partial \pi _{\mathbf{x}}^{(1)}(\boldsymbol{\hat{\theta}}_{1}^{(\tau _{1})})}{%
\partial \theta _{j}^{(1)}}\! +\!\frac{\hat{\tau}_{1}^{(\tau _{1})}\! -\!\tau _{1}}{%
\pi _{\mathbf{0}}^{(1)}(\boldsymbol{\hat{\theta}}_{1}^{(\tau _{1})})}\frac{%
\partial \pi _{\mathbf{0}}^{(1)}(\boldsymbol{\hat{\theta}}_{1}^{(\tau _{1})})}{%
\partial \theta _{j}^{(1)}}\! +\!\sum\limits_{l=1}^{n}\sum\limits_{\mathbf{x}%
\in \Omega _{-l}}\frac{R_{\mathbf{x}}^{(A_{l},\tau _{1})}}{\pi _{\mathbf{x}%
}^{(A_{l})}(\boldsymbol{\hat{\theta}}_{1}^{(\tau _{1})})}\frac{\partial \pi _{%
\mathbf{x}}^{(A_{l})}(\boldsymbol{\hat{\theta}}_{1}^{(\tau _{1})})}{\partial
\theta _{j}^{(1)}}.  \notag \\
&&  \label{deriv-l1}
\end{eqnarray}}

Since%
\begin{equation}
\sum\limits_{\mathbf{x}\in \Omega }\partial \pi _{\mathbf{x}}^{(1)}(\boldsymbol{%
\hat{\theta}}_{1}^{(\tau _{1})})/\partial \theta _{j}^{(1)}=0\text{ \ and \ }%
\sum\limits_{\mathbf{x}\in \Omega _{-l}}\partial \pi _{\mathbf{x}%
}^{(A_{l})}(\boldsymbol{\hat{\theta}}_{1}^{(\tau _{1})})/\partial \theta
_{j}^{(1)}=0,  \label{sum-deriv0}
\end{equation}
from (\ref{deriv-l1}) we get that%
\begin{eqnarray}
&&\tau _{1}^{-1/2}\left\{ \sum\limits_{\mathbf{x}\in \Omega }\frac{R_{%
\mathbf{x}}^{(\tau _{1})}-(\tau _{1}-M^{(\tau _{1})})\pi _{\mathbf{x}}^{(1)}(%
\boldsymbol{\theta }_{1}^{\ast })}{\pi _{\mathbf{x}}^{(1)}(\boldsymbol{\hat{\theta}}%
_{1}^{(\tau _{1})})}\frac{\partial \pi _{\mathbf{x}}^{(1)}(\boldsymbol{\hat{%
\theta}}_{1}^{(\tau _{1})})}{\partial \theta _{j}^{(1)}}\right.  \notag \\
&&+\left. \sum\limits_{l=1}^{n}\sum\limits_{\mathbf{x}\in \Omega _{-l}}%
\frac{R_{\mathbf{x}}^{(A_{l},\tau _{1})}-M_{l}^{(\tau _{1})}\pi _{\mathbf{x}%
}^{(A_{l})}(\boldsymbol{\theta }_{1}^{\ast })}{\pi _{\mathbf{x}}^{(A_{l})}(%
\boldsymbol{\hat{\theta}}_{1}^{(\tau _{1})})}\frac{\partial \pi _{\mathbf{x}%
}^{(A_{l})}(\boldsymbol{\hat{\theta}}_{1}^{(\tau _{1})})}{\partial \theta
_{j}^{(1)}}\right\} -\tau _{1}^{-1/2}\frac{\partial }{\partial \theta
_{j}^{(1)}}l_{(1)}^{(\tau _{1})}\left( \hat{\tau}_{1}^{(\tau _{1})},\boldsymbol{%
\hat{\theta}}_{1}^{(\tau _{1})}\right)  \notag \\
&=&-\tau _{1}^{-1/2}(\hat{\tau}_{1}^{(\tau _{1})}-\tau _{1})\left[ \frac{1}{%
\pi _{\mathbf{0}}^{(1)}(\boldsymbol{\hat{\theta}}_{1}^{(\tau _{1})})}\frac{%
\partial \pi _{\mathbf{0}}^{(1)}(\boldsymbol{\hat{\theta}}_{1}^{(\tau _{1})})}{%
\partial \theta _{j}^{(1)}}\right]  \notag \\
&&+\tau _{1}^{1/2}\left\{ \frac{\tau _{1}-M^{(\tau _{1})}}{\tau _{1}}%
\sum\limits_{\mathbf{x}\in \Omega }\frac{\pi _{\mathbf{x}}^{(1)}(\boldsymbol{%
\hat{\theta}}_{1}^{(\tau _{1})})-\pi _{\mathbf{x}}^{(1)}(\boldsymbol{\theta }%
_{1}^{\ast })}{\pi _{\mathbf{x}}^{(1)}(\boldsymbol{\hat{\theta}}_{1}^{(\tau
_{1})})}\frac{\partial \pi _{\mathbf{x}}^{(1)}(\boldsymbol{\hat{\theta}}%
_{1}^{(\tau _{1})})}{\partial \theta _{j}^{(1)}}\right.  \notag \\
&&+\left. \sum\limits_{l=1}^{n}\frac{M_{l}^{(\tau _{1})}}{\tau _{1}}%
\sum\limits_{\mathbf{x}\in \Omega _{-l}}\frac{\pi _{\mathbf{x}}^{(A_{l})}(%
\boldsymbol{\hat{\theta}}_{1}^{(\tau _{1})})-\pi _{\mathbf{x}}^{(A_{l})}(\boldsymbol{%
\theta }_{1}^{\ast })}{\pi _{\mathbf{x}}^{(A_{l})}(\boldsymbol{\hat{\theta}}%
_{1}^{(\tau _{1})})}\frac{\partial \pi _{\mathbf{x}}^{(A_{l})}(\mathbf{\hat{%
\theta}}_{1}^{(\tau _{1})})}{\partial \theta _{j}^{(1)}}\right\} .
\label{expre-med}
\end{eqnarray}

Let $Y_{\mathbf{x}}^{(\tau _{1})}=R_{\mathbf{x}}^{(\tau _{1})}-(\tau
_{1}-M^{(\tau _{1})})\pi _{\mathbf{x}}^{(1)}(\boldsymbol{\theta }_{1}^{\ast })$, 
$Y_{\mathbf{x}}^{(A_{l},\tau _{1})}=R_{\mathbf{x}}^{(A_{l},\tau
_{1})}-M_{l}^{(\tau _{1})}\pi _{\mathbf{x}}^{(A_{l})}(\boldsymbol{\theta }%
_{1}^{\ast })$ and%
\begin{eqnarray*}
Z_{j+1}^{(\tau _{1})} &=&\tau _{1}^{-1/2}\left[ \sum\limits_{\mathbf{x}\in
\Omega }\frac{R_{\mathbf{x}}^{(\tau _{1})}}{\pi _{\mathbf{x}}^{(1)}(\boldsymbol{%
\theta }_{1}^{\ast })}\frac{\partial \pi _{\mathbf{x}}^{(1)}(\boldsymbol{\theta }%
_{1}^{\ast })}{\partial \theta _{j}^{(1)}}+\sum\limits_{l=1}^{n}\sum%
\limits_{\mathbf{x}\in \Omega _{-l}}\frac{R_{\mathbf{x}}^{(A_{l},\tau _{1})}%
}{\pi _{\mathbf{x}}^{(A_{l})}(\boldsymbol{\theta }_{1}^{\ast })}\frac{\partial
\pi _{\mathbf{x}}^{(A_{l})}(\boldsymbol{\theta }_{1}^{\ast })}{\partial \theta
_{j}^{(1)}}\right] \\
&=&\tau _{1}^{-1/2}\left[ \sum\limits_{\mathbf{x}\in \Omega }\frac{Y_{%
\mathbf{x}}^{(\tau _{1})}}{\pi _{\mathbf{x}}^{(1)}(\boldsymbol{\theta }%
_{1}^{\ast })}\frac{\partial \pi _{\mathbf{x}}^{(1)}(\boldsymbol{\theta }%
_{1}^{\ast })}{\partial \theta _{j}^{(1)}}+\sum\limits_{l=1}^{n}\sum%
\limits_{\mathbf{x}\in \Omega _{-l}}\frac{Y_{\mathbf{x}}^{(A_{l},\tau _{1})}%
}{\pi _{\mathbf{x}}^{(A_{l})}(\boldsymbol{\theta }_{1}^{\ast })}\frac{\partial
\pi _{\mathbf{x}}^{(A_{l})}(\boldsymbol{\theta }_{1}^{\ast })}{\partial \theta
_{j}^{(1)}}\right] ,
\end{eqnarray*}%
where the last equality is obtained using (\ref{sum-deriv0}) but replacing $%
\boldsymbol{\hat{\theta}}_{1}^{(\tau _{1})}$ by $\boldsymbol{\theta }_{1}^{\ast }$.
Then, the difference between the left-hand side of (\ref{expre-med}) and $%
Z_{j+1}^{(\tau _{1})}$ is given by%
\begin{eqnarray}
&&\tau _{1}^{-1/2}\left\{ \sum\limits_{\mathbf{x}\in \Omega }\frac{Y_{%
\mathbf{x}}^{(\tau _{1})}}{\pi _{\mathbf{x}}^{(1)}(\boldsymbol{\hat{\theta}}%
_{1}^{(\tau _{1})})}\frac{\partial \pi _{\mathbf{x}}^{(1)}(\boldsymbol{\hat{%
\theta}}_{1}^{(\tau _{1})})}{\partial \theta _{j}^{(1)}}+\sum%
\limits_{l=1}^{n}\sum\limits_{\mathbf{x}\in \Omega _{-l}}\frac{Y_{\mathbf{x}%
}^{(A_{l},\tau _{1})}}{\pi _{\mathbf{x}}^{(A_{l})}(\boldsymbol{\hat{\theta}}%
_{1}^{(\tau _{1})})}\frac{\partial \pi _{\mathbf{x}}^{(A_{l})}(\boldsymbol{\hat{%
\theta}}_{1}^{(\tau _{1})})}{\partial \theta _{j}^{(1)}}\right\}  \notag \\
&&-\tau _{1}^{-1/2}\frac{\partial }{\partial \theta _{j}^{(1)}}%
l_{(1)}^{(\tau _{1})}\left( \hat{\tau}_{1}^{(\tau _{1})},\boldsymbol{\hat{\theta}%
}_{1}^{(\tau _{1})}\right) -Z_{j+1}^{(\tau _{1})}  \notag \\
&=&\tau _{1}^{-1/2}\left\{ \sum\limits_{\mathbf{x}\in \Omega }Y_{\mathbf{x}%
}^{(\tau _{1})}\left[ \frac{1}{\pi _{\mathbf{x}}^{(1)}(\boldsymbol{\hat{\theta}}%
_{1}^{(\tau _{1})})}\frac{\partial \pi _{\mathbf{x}}^{(1)}(\boldsymbol{\hat{%
\theta}}_{1}^{(\tau _{1})})}{\partial \theta _{j}^{(1)}}-\frac{1}{\pi _{%
\mathbf{x}}^{(1)}(\boldsymbol{\theta }_{1}^{\ast })}\frac{\partial \pi _{\mathbf{%
x}}^{(1)}(\boldsymbol{\theta }_{1}^{\ast })}{\partial \theta _{j}^{(1)}}\right]
\right.  \notag \\
&&+\left. \sum\limits_{l=1}^{n}\sum\limits_{\mathbf{x}\in \Omega _{-l}}Y_{%
\mathbf{x}}^{(A_{l},\tau _{1})}\left[ \frac{1}{\pi _{\mathbf{x}}^{(A_{l})}(%
\boldsymbol{\hat{\theta}}_{1}^{(\tau _{1})})}\frac{\partial \pi _{\mathbf{x}%
}^{(A_{l})}(\boldsymbol{\hat{\theta}}_{1}^{(\tau _{1})})}{\partial \theta
_{j}^{(1)}}-\frac{1}{\pi _{\mathbf{x}}^{(A_{l})}(\boldsymbol{\theta }_{1}^{\ast
})}\frac{\partial \pi _{\mathbf{x}}^{(A_{l})}(\boldsymbol{\theta }_{1}^{\ast })}{%
\partial \theta _{j}^{(1)}}\right] \right\}  \notag \\
&&-\tau _{1}^{-1/2}\frac{\partial }{\partial \theta _{j}^{(1)}}%
l_{(1)}^{(\tau _{1})}\left( \hat{\tau}_{1}^{(\tau _{1})},\boldsymbol{\hat{\theta}%
}_{1}^{(\tau _{1})}\right) .  \label{left-hand-side-zj+1}
\end{eqnarray}%
Since unconditionally $E(Y_{\mathbf{x}}^{(\tau _{1})})=0$ and $V(Y_{\mathbf{x%
}}^{(\tau _{1})})=\tau _{1}(1-n/N)\pi _{\mathbf{x}}^{(1)}(\boldsymbol{\theta }%
_{1}^{\ast })[1-\pi _{\mathbf{x}}^{(1)}(\boldsymbol{\theta }_{1}^{\ast })]$, and
also $E(Y_{\mathbf{x}}^{(A_{l},\tau _{1})})=0$ and $V(Y_{\mathbf{x}%
}^{(A_{l},\tau _{1})})=\tau _{1}(1/N)\pi _{\mathbf{x}}^{(A_{l})}(\boldsymbol{%
\theta }_{1}^{\ast })[1-\pi _{\mathbf{x}}^{(A_{l})}(\boldsymbol{\theta }%
_{1}^{\ast })]$, it follows that $\tau _{1}^{-1/2}Y_{\mathbf{x}}^{(\tau
_{1})}=O_{p}(1)$ and $\tau _{1}^{-1/2}Y_{\mathbf{x}}^{(A_{l},\tau
_{1})}=O_{p}(1)$. Consequently, these results along with conditions \textbf{%
(3)-(4)} and conditions (\textbf{i}) and (\textbf{iii}) of the theorem imply
that (\ref{left-hand-side-zj+1}) converges to zero in probability.

On the other hand, by the mean value theorem for functions of several variables 
we have that%
\begin{eqnarray}
\pi _{\mathbf{x}}^{(1)}(\boldsymbol{\hat{\theta}}_{1}^{(\tau _{1})})-\pi _{%
\mathbf{x}}^{(1)}(\boldsymbol{\theta }_{1}^{\ast })
&=&\sum\nolimits_{i=1}^{q_{1}}\left( \hat{\theta}_{1i}^{(\tau _{1})}-\theta
_{1i}^{\ast }\right) \partial \pi _{\mathbf{x}}^{(1)}(\boldsymbol{\theta }_{1%
\mathbf{x}}^{(\tau _{1})})/\partial \theta _{i}^{(1)}\text{ \ \ and}
\label{dif-pix} \\
\pi _{\mathbf{x}}^{(A_{l})}(\boldsymbol{\hat{\theta}}_{1}^{(\tau _{1})})-\pi _{%
\mathbf{x}}^{(A_{l})}(\boldsymbol{\theta }_{1}^{\ast })
&=&\sum\nolimits_{i=1}^{q_{1}}\left( \hat{\theta}_{1i}^{(\tau _{1})}-\theta
_{1i}^{\ast }\right) \partial \pi _{\mathbf{x}}^{(A_{l})}(\boldsymbol{\theta }%
_{A_{l}\mathbf{x}}^{(\tau _{1})})/\partial \theta _{i}^{(1)},  \notag
\end{eqnarray}%
where $\boldsymbol{\theta }_{1\mathbf{x}}^{(\tau _{1})}$ and $\boldsymbol{\theta }%
_{A_{l}\mathbf{x}}^{(\tau _{1})}$ are between $\boldsymbol{\hat{\theta}}%
_{1}^{(\tau _{1})}$ and $\boldsymbol{\theta }_{1}^{\ast }$. Since the difference
between the right-hand side of (\ref{expre-med}) and $Z_{j+1}^{(\tau _{1})}$
also converges to zero in probability, we have that%
\begin{eqnarray}
&&-\tau _{1}^{-1/2}(\hat{\tau}_{1}^{(\tau _{1})}-\tau _{1})\left[ \frac{1}{%
\pi _{\mathbf{0}}^{(1)}(\boldsymbol{\hat{\theta}}_{1}^{(\tau _{1})})}\frac{%
\partial \pi _{\mathbf{0}}^{(1)}(\boldsymbol{\hat{\theta}}_{1}^{(\tau _{1})})}{%
\partial \theta _{j}^{(1)}}\right]  \notag \\
&&+\tau _{1}^{1/2}\left\{ \frac{\tau _{1}-M^{(\tau _{1})}}{\tau _{1}}%
\sum\limits_{\mathbf{x}\in \Omega }\frac{1}{\pi _{\mathbf{x}}^{(1)}(\boldsymbol{%
\hat{\theta}}_{1}^{(\tau _{1})})}\frac{\partial \pi _{\mathbf{x}}^{(1)}(%
\boldsymbol{\hat{\theta}}_{1}^{(\tau _{1})})}{\partial \theta _{j}^{(1)}}%
\sum_{i=1}^{q_{1}}\left( \hat{\theta}_{1i}^{(\tau _{1})}-\theta _{1i}^{\ast
}\right) \frac{\partial \pi _{\mathbf{x}}^{(1)}(\boldsymbol{\theta }_{1\mathbf{x}%
}^{(\tau _{1})})}{\partial \theta _{i}^{(1)}}\right.  \notag \\
&&+\left. \sum\limits_{l=1}^{n}\frac{M_{l}^{(\tau _{1})}}{\tau _{1}}%
\sum\limits_{\mathbf{x}\in \Omega _{-l}}\frac{1}{\pi _{\mathbf{x}%
}^{(A_{l})}(\boldsymbol{\hat{\theta}}_{1}^{(\tau _{1})})}\frac{\partial \pi _{%
\mathbf{x}}^{(A_{l})}(\boldsymbol{\hat{\theta}}_{1}^{(\tau _{1})})}{\partial
\theta _{j}^{(1)}}\sum\limits_{i=1}^{q_{1}}\left( \hat{\theta}_{1i}^{(\tau
_{1})}-\theta _{1i}^{\ast }\right) \frac{\partial \pi _{\mathbf{x}%
}^{(A_{l})}(\boldsymbol{\theta }_{A_{l}\mathbf{x}}^{(\tau _{1})})}{\partial
\theta _{i}^{(1)}}\right\} -Z_{j+1}^{(\tau _{1})}  \notag \\
&=&\left[ \mathbf{\hat{\Sigma}}_{1}^{-1}\right] _{j+1,1}\left[ \tau
_{1}^{-1/2}(\hat{\tau}_{1}^{(\tau _{1})}-\tau _{1})\right]
+\sum_{i=1}^{q_{1}}\left[ \mathbf{\hat{\Sigma}}_{1}^{-1}\right] _{j+1,i+1}%
\left[ \tau _{1}^{1/2}\left( \hat{\theta}_{1i}^{(\tau _{1})}-\theta
_{1i}^{\ast }\right) \right] -Z_{j+1}^{(\tau _{1})}\overset{P}{\rightarrow }%
0,  \notag \\
&& \label{right-hand-side-zj+1}
\end{eqnarray}%
where%
\begin{eqnarray}
\left[ \mathbf{\hat{\Sigma}}_{1}^{-1}\right] _{j+1,1} &=&-\frac{1}{\pi _{%
\mathbf{0}}^{(1)}(\boldsymbol{\hat{\theta}}_{1}^{(\tau _{1})})}\frac{\partial
\pi _{\mathbf{0}}^{(1)}(\boldsymbol{\hat{\theta}}_{1}^{(\tau _{1})})}{\partial
\theta _{j}^{(1)}}\text{ \ \ \ and }  \notag \\
\left[ \mathbf{\hat{\Sigma}}_{1}^{-1}\right] _{j+1,i+1} &=&\frac{\tau
_{1}-M^{(\tau _{1})}}{\tau _{1}} \sum\limits_{\mathbf{x}\in \Omega }%
\frac{1}{\pi _{\mathbf{x}}^{(1)}(\boldsymbol{\hat{\theta}}_{1}^{(\tau _{1})})}%
\frac{\partial \pi _{\mathbf{x}}^{(1)}(\boldsymbol{\hat{\theta}}_{1}^{(\tau
_{1})})}{\partial \theta _{j}^{(1)}}\frac{\partial \pi _{\mathbf{x}}^{(1)}(%
\boldsymbol{\theta }_{1\mathbf{x}}^{(\tau _{1})})}{\partial \theta _{i}^{(1)}} 
\notag \\
&&+\sum\limits_{l=1}^{n}\frac{M_{l}^{(\tau _{1})}}{\tau _{1}}\sum\limits_{%
\mathbf{x}\in \Omega _{-l}}\frac{1}{\pi _{\mathbf{x}}^{(A_{l})}(\boldsymbol{\hat{%
\theta}}_{1}^{(\tau _{1})})}\frac{\partial \pi _{\mathbf{x}}^{(A_{l})}(%
\boldsymbol{\hat{\theta}}_{1}^{(\tau _{1})})}{\partial \theta _{j}^{(1)}}\frac{%
\partial \pi _{\mathbf{x}}^{(A_{l})}(\boldsymbol{\theta }_{A_{l}\mathbf{x}%
}^{(\tau _{1})})}{\partial \theta _{i}^{(1)}}.  \label{hat-sigma-matrix11}
\end{eqnarray}

Expression (\ref{umlet1}) suggests the following equality in terms of $\hat{%
\tau}_{1}^{(\tau _{1})}-\tau _{1}$ and $\pi _{\mathbf{0}}^{(1)}(\boldsymbol{\hat{%
\theta}}_{1}^{(\tau _{1})})-\pi _{\mathbf{0}}^{(1)}(\boldsymbol{\theta }%
_{1}^{\ast }):$ 
{\setlength\arraycolsep{1pt}
\begin{eqnarray*}
&&\tau _{1}^{-1/2}\left\{ \hat{\tau}_{1}^{(\tau _{1})}\left[ 1-(1-n/N)\pi _{%
\mathbf{0}}^{(1)}(\boldsymbol{\hat{\theta}}_{1}^{(\tau _{1})})\right] -\left(
M^{(\tau _{1})}+R_{1}^{(\tau _{1})}\right) \right\} =\tau _{1}^{-1/2}\left( 
\hat{\tau}_{1}^{(\tau _{1})}-\tau _{1}\right) \\
&&\times \left[ 1\! -\! (1 \!- \! n/N)\pi _{\mathbf{0}}^{(1)}(\boldsymbol{\hat{\theta}}%
_{1}^{(\tau _{1})})\right]\! -\!\tau _{1}^{-1/2}\left\{ \left( M^{(\tau
_{1})}\!+\! R_{1}^{(\tau _{1})}\right)\! -\!\tau _{1}\left[ 1\! -\! (1\! -\! n/N)\pi _{\mathbf{0}%
}^{(1)}(\boldsymbol{\theta }_{1}^{\ast })\right] \right\} \\
&&-\tau _{1}^{1/2}(1-n/N)\left[ \pi _{\mathbf{0}}^{(1)}(\boldsymbol{\hat{\theta}}%
_{1}^{(\tau _{1})})-\pi _{\mathbf{0}}^{(1)}(\boldsymbol{\theta }_{1}^{\ast })%
\right] .
\end{eqnarray*}}

By condition (ii) of the theorem it follows that the left hand-side of the
previous equation converges to zero in probability. Therefore, if we divide
the right hand-side of this equation by $(1-n/N)\pi _{\mathbf{0}}^{(1)}(%
\boldsymbol{\theta }_{1}^{\ast })$ and use (\ref{dif-pix}), we will get that the
following expression also converges to zero in probability, that is

{\setlength\arraycolsep{0pt} 
\begin{eqnarray}
&&\tau _{1}^{-1/2}\left( \hat{\tau}_{1}^{(\tau _{1})}-\tau _{1}\right) \frac{%
 1-(1-n/N)\pi _{\mathbf{0}}^{(1)}(\boldsymbol{\hat{\theta}}_{1}^{(\tau
_{1})}) }{(1-n/N)\pi _{\mathbf{0}}^{(1)}(\boldsymbol{\theta }_{1}^{\ast })%
}  \notag \\
&&-\tau _{1}^{-1/2}\frac{\left( M^{(\tau _{1})}+R_{1}^{(\tau _{1})}\right)
-\tau _{1}\left[ 1-(1-n/N)\pi _{\mathbf{0}}^{(1)}(\boldsymbol{\theta }_{1}^{\ast
})\right] }{(1-n/N)\pi _{\mathbf{0}}^{(1)}(\boldsymbol{\theta }_{1}^{\ast })} 
\notag \\
&&-\sum_{i=1}^{q_{1}}\tau _{1}^{1/2}\left( \hat{\theta}_{1i}^{(\tau
_{1})}-\theta _{1i}^{\ast }\right) \frac{1}{\pi _{\mathbf{0}}^{(1)}(\boldsymbol{%
\theta }_{1}^{\ast })}\frac{\partial \pi _{\mathbf{0}}^{(1)}(\boldsymbol{\theta }%
_{1\mathbf{0}}^{(\tau _{1})})}{\partial \theta _{i}^{(1)}},  \notag \\
&=&\left[ \mathbf{\hat{\Sigma}}_{1}^{-1}\right]_{1,1}\!\left[ \tau
_{1}^{-1/2}\left( \hat{\tau}_{1}^{(\tau _{1})}\! -\!\tau _{1}\right) \right]
\! +\!\sum\nolimits_{i=1}^{q_{1}}\!\left[ \mathbf{\hat{\Sigma}}_{1}^{-1}\right]
_{1,i+1}\!\left[ \tau _{1}^{1/2}\left( \hat{\theta}_{1i}^{(\tau _{1})}\! -\!\theta
_{1i}^{\ast }\right) \right]\! -\! Z_{1}^{(\tau _{1})}\overset{P}{\rightarrow }0,
\label{expr-z1}
\end{eqnarray}}
where $\boldsymbol{\theta }_{1\mathbf{0}}^{(\tau _{1})}$ is between $\boldsymbol{%
\hat{\theta}}_{1}^{(\tau _{1})}$ and $\boldsymbol{\theta }_{1}^{\ast }$ and%
\begin{equation}
\left[ \mathbf{\hat{\Sigma}}_{1}^{-1}\right] _{1,1}=\frac{1-(1-n/N)\pi _{%
\mathbf{0}}^{(1)}(\boldsymbol{\hat{\theta}}_{1}^{(\tau _{1})})}{(1-n/N)\pi _{%
\mathbf{0}}^{(1)}(\boldsymbol{\theta }_{1}^{\ast })}\text{, \ }\left[ \mathbf{%
\hat{\Sigma}}_{1}^{-1}\right] _{1,i+1}=-\frac{1}{\pi _{\mathbf{0}}^{(1)}(%
\boldsymbol{\theta }_{1}^{\ast })}\frac{\partial \pi _{\mathbf{0}}^{(1)}(\boldsymbol{%
\theta }_{1\mathbf{0}}^{(\tau _{1})})}{\partial \theta _{i}^{(1)}}
\label{hat-sigma-matrix-12}
\end{equation}%
and%
\begin{equation*}
Z_{1}^{(\tau _{1})}=\tau _{1}^{-1/2}\frac{\left( M^{(\tau
_{1})}+R_{1}^{(\tau _{1})}\right) -\tau _{1}\left[ 1-(1-n/N)\pi _{\mathbf{0}%
}^{(1)}(\boldsymbol{\theta }_{1}^{\ast })\right] }{(1-n/N)\pi _{\mathbf{0}%
}^{(1)}(\boldsymbol{\theta }_{1}^{\ast })}.
\end{equation*}

Let $\mathbf{W}_{1}^{(\tau _{1})}=[\tau _{1}^{-1/2}(\hat{\tau}_{1}^{(\tau
_{1})}-\tau _{1}),\tau _{1}^{1/2}(\boldsymbol{\hat{\theta}}_{1}^{(\tau _{1})}-%
\boldsymbol{\theta }_{1}^{\ast })]^{\prime }$ and $\mathbf{Z}^{(\tau
_{1})}=[Z_{1}^{(\tau _{1})},Z_{2}^{(\tau _{1})},\ldots ,Z_{q_{1}+1}^{(\tau
_{1})}]^{\prime }$, by the previous results we have that%
\begin{equation}
\mathbf{\hat{\Sigma}}_{1}^{-1}\mathbf{W}_{1}^{(\tau _{1})}-\mathbf{Z}^{(\tau
_{1})}\overset{P}{\rightarrow }\mathbf{0,}  \label{difsigmauz}
\end{equation}%
where $\mathbf{\hat{\Sigma}}_{1}^{-1}$ is the $(q_{1}+1)\times (q_{1}+1)$
matrix whose elements are defined in (\ref{hat-sigma-matrix11}) and (\ref%
{hat-sigma-matrix-12}). Notice that from the definitions of the matrices $%
\mathbf{\Sigma }_{1}^{-1}$ and $\mathbf{\hat{\Sigma}}_{1}^{-1}$, conditions 
\textbf{(3)}-\textbf{(4)} and condition (\textbf{i}) of the theorem along
with the fact that $(\tau _{1}-M^{(\tau _{1})})/\tau _{1}\overset{P}{%
\rightarrow }1-n/N$ and $M_{l}^{(\tau _{1})}/\tau _{1}\overset{P}{%
\rightarrow }1/N$, it follows that $\mathbf{\hat{\Sigma}}_{1}^{-1}\overset{P}%
{\rightarrow }\mathbf{\Sigma }_{1}^{-1}$.

We will show that $\mathbf{Z}^{(\tau _{1})}\overset{D}{\rightarrow }\mathbf{Z%
}\sim N_{q_{1}+1}(\mathbf{0},\mathbf{\Sigma }_{1}^{-1})$ as $\tau
_{1}\rightarrow \infty $. To do this, we will associate with each element $%
t\in U_{1}$, $t=1,\ldots ,\tau _{1}$, a random vector $\mathbf{V}%
_{t}^{(1)}=[V_{t,1}^{(1)},\ldots ,V_{t,q_{1}+1}^{(1)}]^{\prime }$ such that

\begin{description}
\item[(a)] $V_{t,1}^{(1)}=1$ and $V_{t,j+1}^{(1)}=[\pi _{\mathbf{x}}^{(1)}(%
\boldsymbol{\theta }_{1}^{\ast })]^{-1}\partial \pi _{\mathbf{x}}^{(1)}(\boldsymbol{%
\theta }_{1}^{\ast })/\partial \theta _{j}^{(1)}$, $j=1,\ldots ,q_{1}$, if $%
t\in $ $U_{1}-S_{0}$ and its associated vector $X_{t}^{(1)}$ of link-indicator 
variables equals the vector $\mathbf{x}\in \Omega -\{\mathbf{0}\}$;

\item[(b)] $V_{t,1}^{(1)}=-\left[ 1 -(1 - n/N)\pi _{\mathbf{0}}^{(1)}(\boldsymbol{%
\theta }_{1}^{\ast })\right] / \left[ (1 - n/N)\pi _{\mathbf{0}}^{(1)}(%
\boldsymbol{\theta }_{1}^{\ast })\right] $ and $V_{t,j+1}^{(1)}=[\pi _{\mathbf{0}%
}^{(1)}(\boldsymbol{\theta }_{1}^{\ast })]^{-1}$ $\times\partial \pi _{\mathbf{0}}^{(1)}(%
\boldsymbol{\theta }_{1}^{\ast })/\partial \theta _{j}^{(1)}$, $j=1,\ldots
,q_{1}$, if $t\in $ $U_{1}-S_{0}$ and its associated vector $X_{t}^{(1)}$ of
link-indicator variables equals the vector $\mathbf{0}\in \Omega $, and

\item[(c)] $V_{t,1}^{(1)}=1$ and $V_{t,j+1}^{(1)}=[\pi _{\mathbf{x}%
}^{(A_{l})}(\boldsymbol{\theta }_{1}^{\ast })]^{-1}\partial \pi _{\mathbf{x}%
}^{(A_{l})}(\boldsymbol{\theta }_{1}^{\ast })/\partial \theta _{j}^{(1)}$, $%
j=1,\ldots ,q_{1}$, if $t\in $ $A_{l}\in S_{A}$ and its associated vector $%
X_{t}^{(1)}$ of link-indicator variables equals the vector $\mathbf{x}\in
\Omega _{-l}$.
\end{description}

Since%
{\setlength\arraycolsep{1pt} 
\begin{eqnarray}
\tau _{1}^{-1/2}\sum_{t=1}^{\tau _{1}}V_{t1}^{(1)}&=&\tau _{1}^{-1/2}\left[
\left( M^{(\tau _{1})}+R_{1}^{(\tau _{1})}\right) -\left( \tau _{1}-M^{(\tau
_{1})}-R_{1}^{(\tau _{1})}\right) \frac{1-(1-n/N)\pi _{\mathbf{0}}^{(1)}(%
\boldsymbol{\theta }_{1}^{\ast })}{(1-n/N)\pi _{\mathbf{0}}^{(1)}(\boldsymbol{\theta 
}_{1}^{\ast })}\right] 
\nonumber \\
&=&Z_{1}^{(\tau _{1})},
\nonumber
\end{eqnarray}}%
and
{\setlength\arraycolsep{1pt} 
\begin{eqnarray*}
\tau _{1}^{-1/2}\sum_{t=1}^{\tau _{1}}V_{t,j+1}^{(1)} &=&\tau _{1}^{-1/2} 
\left[ \sum_{\mathbf{x}\in \Omega }\frac{R_{\mathbf{x}}^{(\tau _{1})}}{\pi _{%
\mathbf{x}}^{(1)}(\boldsymbol{\theta }_{1}^{\ast })}\frac{\partial \pi _{\mathbf{%
x}}^{(1)}(\boldsymbol{\theta }_{1}^{\ast })}{\partial \theta _{j}^{(1)}}%
+\sum\limits_{l=1}^{n}\sum\limits_{\mathbf{x}\in \Omega _{-l}}\frac{R_{%
\mathbf{x}}^{(A_{l},\tau _{1})}}{\pi _{\mathbf{x}}^{(A_{l})}(\boldsymbol{\theta }%
_{1}^{\ast })}\frac{\partial \pi _{\mathbf{x}}^{(A_{l})}(\boldsymbol{\theta }%
_{1}^{\ast })}{\partial \theta _{j}^{(1)}}\right]\!\! =Z_{j+1}^{(\tau _{1})} \\
& &\hspace{9.5cm} j =1,\ldots ,q_{1};
\end{eqnarray*}}%
it follows that $\mathbf{Z}^{(\tau _{1})}=\tau _{1}^{-1/2}\sum_{t=1}^{\tau
_{1}}\mathbf{V}_{t}^{(1)}$.

From the definition of $V_{t,j}^{(1)}$ we have that%
\begin{equation*}
\Pr \left\{ V_{t,1}^{(1)}=1\right\} =(1-n/N)\left[ 1-\pi _{\mathbf{0}}^{(1)}(%
\boldsymbol{\theta }_{1}^{\ast })\right] +n/N,
\end{equation*}%
\begin{equation*}
\Pr \left\{ V_{t,1}^{(1)}=-\left[ 1-(1-n/N)\pi _{\mathbf{0}}^{(1)}(\boldsymbol{%
\theta }_{1}^{\ast })\right] /\left[ (1-n/N)\pi _{\mathbf{0}}^{(1)}(\boldsymbol{%
\theta }_{1}^{\ast })\right] \right\} =(1-n/N)\pi _{\mathbf{0}}^{(1)}(%
\boldsymbol{\theta }_{1}^{\ast }),
\end{equation*}%
\begin{equation*}
\Pr \left\{ V_{t,j+1}^{(1)}=[\pi _{\mathbf{x}}^{(1)}(\boldsymbol{\theta }%
_{1}^{\ast })]^{-1}\partial \pi _{\mathbf{x}}^{(1)}(\boldsymbol{\theta }%
_{1}^{\ast })/\partial \theta _{j}^{(1)}\right\} =(1-n/N)\pi _{\mathbf{x}%
}^{(1)}(\boldsymbol{\theta }_{1}^{\ast }),\text{ }\mathbf{x}\in \Omega \text{, }%
j=1,\ldots ,q_{1},\text{ and}
\end{equation*}%
{\setlength\arraycolsep{1pt} 
\begin{eqnarray}
\Pr \left\{ V_{t,j+1}^{(1)}=[\pi _{\mathbf{x}}^{(A_{l})}(\boldsymbol{\theta }%
_{1}^{\ast })]^{-1}\partial \pi _{\mathbf{x}}^{(A_{l})}(\boldsymbol{\theta }%
_{1}^{\ast })/\partial \theta _{j}^{(1)}\right\}& =&(1/N)\pi _{\mathbf{x}%
}^{(A_{l})}(\boldsymbol{\theta }_{1}^{\ast })\text{, }\mathbf{x}\in \Omega _{-l}%
\text{, }j=1,\ldots,q_{1},
\nonumber \\
& & \hspace{4.6cm}l=1,\ldots,n;
\nonumber
\end{eqnarray}}%
therefore, the expected values of the variables $V_{t,j}^{(1)}$ are%
\begin{equation*}
E\left( V_{t,1}^{(1)}\right) =(1-n/N)\left[ 1-\pi _{\mathbf{0}}^{(1)}(%
\boldsymbol{\theta }_{1}^{\ast })\right] +n/N-\left[ 1-(1-n/N)\pi _{\mathbf{0}%
}^{(1)}(\boldsymbol{\theta }_{1}^{\ast })\right] =0
\end{equation*}%
and
\begin{eqnarray*}
E\left( V_{t,j+1}^{(1)}\right)&=&\sum_{\mathbf{x}\in \Omega }\partial \pi _{%
\mathbf{x}}^{(1)}(\boldsymbol{\theta }_{1}^{\ast })/\partial \theta
_{j}^{(1)}(1-n/N)+\sum\limits_{l=1}^{n}\sum\limits_{\mathbf{x}\in \Omega
_{-l}}\partial \pi _{\mathbf{x}}^{(A_{l})}(\boldsymbol{\theta }_{1}^{\ast
})/\partial \theta _{j}^{(1)}(1/N)=0,  \\
& &\hspace{50ex} j=1,\ldots,q_{1},
\end{eqnarray*}%
because of (\ref{sum-deriv0}). Thus, $E\left( \mathbf{V}_{t}^{(1)}\right) =%
\mathbf{0},$ $t=1,\ldots ,\tau _{1}$. Furthermore, their variances are%
\begin{eqnarray*}
V\left( V_{t,1}^{(1)}\right)&=&(1-n/N)\!\left[ 1-\pi _{\mathbf{0}}^{(1)}(%
\boldsymbol{\theta }_{1}^{\ast })\right] +n/N+\frac{\left[ 1-(1-n/N)\pi _{%
\mathbf{0}}^{(1)}(\boldsymbol{\theta }_{1}^{\ast })\right] ^{2}}{(1-n/N)\pi _{%
\mathbf{0}}^{(1)}(\boldsymbol{\theta }_{1}^{\ast })}  \\
&=& \frac{1-(1-n/N)\pi _{\mathbf{0}}^{(1)}(\boldsymbol{\theta }_{1}^{\ast })}
{(1-n/N)\pi _{\mathbf{0}}^{(1)}(\boldsymbol{\theta }_{1}^{\ast })}
\end{eqnarray*}
and%
{\setlength\arraycolsep{0pt}
\begin{eqnarray*}
V\left( V_{t,j+1}^{(1)}\right)&=&(1-n/N)\sum_{\mathbf{x}\in \Omega }\frac{1%
}{\pi _{\mathbf{x}}^{(1)}(\boldsymbol{\theta }_{1}^{\ast })}\left[ \frac{%
\partial \pi _{\mathbf{x}}^{(1)}(\boldsymbol{\theta }_{1}^{\ast })}{\partial
\theta _{j}^{(1)}}\right] ^{2}   \\
& &+\frac{1}{N}\sum\limits_{l=1}^{n}\sum\limits_{\mathbf{x}\in \Omega_{-l}}
\frac{1}{\pi _{\mathbf{x}}^{(A_{l})}(\boldsymbol{\theta }_{1}^{\ast })}
\left[ \frac{\partial \pi _{\mathbf{x}}^{(A_{l})}(\boldsymbol{\theta}_{1}^{\ast })}
{\partial \theta _{j}^{(1)}}\right]^{2}, \, j=1,\ldots ,q_{1},
\end{eqnarray*}}%
and their covariances are%
{\setlength\arraycolsep{0pt}
\begin{eqnarray*}
Cov\left( V_{t,1}^{(1)},V_{t,j+1}^{(1)}\right) &=&\sum_{\mathbf{x}\in \Omega
-\{\mathbf{0}\}}\frac{\partial \pi _{\mathbf{x}}^{(1)}(\boldsymbol{\theta }%
_{1}^{\ast })}{\partial \theta _{j}^{(1)}}(1-n/N)-\frac{1-(1-n/N)\pi _{%
\mathbf{0}}^{(1)}(\boldsymbol{\theta }_{1}^{\ast })}{(1-n/N)\pi _{\mathbf{0}%
}^{(1)}(\boldsymbol{\theta }_{1}^{\ast })}\frac{\partial \pi _{\mathbf{0}}^{(1)}(%
\boldsymbol{\theta }_{1}^{\ast })}{\partial \theta _{j}^{(1)}} \\
&&\times(1-n/N)+\sum\limits_{l=1}^{n}\sum\limits_{\mathbf{x}\in \Omega _{-l}}\frac{%
\partial \pi _{\mathbf{x}}^{(A_{l})}(\boldsymbol{\theta }_{1}^{\ast })}{\partial
\theta _{j}^{(1)}}\frac{1}{N} \\
&=&-\frac{1}{\pi _{\mathbf{0}}^{(1)}(\boldsymbol{\theta }_{1}^{\ast })}\frac{%
\partial \pi _{\mathbf{0}}^{(1)}(\boldsymbol{\theta }_{1}^{\ast })}{\partial
\theta _{j}^{(1)}}\text{, }j=1,\ldots ,q_{1},\text{ \ and} \\
Cov\left( V_{t,j+1}^{(1)},V_{t,j^{\prime }+1}^{(1)}\right) &=&(1-n/N)\sum_{%
\mathbf{x}\in \Omega }\frac{1}{\pi _{\mathbf{x}}^{(1)}(\boldsymbol{\theta }%
_{1}^{\ast })}\frac{\partial \pi _{\mathbf{x}}^{(1)}(\boldsymbol{\theta }%
_{1}^{\ast })}{\partial \theta _{j}^{(1)}}\frac{\partial \pi _{\mathbf{x}%
}^{(1)}(\boldsymbol{\theta }_{1}^{\ast })}{\partial \theta _{j^{\prime }}^{(1)}}
\\
&&+\frac{1}{N}\sum\limits_{l=1}^{n}\sum\limits_{\mathbf{x}\in \Omega _{-l}}%
\frac{1}{\pi _{\mathbf{x}}^{(A_{l})}(\boldsymbol{\theta }_{1}^{\ast })}%
\frac{\partial \pi _{\mathbf{x}}^{(A_{l})}(\boldsymbol{\theta }_{1}^{\ast })}{%
\partial \theta _{j}^{(1)}}\frac{\partial \pi _{\mathbf{x}}^{(A_{l})}(%
\boldsymbol{\theta}_{1}^{\ast})}{\partial \theta _{j^{\prime}}^{(1)}},  \\
&& \hspace{20ex}j,j^{\prime}=1,\ldots ,q_{1},\, j\neq j^{\prime}.
\end{eqnarray*}}%
Therefore, the variance-covariance matrix of $\mathbf{V}_{t}^{(1)}$ is $%
\mathbf{\Sigma }_{1}^{-1}$.

Finally, since the $\mathbf{V}_{t}^{(1)}$, $t=1,\ldots ,\tau _{1}$, are
independent and identically distributed random vectors, by the central limit
theorem it follows that%
\begin{equation*}
\mathbf{Z}^{(\tau _{1})}=\tau _{1}^{-1/2}\sum\nolimits_{t=1}^{\tau _{1}}%
\mathbf{V}_{t}^{(1)}\overset{D}{\rightarrow }\mathbf{Z}\sim N_{q_{1}+1}(%
\mathbf{0},\mathbf{\Sigma }_{1}^{-1}).
\end{equation*}%
Consequently by (\ref{difsigmauz}), 
\begin{equation*}
\mathbf{W}_{1}^{(\tau _{1})}=\left[ \tau _{1}^{-1/2}\left( \hat{\tau}%
_{1}^{(\tau _{1})}-\tau _{1}\right) ,\tau _{1}^{1/2}\left( \boldsymbol{\hat{%
\theta}}_{1}^{(\tau _{1})}-\boldsymbol{\theta }_{1}^{\ast }\right) \right] 
\overset{D}{\rightarrow }\mathbf{\Sigma }_{1}\mathbf{Z}\sim N_{q_{1}+1}(%
\mathbf{0},\mathbf{\Sigma }_{1})
\end{equation*}%
as $\mathbf{\hat{\Sigma}}_{1}\overset{P}{\rightarrow }\mathbf{\Sigma }_{1}$.
\end{proof}

\subsection{Asymptotic multivariate normal distribution of estimators of 
$\boldsymbol{\protect\theta}_{1}^{\ast}$}

\begin{theorem}
Let $\boldsymbol{\theta }_{1}^{\ast }=(\theta _{1}^{\ast },\ldots ,\theta
_{q_{1}}^{\ast })$ be the true value of $\boldsymbol{\theta }_{1}$. Let $\boldsymbol{%
\hat{\theta}}_{1}^{(\tau _{1})}=(\boldsymbol{\hat{\theta}}_{11}^{(\tau
_{1})},\ldots ,\boldsymbol{\hat{\theta}}_{1q_{1}}^{(\tau _{1})})$ be an
estimator of $\boldsymbol{\theta }_{1}^{\ast }$, such that

\begin{description}
\item[(i)] $\boldsymbol{\hat{\theta}}_{1}^{(\tau _{1})}\overset{P}{\rightarrow }%
\boldsymbol{\theta }_{1}^{\ast }.$

\item[(ii)] $\tau _{1}^{-1/2}\left\{ \frac{\partial }{\partial \theta
_{j}^{(1)}}\ln \left[ L_{11}^{(\tau _{1})}(\boldsymbol{\hat{\theta}}_{1}^{(\tau
_{1})})L_{0}^{(\tau _{1})}(\boldsymbol{\hat{\theta}}_{1}^{(\tau _{1})})\right]
\right\} \overset{P}{\rightarrow }0$, $j=1,\ldots ,q_{1}.$
\end{description}

In addition, let $\mathbf{\Psi }_{1}^{-1}$ be the $q_{1}\times q_{1}$ matrix
whose elements are
{\setlength\arraycolsep{1pt} 
\begin{eqnarray*}
\left[ \mathbf{\Psi }_{1}^{-1}\right]_{i,j}&=&\left[ \mathbf{\Psi }%
_{1}^{-1}\right]_{j,i}=(1-n/N)[1-\pi _{\mathbf{0}}^{(1)}(\boldsymbol{\theta }%
_{1}^{\ast })]  \\
&&\times\sum\limits_{\mathbf{x}\in \Omega -\{\mathbf{0}\}}\left[ 1/%
\tilde{\pi}_{\mathbf{x}}^{(1)}(\boldsymbol{\theta }_{1}^{\ast })\right] \left[
\partial \tilde{\pi}_{\mathbf{x}}^{(1)}(\boldsymbol{\theta }_{1}^{\ast
})/\partial \theta _{i}^{(1)}\right] \left[ \partial \tilde{\pi}_{\mathbf{x}%
}^{(1)}(\boldsymbol{\theta }_{1}^{\ast })/\partial \theta _{j}^{(1)}\right] \\
&&+\frac{1}{N}\sum\nolimits_{l=1}^{n}\sum\nolimits_{\mathbf{x}\in \Omega
_{-l}}\left[ 1/\pi _{\mathbf{x}}^{(A_{l})}(\boldsymbol{\theta }_{1}^{\ast })%
\right] \left[ \partial \pi _{\mathbf{x}}^{(A_{l})}(\boldsymbol{\theta }%
_{1}^{\ast })/\partial \theta _{i}^{(1)}\right] \left[ \partial \pi _{%
\mathbf{x}}^{(A_{l})}(\boldsymbol{\theta }_{1}^{\ast })/\partial \theta
_{j}^{(1)}\right], \\
& & \hspace{9cm} i,j=1,\ldots ,q_{1},
\end{eqnarray*}}%
where $\tilde{\pi}_{\mathbf{x}}^{(1)}(\boldsymbol{\theta }_{1}^{\ast })=\pi _{%
\mathbf{x}}^{(1)}(\boldsymbol{\theta }_{1}^{\ast })/[1-\pi _{\mathbf{0}}^{(1)}(%
\boldsymbol{\theta }_{1}^{\ast })],$ $\mathbf{x}\in \Omega -\{\mathbf{0}\}$, and
suppose that $\mathbf{\Psi }_{1}^{-1}$ is a non-singular matrix. Then%
\begin{equation*}
\tau _{1}^{1/2}\left[ \boldsymbol{\hat{\theta}}_{1}^{(\tau _{1})}-\boldsymbol{\theta 
}_{1}^{\ast }\right] \overset{D}{\rightarrow }N_{q_{1}}\left( \mathbf{0},%
\mathbf{\Psi }_{1}\right) ,
\end{equation*}%
where $\mathbf{\Psi }_{1}$ is the inverse of $\ \mathbf{\Psi }_{1}^{-1}$ and 
$\mathbf{0}=(0,\ldots ,0)\in\mathbb{R}^{q_{1}}$.

Furthermore, if \ $\hat{\tau}_{1}^{(\tau _{1})}$ is an estimator of $\tau
_{1}$ such that

\begin{description}
\item[(iii)] $\tau _{1}^{-1/2}\left\{ \hat{\tau}_{1}^{(\tau _{1})}-\left(
M^{(\tau _{1})}+R_{1}^{(\tau _{1})}\right) /\left[ 1-(1-n/N)\pi _{\mathbf{0}%
}^{(1)}\left( \boldsymbol{\hat{\theta}}_{1}^{(\tau _{1})}\right) \right]
\right\} \overset{P}{\rightarrow }0,$
\end{description}

then%
\begin{equation*}
\tau _{1}^{-1/2}\left( \hat{\tau}_{1}^{(\tau _{1})}-\tau _{1}\right) \overset%
{D}{\rightarrow }N(0,\sigma _{1}^{2}),
\end{equation*}%
where%
\begin{equation}
\sigma _{1}^{2}=\frac{1-n/N}{1-(1-n/N)\pi _{\mathbf{0}}^{(1)}(\boldsymbol{\theta 
}_{1}^{\ast })}\left\{ \pi _{\mathbf{0}}^{(1)}(\boldsymbol{\theta }_{1}^{\ast })+%
\frac{(1-n/N)\left[ \nabla \pi _{\mathbf{0}}^{(1)}\left( \boldsymbol{\theta }%
_{1}^{\ast }\right) \right] ^{\prime }\mathbf{\Psi }_{1}\left[ \nabla \pi _{%
\mathbf{0}}^{(1)}\left( \boldsymbol{\theta }_{1}^{\ast }\right) \right] }{%
1-(1-n/N)\pi _{\mathbf{0}}^{(1)}(\boldsymbol{\theta }_{1}^{\ast })}\right\} ,
\label{sigma12}
\end{equation}%
and $\nabla \pi _{\mathbf{0}}^{(1)}\left( \boldsymbol{\theta }_{1}^{\ast
}\right) =\left[ \partial \pi _{\mathbf{0}}^{(1)}\left( \boldsymbol{\theta }%
_{1}^{\ast }\right) /\partial \theta _{1}^{(1)},\ldots ,\partial \pi _{%
\mathbf{0}}^{(1)}\left( \boldsymbol{\theta }_{1}^{\ast }\right) /\partial \theta
_{q_{1}}^{(1)}\right] ^{\prime }$ is the gradient of $\pi _{\mathbf{0}%
}^{(1)}(\boldsymbol{\theta }_{1})$ evaluated at $\boldsymbol{\theta }_{1}^{\ast }$.
\end{theorem}

\begin{proof}
From the definitions of $L_{11}^{(\tau _{1})}(\boldsymbol{\theta}_{1})$ and 
$L_{0}^{(\tau _{1})}(\boldsymbol{\theta}_{1})$ we have that
\begin{eqnarray}
\frac{\partial }{\partial \theta _{j}^{(1)}}\ln \left[ L_{11}^{(\tau _{1})}(%
\boldsymbol{\hat{\theta}}_{1}^{(\tau _{1})})L_{0}^{(\tau _{1})}(\boldsymbol{\hat{%
\theta}}_{1}^{(\tau _{1})})\right] &=&\sum\limits_{\mathbf{x}\in \Omega -%
\mathbf{0}}\frac{R_{\mathbf{x}}^{(\tau _{1})}}{\tilde{\pi}_{\mathbf{x}%
}^{(1)}(\boldsymbol{\hat{\theta}}_{1}^{(\tau _{1})})}\frac{\partial \tilde{\pi}_{%
\mathbf{x}}^{(1)}(\boldsymbol{\hat{\theta}}_{1}^{(\tau _{1})})}{\partial \theta
_{j}^{(1)}}  \notag \\
&&+\sum\limits_{l=1}^{n}\sum\limits_{\mathbf{x}\in \Omega _{-l}}\frac{R_{%
\mathbf{x}}^{(A_{l},\tau _{1})}}{\pi _{\mathbf{x}}^{(A_{l})}(\boldsymbol{\hat{%
\theta}}_{1}^{(\tau _{1})})}\frac{\partial \pi _{\mathbf{x}}^{(A_{l})}(%
\boldsymbol{\hat{\theta}}_{1}^{(\tau _{1})})}{\partial \theta _{j}^{(1)}}.
\label{deriv-lnL110}
\end{eqnarray}

Since%
\begin{equation}
\sum\limits_{\mathbf{x}\in \Omega -\{\mathbf{0}\}}\partial \tilde{\pi}_{%
\mathbf{x}}^{(1)}(\boldsymbol{\hat{\theta}}_{1}^{(\tau _{1})})/\partial \theta
_{j}^{(1)}=0\text{ \ and \ }\sum\limits_{\mathbf{x}\in \Omega
_{-l}}\partial \pi _{\mathbf{x}}^{(A_{l})}(\boldsymbol{\hat{\theta}}_{1}^{(\tau
_{1})})/\partial \theta _{j}^{(1)}=0,  \label{sum-deriv0c}
\end{equation}

from (\ref{deriv-lnL110}) we get that%
\begin{eqnarray}
&&\tau _{1}^{-1/2}\left\{ \sum\limits_{\mathbf{x}\in \Omega -\{\mathbf{0}\}}%
\frac{R_{\mathbf{x}}^{(\tau _{1})}-R_{1}^{(\tau _{1})}\tilde{\pi}_{\mathbf{x}%
}^{(1)}(\boldsymbol{\theta }_{1}^{\ast })}{\tilde{\pi}_{\mathbf{x}}^{(1)}(%
\boldsymbol{\hat{\theta}}_{1}^{(\tau _{1})})}\frac{\partial \tilde{\pi}_{\mathbf{%
x}}^{(1)}(\boldsymbol{\hat{\theta}}_{1}^{(\tau _{1})})}{\partial \theta
_{j}^{(1)}}\right.  \notag \\
&&+\left. \sum\limits_{l=1}^{n}\sum\limits_{\mathbf{x}\in \Omega _{-l}}%
\frac{R_{\mathbf{x}}^{(A_{l},\tau _{1})}-M_{l}^{(\tau _{1})}\pi _{\mathbf{x}%
}^{(A_{l})}(\boldsymbol{\theta }_{1}^{\ast })}{\pi _{\mathbf{x}}^{(A_{l})}(%
\boldsymbol{\hat{\theta}}_{1}^{(\tau _{1})})}\frac{\partial \pi _{\mathbf{x}%
}^{(A_{l})}(\boldsymbol{\hat{\theta}}_{1}^{(\tau _{1})})}{\partial \theta
_{j}^{(1)}}\right\} \notag \\
& &-\frac{\partial }{\partial \theta _{j}^{(1)}}\ln \left[
L_{11}^{(\tau _{1})}(\boldsymbol{\hat{\theta}}_{1}^{(\tau _{1})})L_{0}^{(\tau
_{1})}(\boldsymbol{\hat{\theta}}_{1}^{(\tau _{1})})\right]  \notag \\
&=&\tau _{1}^{1/2}\left\{ \frac{R_{1}^{(\tau _{1})}}{\tau _{1}}\sum\limits_{%
\mathbf{x}\in \Omega -\{\mathbf{0}\}}\frac{\tilde{\pi}_{\mathbf{x}}^{(1)}(%
\boldsymbol{\hat{\theta}}_{1}^{(\tau _{1})})-\tilde{\pi}_{\mathbf{x}}^{(1)}(%
\boldsymbol{\theta }_{1}^{\ast })}{\tilde{\pi}_{\mathbf{x}}^{(1)}(\boldsymbol{\hat{%
\theta}}_{1}^{(\tau _{1})})}\frac{\partial \tilde{\pi}_{\mathbf{x}}^{(1)}(%
\boldsymbol{\hat{\theta}}_{1}^{(\tau _{1})})}{\partial \theta _{j}^{(1)}}\right.
\notag \\
&&+\left. \sum\limits_{l=1}^{n}\frac{M_{l}^{(\tau _{1})}}{\tau _{1}}%
\sum\limits_{\mathbf{x}\in \Omega _{-l}}\frac{\pi _{\mathbf{x}}^{(A_{l})}(%
\boldsymbol{\hat{\theta}}_{1}^{(\tau _{1})})-\pi _{\mathbf{x}}^{(A_{l})}(\boldsymbol{%
\theta }_{1}^{\ast })}{\pi _{\mathbf{x}}^{(A_{l})}(\boldsymbol{\hat{\theta}}%
_{1}^{(\tau _{1})})}\frac{\partial \pi _{\mathbf{x}}^{(A_{l})}(\boldsymbol{\hat{%
\theta}}_{1}^{(\tau _{1})})}{\partial \theta _{j}^{(1)}}\right\} .
\label{expre-medc}
\end{eqnarray}

Let $Y_{\mathbf{x}}^{(\tau _{1})}=R_{\mathbf{x}}^{(\tau _{1})}-R_{1}^{(\tau
_{1})}\tilde{\pi}_{\mathbf{x}}^{(1)}(\boldsymbol{\theta }_{1}^{\ast })$, $Y_{%
\mathbf{x}}^{(A_{l},\tau _{1})}=R_{\mathbf{x}}^{(A_{l},\tau
_{1})}-M_{l}^{(\tau _{1})}\pi _{\mathbf{x}}^{(A_{l})}(\boldsymbol{\theta }%
_{1}^{\ast })$ and%
\begin{eqnarray*}
Z_{j}^{(\tau _{1})} &=&\tau _{1}^{-1/2}\left[ \sum\limits_{\mathbf{x}\in
\Omega -\{\mathbf{0}\}}\frac{R_{\mathbf{x}}^{(\tau _{1})}}{\tilde{\pi}_{%
\mathbf{x}}^{(1)}(\boldsymbol{\theta }_{1}^{\ast })}\frac{\partial \tilde{\pi}_{%
\mathbf{x}}^{(1)}(\boldsymbol{\theta }_{1}^{\ast })}{\partial \theta _{j}^{(1)}}%
+\sum\limits_{l=1}^{n}\sum\limits_{\mathbf{x}\in \Omega _{-l}}\frac{R_{%
\mathbf{x}}^{(A_{l},\tau _{1})}}{\pi _{\mathbf{x}}^{(A_{l})}(\boldsymbol{\theta }%
_{1}^{\ast })}\frac{\partial \pi _{\mathbf{x}}^{(A_{l})}(\boldsymbol{\theta }%
_{1}^{\ast })}{\partial \theta _{j}^{(1)}}\right] \\
&=&\tau _{1}^{-1/2}\left[ \sum\limits_{\mathbf{x}\in \Omega -\{\mathbf{0}\}}%
\frac{Y_{\mathbf{x}}^{(\tau _{1})}}{\tilde{\pi}_{\mathbf{x}}^{(1)}(\boldsymbol{%
\theta }_{1}^{\ast })}\frac{\partial \tilde{\pi}_{\mathbf{x}}^{(1)}(\boldsymbol{%
\theta }_{1}^{\ast })}{\partial \theta _{j}^{(1)}}+\sum\limits_{l=1}^{n}%
\sum\limits_{\mathbf{x}\in \Omega _{-l}}\frac{Y_{\mathbf{x}}^{(A_{l},\tau
_{1})}}{\pi _{\mathbf{x}}^{(A_{l})}(\boldsymbol{\theta }_{1}^{\ast })}\frac{%
\partial \pi _{\mathbf{x}}^{(A_{l})}(\boldsymbol{\theta }_{1}^{\ast })}{\partial
\theta _{j}^{(1)}}\right] ,
\end{eqnarray*}%
where the last equality is obtained using (\ref{sum-deriv0c}) but replacing 
$\boldsymbol{\hat{\theta}}_{1}^{(\tau _{1})}$ by $\boldsymbol{\theta }_{1}^{\ast }$.
Then, the difference between the left-hand side of (\ref{expre-medc}) and $%
Z_{j}^{(\tau _{1})}$ is given by%
\begin{eqnarray}
&&\tau _{1}^{-1/2}\left\{ \sum\limits_{\mathbf{x}\in \Omega -\{\mathbf{0}\}}%
\frac{Y_{\mathbf{x}}^{(\tau _{1})}}{\tilde{\pi}_{\mathbf{x}}^{(1)}(\boldsymbol{%
\hat{\theta}}_{1}^{(\tau _{1})})}\frac{\partial \tilde{\pi}_{\mathbf{x}%
}^{(1)}(\boldsymbol{\hat{\theta}}_{1}^{(\tau _{1})})}{\partial \theta _{j}^{(1)}}%
+\sum\limits_{l=1}^{n}\sum\limits_{\mathbf{x}\in \Omega _{-l}}\frac{Y_{%
\mathbf{x}}^{(A_{l},\tau _{1})}}{\pi _{\mathbf{x}}^{(A_{l})}(\boldsymbol{\hat{%
\theta}}_{1}^{(\tau _{1})})}\frac{\partial \pi _{\mathbf{x}}^{(A_{l})}(%
\boldsymbol{\hat{\theta}}_{1}^{(\tau _{1})})}{\partial \theta _{j}^{(1)}}\right\}
\notag \\
&&-\tau _{1}^{-1/2}\frac{\partial }{\partial \theta _{j}^{(1)}}\ln \left[
L_{11}^{(\tau _{1})}(\boldsymbol{\hat{\theta}}_{1}^{(\tau _{1})})L_{0}^{(\tau
_{1})}(\boldsymbol{\hat{\theta}}_{1}^{(\tau _{1})})\right] -Z_{j}^{(\tau _{1})} 
\notag \\
&=&\tau _{1}^{-1/2}\left\{ \sum\limits_{\mathbf{x}\in \Omega -\{\mathbf{0}%
\}}Y_{\mathbf{x}}^{(\tau _{1})}\left[ \frac{1}{\tilde{\pi}_{\mathbf{x}%
}^{(1)}(\boldsymbol{\hat{\theta}}_{1}^{(\tau _{1})})}\frac{\partial \tilde{\pi}_{%
\mathbf{x}}^{(1)}(\boldsymbol{\hat{\theta}}_{1}^{(\tau _{1})})}{\partial \theta
_{j}^{(1)}}-\frac{1}{\tilde{\pi}_{\mathbf{x}}^{(1)}(\boldsymbol{\theta }%
_{1}^{\ast })}\frac{\partial \tilde{\pi}_{\mathbf{x}}^{(1)}(\boldsymbol{\theta }%
_{1}^{\ast })}{\partial \theta _{j}^{(1)}}\right] \right.  \notag \\
&&+\left. \sum\limits_{l=1}^{n}\sum\limits_{\mathbf{x}\in \Omega _{-l}}Y_{%
\mathbf{x}}^{(A_{l},\tau _{1})}\left[ \frac{1}{\pi _{\mathbf{x}}^{(A_{l})}(%
\boldsymbol{\hat{\theta}}_{1}^{(\tau _{1})})}\frac{\partial \pi _{\mathbf{x}%
}^{(A_{l})}(\boldsymbol{\hat{\theta}}_{1}^{(\tau _{1})})}{\partial \theta
_{j}^{(1)}}-\frac{1}{\pi _{\mathbf{x}}^{(A_{l})}(\boldsymbol{\theta }_{1}^{\ast
})}\frac{\partial \pi _{\mathbf{x}}^{(A_{l})}(\boldsymbol{\theta }_{1}^{\ast })}{%
\partial \theta _{j}^{(1)}}\right] \right\}  \notag \\
&&-\tau _{1}^{-1/2}\frac{\partial }{\partial \theta _{j}^{(1)}}\ln \left[
L_{11}^{(\tau _{1})}(\boldsymbol{\hat{\theta}}_{1}^{(\tau _{1})})L_{0}^{(\tau
_{1})}(\boldsymbol{\hat{\theta}}_{1}^{(\tau _{1})})\right] .
\label{left-hand-side-zj}
\end{eqnarray}%
Since $\tau _{1}^{-1/2}Y_{\mathbf{x}}^{(\tau _{1})}=O_{p}(1)$ and $\tau
_{1}^{-1/2}Y_{\mathbf{x}}^{(A_{l},\tau _{1})}=O_{p}(1)$, these results along
with conditions \textbf{(3)-(4)} and conditions (\textbf{i}) and (\textbf{ii}%
) of the theorem imply that (\ref{left-hand-side-zj}) converges to zero in
probability.

On the other hand, by the mean value theorem of several variables we have
that%
\begin{eqnarray}
\tilde{\pi}_{\mathbf{x}}^{(1)}(\boldsymbol{\hat{\theta}}_{1}^{(\tau _{1})})-%
\tilde{\pi}_{\mathbf{x}}^{(1)}(\boldsymbol{\theta }_{1}^{\ast })
&=&\sum\nolimits_{i=1}^{q_{1}}\left( \hat{\theta}_{1i}^{(\tau _{1})}-\theta
_{1i}^{\ast }\right) \partial \tilde{\pi}_{\mathbf{x}}^{(1)}(\boldsymbol{\theta }%
_{1\mathbf{x}}^{(\tau _{1})})/\partial \theta _{i}^{(1)}\text{ \ \ and}
\label{dif-pitildex} \\
\pi _{\mathbf{x}}^{(A_{l})}(\boldsymbol{\hat{\theta}}_{1}^{(\tau _{1})})-\pi _{%
\mathbf{x}}^{(A_{l})}(\boldsymbol{\theta }_{1}^{\ast })
&=&\sum\nolimits_{i=1}^{q_{1}}\left( \hat{\theta}_{1i}^{(\tau _{1})}-\theta
_{1i}^{\ast }\right) \partial \pi _{\mathbf{x}}^{(A_{l})}(\boldsymbol{\theta }%
_{A_{l}\mathbf{x}}^{(\tau _{1})})/\partial \theta _{i}^{(1)},  \notag
\end{eqnarray}%
where $\boldsymbol{\theta }_{1\mathbf{x}}^{(\tau _{1})}$ and $\boldsymbol{\theta }%
_{A_{l}\mathbf{x}}^{(\tau _{1})}$ are between $\boldsymbol{\hat{\theta}}%
_{1}^{(\tau _{1})}$ and $\boldsymbol{\theta }_{1}^{\ast }$. Since the difference
between the right-hand side of (\ref{expre-medc}) and $Z_{j}^{(\tau _{1})}$
also converges to zero in probability, we have that%
\begin{eqnarray}
&&\tau _{1}^{1/2}\left\{ \frac{R_{1}^{(\tau _{1})}}{\tau _{1}}\sum\limits_{%
\mathbf{x}\in \Omega -\{\mathbf{0}\}}\frac{1}{\tilde{\pi}_{\mathbf{x}}^{(1)}(%
\boldsymbol{\hat{\theta}}_{1}^{(\tau _{1})})}\frac{\partial \tilde{\pi}_{\mathbf{%
x}}^{(1)}(\boldsymbol{\hat{\theta}}_{1}^{(\tau _{1})})}{\partial \theta
_{j}^{(1)}}\sum_{i=1}^{q_{1}}\left( \hat{\theta}_{1i}^{(\tau _{1})}-\theta
_{1i}^{\ast }\right) \frac{\partial \tilde{\pi}_{\mathbf{x}}^{(1)}(\boldsymbol{%
\theta }_{1\mathbf{x}}^{(\tau _{1})})}{\partial \theta _{i}^{(1)}}\right. 
\notag \\
&&+\left. \sum\limits_{l=1}^{n}\frac{M_{l}^{(\tau _{1})}}{\tau _{1}}%
\sum\limits_{\mathbf{x}\in \Omega _{-l}}\frac{1}{\pi _{\mathbf{x}%
}^{(A_{l})}(\boldsymbol{\hat{\theta}}_{1}^{(\tau _{1})})}\frac{\partial \pi _{%
\mathbf{x}}^{(A_{l})}(\boldsymbol{\hat{\theta}}_{1}^{(\tau _{1})})}{\partial
\theta _{j}^{(1)}}\sum\limits_{i=1}^{q_{1}}\left( \hat{\theta}_{1i}^{(\tau
_{1})}-\theta _{1i}^{\ast }\right) \frac{\partial \pi _{\mathbf{x}%
}^{(A_{l})}(\boldsymbol{\theta }_{A_{l}\mathbf{x}}^{(\tau _{1})})}{\partial
\theta _{i}^{(1)}}\right\} -Z_{j}^{(\tau _{1})}  \notag \\
&=&\sum_{i=1}^{q_{1}}\left[ \mathbf{\hat{\Psi}}_{1}^{-1}\right] _{j,i}\left[
\tau _{1}^{1/2}\left( \hat{\theta}_{1i}^{(\tau _{1})}-\theta _{1i}^{\ast
}\right) \right] -Z_{j}^{(\tau _{1})}\overset{P}{\rightarrow }0,
\label{right-hand-side-zj}
\end{eqnarray}%
where%
\begin{eqnarray}
\left[ \mathbf{\hat{\Psi}}_{1}^{-1}\right] _{j,i} &=&\frac{R_{1}^{(\tau
_{1})}}{\tau _{1}}\sum\limits_{\mathbf{x}\in \Omega -\{\mathbf{0}\}}\frac{1%
}{\tilde{\pi}_{\mathbf{x}}^{(1)}(\boldsymbol{\hat{\theta}}_{1}^{(\tau _{1})})}%
\frac{\partial \tilde{\pi}_{\mathbf{x}}^{(1)}(\boldsymbol{\hat{\theta}}%
_{1}^{(\tau _{1})})}{\partial \theta _{j}^{(1)}}\frac{\partial \tilde{\pi}_{%
\mathbf{x}}^{(1)}(\boldsymbol{\theta }_{1\mathbf{x}}^{(\tau _{1})})}{\partial
\theta _{i}^{(1)}}  \notag \\
&&+\sum\limits_{l=1}^{n}\frac{M_{l}^{(\tau _{1})}}{\tau _{1}}\sum\limits_{%
\mathbf{x}\in \Omega _{-l}}\frac{1}{\pi _{\mathbf{x}}^{(A_{l})}(\boldsymbol{\hat{%
\theta}}_{1}^{(\tau _{1})})}\frac{\partial \pi _{\mathbf{x}}^{(A_{l})}(%
\boldsymbol{\hat{\theta}}_{1}^{(\tau _{1})})}{\partial \theta _{j}^{(1)}}\frac{%
\partial \pi _{\mathbf{x}}^{(A_{l})}(\boldsymbol{\theta }_{A_{l}\mathbf{x}%
}^{(\tau _{1})})}{\partial \theta _{i}^{(1)}}.  \label{hat-psi-matrix}
\end{eqnarray}

Notice that from the definitions of the matrices $\mathbf{\Psi }_{1}^{-1}$
and $\mathbf{\hat{\Psi}}_{1}^{-1}$, conditions \textbf{(3)}-\textbf{(4)} and
condition (\textbf{i}) of the theorem along with the fact that $R_{1}^{(\tau
_{1})}/\tau _{1}\overset{P}{\rightarrow }(1-n/N)[1-\pi _{\mathbf{0}}^{(1)}(%
\boldsymbol{\theta }_{1}^{\ast })]$ and $M_{l}^{(\tau _{1})}/\tau _{1}\overset{P}%
{\rightarrow }1/N$, it follows that $\mathbf{\hat{\Psi}}_{1}^{-1}\overset{P}{%
\rightarrow }\mathbf{\Psi }_{1}^{-1}$.

By condition (\textbf{iii}) of the theorem and using exactly the same
procedure as that used to obtain expression (\ref{expr-z1}) we will get that
expression which we will put in the following terms:%
\begin{equation}
\hat{a}_{1}\left[ \tau _{1}^{-1/2}\left( \hat{\tau}_{1}^{(\tau _{1})}-\tau
_{1}\right) \right] +\sum\nolimits_{i=1}^{q_{1}}\hat{a}_{i+1}\left[ \tau
_{1}^{1/2}\left( \hat{\theta}_{1i}^{(\tau _{1})}-\theta _{1i}^{\ast }\right) %
\right] -Z^{(\tau _{1})}\overset{P}{\rightarrow }0,  \label{difz}
\end{equation}%
where 
\begin{equation*}
\hat{a}_{1}=\frac{1-(1-n/N)\pi _{\mathbf{0}}^{(1)}(\boldsymbol{\hat{\theta}}%
_{1}^{(\tau _{1})})}{(1-n/N)\pi _{\mathbf{0}}^{(1)}(\boldsymbol{\theta }%
_{1}^{\ast })}\text{, \ }\hat{a}_{i+1}=-\frac{1}{\pi _{\mathbf{0}}^{(1)}(%
\boldsymbol{\theta }_{1}^{\ast })}\frac{\partial \pi _{\mathbf{0}}^{(1)}(\boldsymbol{%
\theta }_{1\mathbf{0}}^{(\tau _{1})})}{\partial \theta _{i}^{(1)}}\text{, }%
i=1,\ldots ,q_{1},
\end{equation*}%
\begin{equation}
Z^{(\tau _{1})}=\tau _{1}^{-1/2}\frac{\left( M^{(\tau _{1})}+R_{1}^{(\tau
_{1})}\right) -\tau _{1}\left[ 1-(1-n/N)\pi _{\mathbf{0}}^{(1)}(\boldsymbol{%
\theta }_{1}^{\ast })\right] }{(1-n/N)\pi _{\mathbf{0}}^{(1)}(\boldsymbol{\theta 
}_{1}^{\ast })},  \label{z}
\end{equation}%
and $\boldsymbol{\theta }_{1\mathbf{0}}^{(\tau _{1})}$ is between $\boldsymbol{\hat{%
\theta}}_{1}^{(\tau _{1})}$ and $\boldsymbol{\theta }_{1}^{\ast }$. Notice that
conditions \textbf{(3)}-\textbf{(4)} and condition (\textbf{i}) of the
theorem imply that $\hat{a}_{i}\overset{P}{\rightarrow }a_{i},$ $i=1,\ldots
, $ $q_{1}+1$, where $a_{1}=\left[ 1-(1-n/N)\pi _{\mathbf{0}}^{(1)}
(\boldsymbol{\theta }_{1}^{\ast })\right]/ $ $(1-n/N)\pi _{\mathbf{0}}^{(1)}
(\boldsymbol{\theta }_{1}^{\ast })$, and $a_{i+1}=-\left[ \partial \pi_{\mathbf{0}}^{(1)}
(\boldsymbol{\theta }_{1}^{\ast})/\partial \theta _{i}^{(1)}\right] /
\pi _{\mathbf{0}}^{(1)}(\boldsymbol{\theta}_{1}^{\ast })$, $i=1,\ldots ,q_{1}.$

Let $\mathbf{Z}^{(\tau _{1})}=\left[ Z_{1}^{(\tau _{1})},Z_{2}^{(\tau
_{1})},\ldots ,Z_{q_{1}}^{(\tau _{1})}\right] ^{\prime }$, then by the
previous results we have that%
\begin{equation}
\mathbf{\hat{\Psi}}_{1}^{-1}\left[ \tau _{1}^{1/2}\left( \boldsymbol{\hat{\theta}}%
_{1}^{(\tau _{1})}-\boldsymbol\theta _{1}^{\ast }\right) ^{\prime }\right] -\mathbf{Z}%
^{(\tau _{1})}\overset{P}{\rightarrow }\mathbf{0}^{\prime }\mathbf{,}
\label{difpsiuz}
\end{equation}%
where $\mathbf{\hat{\Psi}}_{1}^{-1}$ is the $q_{1}\times q_{1}$ matrix whose
elements are defined in (\ref{hat-psi-matrix}).

We will show that $\mathbf{Z}^{(\tau _{1})}\overset{D}{\rightarrow }\mathbf{Z%
}\sim N_{q_{1}}(\mathbf{0}^{\prime },\mathbf{\Psi }_{1}^{-1})$ as $\tau
_{1}\rightarrow \infty $, where $\mathbf{Z}=(Z_{1},\ldots
,Z_{q_{1}})^{\prime }$, and that $Z^{(\tau _{1})}\overset{D}{\rightarrow }%
Z\sim N(0,a_{1})$, where $Z^{(\tau _{1})}$ is given by (\ref{z}). To do
this, we will associate with each element $t\in U_{1}$, $t=1,\ldots ,\tau
_{1}$, a random vector $\mathbf{V}_{t}^{(1)}=[V_{t,1}^{(1)},\ldots
,V_{t,q_{1}}^{(1)}]^{\prime }$ and a random variable $V_{t}^{(1)}$ such that

\begin{description}
\item[(a)] $V_{t,j}^{(1)}=[\tilde{\pi}_{\mathbf{x}}^{(1)}(\boldsymbol{\theta }%
_{1}^{\ast })]^{-1}\partial \tilde{\pi}_{\mathbf{x}}^{(1)}(\boldsymbol{\theta }%
_{1}^{\ast })/\partial \theta _{j}^{(1)}$, $j=1,\ldots ,q_{1}$, and $%
V_{t}^{(1)}=1$, if $t\in $ $U_{1}-S_{0}$ and its associated vector $\mathbf{X%
}_{t}^{(1)}$ of link-indicator variables equals the vector $\mathbf{x}\in
\Omega -\{\mathbf{0}\}$;

\item[(b)] $V_{t,j}^{(1)}=0$, $j\!=\!1,\ldots ,q_{1}$, and $V_{t}^{(1)}=-\left[
1-(1-n/N)\pi _{\mathbf{0}}^{(1)}(\boldsymbol{\theta }_{1}^{\ast })\right]\! /\!
\left[ (1-n/N)\pi _{\mathbf{0}}^{(1)}(\boldsymbol{\theta }_{1}^{\ast })\right] $%
, if $t\in $ $U_{1}-S_{0}$ and its associated vector $\mathbf{X}_{t}^{(1)}$
of link-indicator variables equals the vector $\mathbf{0}\in \Omega $, and

\item[(c)] $V_{t,j}^{(1)}=[\pi _{\mathbf{x}}^{(A_{l})}(\boldsymbol{\theta }%
_{1}^{\ast })]^{-1}\partial \pi _{\mathbf{x}}^{(A_{l})}(\boldsymbol{\theta }%
_{1}^{\ast })/\partial \theta _{j}^{(1)}$, $j=1,\ldots ,q_{1}$, and $%
V_{t}^{(1)}=1$, if $t\in $ $A_{l}\in S_{A}$ and its associated vector $%
\mathbf{X}_{t}^{(1)}$ of link-indicator variables equals the vector $\mathbf{%
x}\in \Omega _{-l}$.
\end{description}

Since%
{\setlength\arraycolsep{1pt}
\begin{eqnarray*}
\tau _{1}^{-1/2}\sum_{t=1}^{\tau _{1}}V_{t,j}^{(1)} &=&\tau _{1}^{-1/2}\left[
\sum_{\mathbf{x}\in \Omega -\{\mathbf{0}\}}\frac{R_{\mathbf{x}}^{(\tau _{1})}%
}{\tilde{\pi}_{\mathbf{x}}^{(1)}(\boldsymbol{\theta }_{1}^{\ast })}\frac{%
\partial \tilde{\pi}_{\mathbf{x}}^{(1)}(\boldsymbol{\theta }_{1}^{\ast })}{%
\partial \theta _{j}^{(1)}}+\sum\limits_{l=1}^{n}\sum\limits_{\mathbf{x}%
\in \Omega _{-l}}\frac{R_{\mathbf{x}}^{(A_{l},\tau _{1})}}{\pi _{\mathbf{x}%
}^{(A_{l})}(\boldsymbol{\theta }_{1}^{\ast })}\frac{\partial \pi _{\mathbf{x}%
}^{(A_{l})}(\boldsymbol{\theta }_{1}^{\ast })}{\partial \theta _{j}^{(1)}}\right] \\
&=&Z_{j}^{(\tau _{1})}, \, j =1,\ldots,q_{1},
\end{eqnarray*}}%
it follows that $\mathbf{Z}^{(\tau _{1})}=\tau _{1}^{-1/2}\sum_{t=1}^{\tau
_{1}}\mathbf{V}_{t}^{(1)}$, and 
{\setlength\arraycolsep{1pt}
\begin{eqnarray*}
\tau _{1}^{-1/2}\sum_{t=1}^{\tau _{1}}V_{t}^{(1)}&=&\tau _{1}^{-1/2}\left[
 M^{(\tau _{1})}+R_{1}^{(\tau _{1})} -\left( \tau _{1}-M^{(\tau
_{1})}-R_{1}^{(\tau _{1})}\right) \frac{1-(1-n/N)\pi _{\mathbf{0}}^{(1)}(%
\boldsymbol{\theta }_{1}^{\ast })}{(1-n/N)\pi _{\mathbf{0}}^{(1)}(\boldsymbol{\theta 
}_{1}^{\ast })}\right]  \\
&=& Z^{(\tau _{1})}.
\end{eqnarray*}}

From the definition of $V_{t,j}^{(1)}$ and $V_{t}^{(1)}$ we have that%
{\setlength\arraycolsep{1pt}
\begin{eqnarray*}
\Pr \left\{ V_{t,j}^{(1)}=[\tilde{\pi}_{\mathbf{x}}^{(1)}(\boldsymbol{\theta }%
_{1}^{\ast })]^{-1}\partial{\tilde\pi}_{\mathbf{x}}^{(1)}(\boldsymbol{\theta }%
_{1}^{\ast })/\partial \theta _{j}^{(1)}\right\}& =&(1-n/N)\pi_{%
\mathbf{x}}^{(1)}(\boldsymbol{\theta }_{1}^{\ast }),\text{ }\mathbf{x}\in \Omega
-\{\mathbf{0}\}, 
\nonumber \\
& & \hspace{3.5cm} j=1,\ldots ,q_{1}, 
\nonumber  \\
\Pr \left\{ V_{t,j}^{(1)}=0\right\}& =&(1-n/N)\pi_{\mathbf{0}}^{(1)}
(\boldsymbol{\theta}_{1}^{\ast}),\quad j=1,\ldots,q_{1}, 
\nonumber \\
\Pr \left\{ V_{t,j}^{(1)}=[\pi _{\mathbf{x}}^{(A_{l})}(\boldsymbol{\theta }%
_{1}^{\ast })]^{-1}\partial \pi _{\mathbf{x}}^{(A_{l})}(\boldsymbol{\theta }%
_{1}^{\ast })/\partial \theta _{j}^{(1)}\right\}& =&(1/N)\pi _{\mathbf{x}%
}^{(A_{l})}(\boldsymbol{\theta }_{1}^{\ast })\text{, }\mathbf{x}\in \Omega _{-l}%
\text{, }j=1,\ldots,q_{1},
\nonumber \\
& &\hspace{4.5cm} l=1,\ldots,n,
\end{eqnarray*}}
and%
\begin{equation*}
\Pr \left\{ V_{t}^{(1)}=1\right\} =(1-n/N)\left[ 1-\pi _{\mathbf{0}}^{(1)}(%
\boldsymbol{\theta }_{1}^{\ast })\right] +n/N\text{ and }
\end{equation*}%
\begin{equation*}
\Pr \left\{ V_{t}^{(1)}=-\left[ 1-(1-n/N)\pi _{\mathbf{0}}^{(1)}(\boldsymbol{%
\theta }_{1}^{\ast })\right] /\left[ (1-n/N)\pi _{\mathbf{0}}^{(1)}(\boldsymbol{%
\theta }_{1}^{\ast })\right] \right\} =(1-n/N)\pi _{\mathbf{0}}^{(1)}(%
\boldsymbol{\theta }_{1}^{\ast });
\end{equation*}

therefore, the expected values of the variables $V_{t,j}^{(1)}$ and $%
V_{t}^{(1)}$ are%
{\setlength\arraycolsep{1pt}
\begin{eqnarray*}
E\left( V_{t,j}^{(1)}\right)&=&\sum_{\mathbf{x}\in \Omega -\{\mathbf{0}%
\}}\partial \tilde{\pi}_{\mathbf{x}}^{(1)}(\boldsymbol{\theta }_{1}^{\ast
})/\partial \theta _{j}^{(1)}(1-n/N)\left[ 1-\pi _{\mathbf{0}}^{(1)}(\boldsymbol{%
\theta }_{1}^{\ast })\right]
\nonumber \\
& &+\sum\limits_{l=1}^{n}\sum\limits_{\mathbf{x}\in \Omega _{-l}}\partial
\pi _{\mathbf{x}}^{(A_{l})}(\boldsymbol{\theta }_{1}^{\ast })/\partial\theta_{j}^{(1)}(1/N)
=0,\quad j=1,\ldots,q_{1},
\nonumber
\end{eqnarray*}}%
and%
\begin{equation*}
E\left( V_{t}^{(1)}\right) =(1-n/N)\left[ 1-\pi _{\mathbf{0}}^{(1)}(\boldsymbol{%
\theta }_{1}^{\ast })\right] +n/N-\left[ 1-(1-n/N)\pi _{\mathbf{0}}^{(1)}(%
\boldsymbol{\theta }_{1}^{\ast })\right] =0
\end{equation*}%
because of (\ref{sum-deriv0c}). Thus, $E\left( \mathbf{V}_{t}^{(1)}\right) =%
\mathbf{0}$ and $E\left( V_{t}^{(1)}\right) =0$, $t=1,\ldots ,\tau _{1}$.
Furthermore, their variances are%
{\setlength\arraycolsep{1pt}
\begin{eqnarray*}
V\left( V_{t,j}^{(1)}\right) &=&(1-n/N)\left[ 1-\pi _{\mathbf{0}}^{(1)}(\boldsymbol{%
\theta }_{1}^{\ast })\right]\sum_{\mathbf{x}\in \Omega -\{\mathbf{%
0}\}}\frac{1}{\tilde{\pi}_{\mathbf{x}}^{(1)}(\boldsymbol{\theta }_{1}^{\ast })}%
\left[ \frac{\partial \tilde{\pi}_{\mathbf{x}}^{(1)}(\boldsymbol{\theta }%
_{1}^{\ast })}{\partial \theta _{j}^{(1)}}\right] ^{2} \\
& &+\frac{1}{N}%
\sum\limits_{l=1}^{n}\sum\limits_{\mathbf{x}\in \Omega _{-l}}\frac{1}{\pi
_{\mathbf{x}}^{(A_{l})}(\boldsymbol{\theta }_{1}^{\ast })}\left[ \frac{\partial
\pi _{\mathbf{x}}^{(A_{l})}(\boldsymbol{\theta }_{1}^{\ast })}{\partial \theta
_{j}^{(1)}}\right] ^{2},\quad j=1,\ldots,q_{1},
\end{eqnarray*}}%
and%
{\setlength\arraycolsep{1pt}
\begin{eqnarray*}
V\left( V_{t}^{(1)}\right)&=&(1-n/N)\left[ 1-\pi _{\mathbf{0}}^{(1)}(%
\boldsymbol{\theta }_{1}^{\ast })\right] +n/N+\frac{\left[ 1-(1-n/N)\pi _{%
\mathbf{0}}^{(1)}(\boldsymbol{\theta }_{1}^{\ast })\right] ^{2}}{(1-n/N)\pi _{%
\mathbf{0}}^{(1)}(\boldsymbol{\theta }_{1}^{\ast })}
\nonumber \\
&=&\frac{1-(1-n/N)\pi_{\mathbf{0}}^{(1)}(\boldsymbol{\theta }_{1}^{\ast })}
{(1-n/N)\pi _{\mathbf{0}}^{(1)}(\boldsymbol{\theta }_{1}^{\ast })},
\end{eqnarray*}}%
and their covariances are%
{\setlength\arraycolsep{0pt}
\begin{eqnarray*}
Cov\left( V_{t,j}^{(1)},V_{t,j^{\prime }}^{(1)}\right) &=&\left(1-\frac{n}{N}\right)
\left[1-\pi _{\mathbf{0}%
}^{(1)}(\boldsymbol{\theta }_{1}^{\ast })\right]\sum_{\mathbf{x}\in \Omega -\{\mathbf{0%
}\}}\frac{1}{\tilde{\pi}_{\mathbf{x}}^{(1)}(\boldsymbol{\theta }_{1}^{\ast })}%
\frac{\partial \tilde{\pi}_{\mathbf{x}}^{(1)}(\boldsymbol{\theta }_{1}^{\ast })}{%
\partial \theta _{j}^{(1)}}\frac{\partial \tilde{\pi}_{\mathbf{x}}^{(1)}(%
\boldsymbol{\theta }_{1}^{\ast })}{\partial \theta _{j^{\prime }}^{(1)}} \\
&&+\frac{1}{N}\sum\limits_{l=1}^{n}\sum\limits_{\mathbf{x}\in \Omega _{-l}}%
\frac{1}{\partial \pi _{\mathbf{x}}^{(A_{l})}(\boldsymbol{\theta }_{1}^{\ast })}%
\frac{\partial \pi _{\mathbf{x}}^{(A_{l})}(\boldsymbol{\theta }_{1}^{\ast })}{%
\partial \theta _{j}^{(1)}}\frac{\partial \pi _{\mathbf{x}}^{(A_{l})}(%
\boldsymbol{\theta }_{1}^{\ast })}{\partial \theta _{j^{\prime }}^{(1)}},
\, j,j^{\prime}\!=\!1,\ldots,q_{1},\, j\!\neq\! j^{\prime},
\end{eqnarray*}}
and
\begin{equation*}
Cov\left( V_{t}^{(1)},V_{t,j}^{(1)}\right)\!\!=\!\!\left(1\!-\!\frac{n}{N}\right)\!\!
\left[1\! -\! \pi _{\mathbf{0}%
}^{(1)}(\boldsymbol{\theta }_{1}^{\ast })\right]\!\sum_{\mathbf{x}\in
\Omega -\{\mathbf{0}\}}\frac{\partial \tilde{\pi}_{\mathbf{x}}^{(1)}(\boldsymbol{%
\theta }_{1}^{\ast })}{\partial \theta _{j}^{(1)}}+\frac{1}{N}%
\sum\limits_{l=1}^{n}\sum\limits_{\mathbf{x}\in \Omega _{-l}}\frac{%
\partial \pi _{\mathbf{x}}^{(A_{l})}(\boldsymbol{\theta }_{1}^{\ast })}{\partial
\theta _{j}^{(1)}}\!=\!0.
\end{equation*}
Therefore, the variance-covariance matrix of $\mathbf{V}_{t}^{(1)}$ is $%
\mathbf{\Psi }_{1}^{-1}$.

Finally, since the $(\mathbf{V}_{t}^{(1)\prime },V_{t}^{(1)})^{\prime }$, $%
t=1,\ldots ,\tau _{1}$, are independent and identically distributed random
vectors, by the central limit theorem it follows that%
\begin{equation*}
(\mathbf{Z}^{(\tau _{1})\prime },Z^{(\tau _{1})})^{\prime }=\tau
_{1}^{-1/2}\sum\nolimits_{t=1}^{\tau _{1}}(\mathbf{V}_{t}^{(1)\prime
},V_{t}^{(1)})^{\prime }\overset{D}{\rightarrow }(\mathbf{Z}^{\prime
},Z)^{\prime }\sim N_{q_{1}+1}\left( \mathbf{0}_{q_{1}+1}^{\prime },\left[ 
\begin{array}{cc}
\mathbf{\Psi }_{1}^{-1} & \mathbf{0}^{\prime } \\ 
\mathbf{0} & a_{1}%
\end{array}%
\right] \right).
\end{equation*}%
Thus, $\mathbf{Z}^{(\tau _{1})}\overset{D}{\rightarrow }\mathbf{Z}\sim
N_{q_{1}}(\mathbf{0}^{\prime },\mathbf{\Psi }_{1}^{-1})$ and $Z^{(\tau _{1})}%
\overset{D}{\rightarrow }Z\sim N(0,a_{1}).$ Consequently by (\ref{difpsiuz}) 
\begin{equation*}
\tau _{1}^{1/2}\left[ \boldsymbol{\hat{\theta}}_{1}^{(\tau _{1})}-\boldsymbol{%
\theta }_{1}^{\ast }\right] \overset{D}{\rightarrow }(\mathbf{\Psi }_{1}%
\mathbf{Z})^{\prime }\sim N_{q_{1}}(\mathbf{0},\mathbf{\Psi }_{1})
\end{equation*}%
as $\mathbf{\hat{\Psi}}_{1}\overset{P}{\rightarrow }\mathbf{\Psi }_{1}$.

At last, from (\ref{difz}) and the previous results%
{\setlength\arraycolsep{1pt}
\begin{eqnarray*}
\lefteqn{\tau _{1}^{-1/2}\left( \hat{\tau}_{1}^{(\tau _{1})}-\tau _{1}\right) \overset%
{D}{\rightarrow}  \frac{1}{a_{1}}\left\{Z-\sum\nolimits_{i=1}^{q_{1}}
a_{i+1}[\mathbf{\Psi}_{1}\mathbf{Z}]_i\right\} } \\
& &=\frac{(1-n/N)\pi _{\mathbf{0}}^{(1)}(\boldsymbol{\theta }_{1}^{\ast })}{1-(1-n/N)
\pi _{\mathbf{0}}^{(1)}(\boldsymbol{\theta }_{1}^{\ast })} 
\left[ Z+\frac{1}{\pi _{\mathbf{0}}^{(1)}(\boldsymbol{\theta}%
_{1}^{\ast })}\left[ \nabla \pi _{\mathbf{0}}^{(1)}\left( \boldsymbol{\theta }%
_{1}^{\ast}\right) \right]^{\prime }\mathbf{\Psi}_{1}\mathbf{Z}\right]\sim N(0,\sigma
^{2}),
\end{eqnarray*}}%
where $[\mathbf{\Psi}_{1}\mathbf{Z}]_i$ is the $i$-th element of $\mathbf{\Psi}_{1}\mathbf{Z}$ and
\begin{equation*}
\sigma ^{2}=\frac{1-n/N}{1-(1-n/N)\pi _{\mathbf{0}}^{(1)}(\boldsymbol{\theta }%
_{1}^{\ast })}\left\{ \pi _{\mathbf{0}}^{(1)}(\boldsymbol{\theta }_{1}^{\ast })+%
\frac{(1-n/N)\left[ \nabla \pi _{\mathbf{0}}^{(1)}\left( \boldsymbol{\theta }%
_{1}^{\ast }\right) \right] ^{\prime }\mathbf{\Psi }_{1}\left[ \nabla \pi _{%
\mathbf{0}}^{(1)}\left( \boldsymbol{\theta }_{1}^{\ast }\right) \right] }{%
1-(1-n/N)\pi _{\mathbf{0}}^{(1)}(\boldsymbol{\theta }_{1}^{\ast })}\right\} .
\end{equation*}
\end{proof}

\subsection{Consistency of the UMLE and CMLE of $(\protect\tau _{1},\boldsymbol{%
\protect\theta }_{1}^{\ast })$}

To prove the consistency of the UMLE and CMLE we will use condition (\textbf{%
5)} and the following inequality of information theory: If $\sum a_{i}$ and $%
\sum b_{i}$ are convergent series of positive numbers such that $\sum
a_{i}\geq \sum b_{i}$, then $\sum a_{i}\log (b_{i})\leq \sum a_{i}\log
(a_{i})$, and the equality is attained if and only if $a_{i}=b_{i}$. See Rao
(1973, p. 58).

\subsubsection{Consistency of the UMLE}

Let us first consider $\boldsymbol{\hat{\theta}}_{1}^{(U)}$. Using (\ref{L_(1)})
and (\ref{L_1}) and the definition of the UMLE $\left( \hat{\tau}_{1}^{(U)},%
\boldsymbol{\hat{\theta}}_{1}^{(U)}\right) $ we get that 
{\setlength\arraycolsep{0pt}%
\begin{eqnarray*}
&&l_{(1)}\left( \hat{\tau}_{1}^{(U)},\boldsymbol{\hat{\theta}}_{1}^{(U)}\right)
=\sum\limits_{\mathbf{x}\in \Omega -\{\mathbf{0}\}}R_{\mathbf{x}}^{(\tau
_{1})}\ln \left\{ \pi _{\mathbf{x}}^{(1)}\left( \boldsymbol{\hat{\theta}}%
_{1}^{(U)}\right) /\left[ 1-\pi _{\mathbf{0}}^{(1)}\left( \boldsymbol{\hat{\theta%
}}_{1}^{(U)}\right) \right] \right\} \\
&&+\sum\limits_{i=1}^{n}\sum\limits_{\mathbf{x}\in \Omega _{-i}}R_{\mathbf{%
x}}^{(A_{i},\tau _{1})}\ln \left[ \pi _{\mathbf{x}}^{(A_{i})}\left( \boldsymbol{%
\hat{\theta}}_{1}^{(U)}\right) \right] +\ln \left[ L_{MULT}\left( \hat{\tau}%
_{1}^{(U)}\right) \right] +\ln \left[ L_{12}\left( \hat{\tau}_{1}^{(U)},%
\boldsymbol{\hat{\theta}}_{1}^{(U)}\right) \right] +C \\
&\geq &\sum\limits_{\mathbf{x}\in \Omega -\{\mathbf{0}\}}R_{\mathbf{x}%
}^{(\tau _{1})}\ln \left\{ \pi _{\mathbf{x}}^{(1)}\left( \boldsymbol{\theta }%
_{1}^{\ast }\right) /\left[ 1-\pi _{\mathbf{0}}^{(1)}\left( \boldsymbol{\theta }%
_{1}^{\ast }\right) \right] \right\} +\sum\limits_{i=1}^{n}\sum\limits_{%
\mathbf{x}\in \Omega _{-i}}R_{\mathbf{x}}^{(A_{i},\tau _{1})}\ln \left[ \pi
_{\mathbf{x}}^{(A_{i})}\left( \boldsymbol{\theta }_{1}^{\ast }\right) \right] \\
&&+\ln \left[ L_{MULT}\left( \tau _{1}\right) \right] +\ln \left[
L_{12}\left( \tau _{1},\boldsymbol{\theta }_{1}^{\ast }\right) \right]
+C=l_{(1)}\left( \tau _{1},\boldsymbol{\theta }_{1}^{\ast }\right) ,
\end{eqnarray*}}%
where $C$ depends only on observable variables. Since $\ln\! \left[
L_{MULT}\!\left( \hat{\tau}_{1}^{(U)}\right) \right] $ and $\ln\! \left[
L_{12}\!\left( \hat{\tau}_{1}^{(U)},\right. \right. $ $\left. \left. \boldsymbol{%
\hat{\theta}}_{1}^{(U)}\right) \right] $ are nonpositive we have that
{\setlength\arraycolsep{1pt}%
\begin{eqnarray}
&&\sum\limits_{\mathbf{x}\in \Omega -\{\mathbf{0}\}}\frac{R_{\mathbf{x}%
}^{(\tau _{1})}}{R_{1}^{(\tau _{1})}}\ln \left[ \frac{\pi _{\mathbf{x}%
}^{(1)}\left( \boldsymbol{\hat{\theta}}_{1}^{(U)}\right) }{1-\pi _{\mathbf{0}%
}^{(1)}\left( \boldsymbol{\hat{\theta}}_{1}^{(U)}\right) }\right]
+\sum\limits_{i=1}^{n}\frac{M_{i}^{(\tau _{1})}}{R_{1}^{(\tau _{1})}}%
\sum\limits_{\mathbf{x}\in \Omega _{-i}}\frac{R_{\mathbf{x}}^{(A_{i},\tau
_{1})}}{M_{i}^{(\tau _{1})}}\ln \left[ \pi _{\mathbf{x}}^{(A_{i})}\left( 
\boldsymbol{\hat{\theta}}_{1}^{(U)}\right) \right]  \notag \\
&\geq &\sum\limits_{\mathbf{x}\in \Omega -\{\mathbf{0}\}}\frac{R_{\mathbf{x}%
}^{(\tau _{1})}}{R_{1}^{(\tau _{1})}}\ln \left[ \frac{\pi _{\mathbf{x}%
}^{(1)}(\boldsymbol{\theta }_{1}^{\ast })}{1-\pi _{\mathbf{0}}^{(1)}\left( 
\boldsymbol{\theta }_{1}^{\ast }\right) }\right] +\sum\limits_{i=1}^{n}\frac{%
M_{i}^{(\tau _{1})}}{R_{1}^{(\tau _{1})}}\sum\limits_{\mathbf{x}\in \Omega
_{-i}}\frac{R_{\mathbf{x}}^{(A_{i},\tau _{1})}}{M_{i}^{(\tau _{1})}}\ln %
\left[ \pi _{\mathbf{x}}^{(A_{i})}(\boldsymbol{\theta }_{1}^{\ast })\right] 
\notag \\
&&+\ln \left[ L_{MULT}\left( \tau _{1}\right) \right] /R_{1}^{(\tau
_{1})}+\ln \left[ L_{12}\left( \tau _{1},\boldsymbol{\theta }_{1}^{\ast }\right) %
\right] /R_{1}^{(\tau _{1})}.  \label{umle-ineq-1}
\end{eqnarray}}%
Now, since
{\setlength\arraycolsep{0pt}
\begin{eqnarray*}
1 &=&\sum_{\mathbf{x}\in \Omega -\{\mathbf{0}\}}\frac{R_{\mathbf{x}}^{(\tau
_{1})}}{R_{1}^{(\tau _{1})}}=\sum_{\mathbf{x}\in \Omega -\{\mathbf{0}\}}%
\frac{\pi _{\mathbf{x}}^{(1)}\left( \boldsymbol{\hat{\theta}}_{1}^{(U)}\right) }{%
1-\pi _{\mathbf{0}}^{(1)}\left( \boldsymbol{\hat{\theta}}_{1}^{(U)}\right) }%
\text{ \ and \ }1=\sum_{\mathbf{x}\in \Omega _{-i}}\frac{R_{\mathbf{x}%
}^{(A_{i},\tau _{1})}}{M_{i}^{(\tau _{1})}}=\sum_{\mathbf{x}\in \Omega
_{-i}}\pi _{\mathbf{x}}^{(A_{i})}\left( \boldsymbol{\hat{\theta}}%
_{1}^{(U)}\right), \\
& & \hspace{12cm}i=1,\ldots ,n,
\end{eqnarray*}}%
using $n+1$ times the previously indicated information theory inequality we
have that
{\setlength\arraycolsep{0pt}%
\begin{eqnarray}
&&\sum\limits_{\mathbf{x}\in \Omega -\{\mathbf{0}\}}\frac{R_{\mathbf{x}%
}^{(\tau _{1})}}{R_{1}^{(\tau _{1})}}\ln \left[ R_{\mathbf{x}}^{(\tau
_{1})}/R_{1}^{(\tau _{1})}\right] +\sum\limits_{i=1}^{n}\frac{M_{i}^{(\tau
_{1})}}{R_{1}^{(\tau _{1})}}\sum\nolimits_{\mathbf{x}\in \Omega _{-i}}\frac{%
R_{\mathbf{x}}^{(A_{i},\tau _{1})}}{M_{i}^{(\tau _{1})}}\ln \left[ R_{%
\mathbf{x}}^{(A_{i},\tau _{1})}/M_{i}^{(\tau _{1})}\right]  \notag \\
&\geq &\sum\limits_{\mathbf{x}\in \Omega -\{\mathbf{0}\}}\frac{R_{\mathbf{x}%
}^{(\tau _{1})}}{R_{1}^{(\tau _{1})}}\ln \!\!\left[ \frac{\pi _{\mathbf{x}%
}^{(1)}\left( \boldsymbol{\hat{\theta}}_{1}^{(U)}\right) }{1\!-\!\pi _{\mathbf{0}%
}^{(1)}\left( \boldsymbol{\hat{\theta}}_{1}^{(U)}\right) }\right]\!\!
+\sum\limits_{i=1}^{n}\frac{M_{i}^{(\tau _{1})}}{R_{1}^{(\tau _{1})}}%
\sum\nolimits_{\mathbf{x}\in \Omega _{-i}}\frac{R_{\mathbf{x}}^{(A_{i},\tau
_{1})}}{M_{i}^{(\tau _{1})}}\ln \!\!\left[ \pi _{\mathbf{x}}^{(A_{i})}\left( 
\boldsymbol{\hat{\theta}}_{1}^{(U)}\right) \right]\!\!.  \label{umle-ineq-2}
\end{eqnarray}}

Thus, by (\ref{umle-ineq-1}) and (\ref{umle-ineq-2}) we get that 
{\setlength\arraycolsep{0pt} 
\begin{eqnarray}
0 &\geq &\sum\limits_{\mathbf{x}\in \Omega -\{\mathbf{0}\}}\frac{R_{\mathbf{%
x}}^{(\tau _{1})}}{R_{1}^{(\tau _{1})}}\ln \left\{ \frac{\pi _{\mathbf{x}%
}^{(1)}\left( \boldsymbol{\hat{\theta}}_{1}^{(U)}\right) /\left[ 1-\pi _{\mathbf{%
0}}^{(1)}\left( \boldsymbol{\hat{\theta}}_{1}^{(U)}\right) \right] }{R_{\mathbf{x%
}}^{(1)}/R_{1}^{(\tau _{1})}}\right\}  \notag \\
&&+\sum_{i=1}^{n}\frac{M_{i}^{(\tau _{1})}}{R_{1}^{(\tau _{1})}}%
\sum\limits_{\mathbf{x}\in \Omega _{-i}}\frac{R_{\mathbf{x}}^{(A_{i},\tau
_{1})}}{M_{i}^{(\tau _{1})}}\ln \left[ \frac{\pi _{\mathbf{x}%
}^{(A_{i})}\left( \boldsymbol{\hat{\theta}}_{1}^{(U)}\right) }{R_{\mathbf{x}%
}^{(A_{i},\tau _{1})}/M_{i}^{(\tau _{1})}}\right]  \notag \\
&\geq &\sum\limits_{\mathbf{x}\in \Omega -\{\mathbf{0}\}}\frac{R_{\mathbf{x}%
}^{(\tau _{1})}}{R_{1}^{(\tau _{1})}}\ln\! \left\{ \frac{\pi _{\mathbf{x}%
}^{(1)}(\boldsymbol{\theta }_{1}^{\ast })/\left[ 1-\pi _{\mathbf{0}}^{(1)}(%
\boldsymbol{\theta }_{1}^{\ast })\right] }{R_{\mathbf{x}}^{(1)}/R_{1}^{(\tau
_{1})}}\right\}\! +\sum_{i=1}^{n}\frac{M_{i}^{(\tau _{1})}}{R_{1}^{(\tau
_{1})}}\sum_{\mathbf{x}\in \Omega _{-i}}\frac{R_{\mathbf{x}}^{(A_{i},\tau
_{1})}}{M_{i}^{(\tau _{1})}}  \notag  \\
&&\times\ln\! \left[ \frac{\pi _{\mathbf{x}}^{(A_{i})}(%
\boldsymbol{\theta }_{1}^{\ast })}{R_{\mathbf{x}}^{(A_{i},\tau
_{1})}/M_{i}^{(\tau _{1})}}\right] 
+\ln \left[ L_{MULT}\left( \tau _{1}\right) \right] /R_{1}^{(\tau
_{1})}+\ln \left[ L_{12}\left( \tau _{1},\boldsymbol{\theta }_{1}^{\ast }\right) %
\right] /R_{1}^{(\tau _{1})}.  \label{umle-double-ineq}
\end{eqnarray}}
From the unconditional distributions of $M_{i}^{(\tau _{1})}$, $M^{(\tau
_{1})}$ and $R_{1}^{(\tau _{1})}$, $R_{\mathbf{x}}^{(A_{i},\tau _{1})}$ and $%
R_{\mathbf{x}}^{(A_{i},\tau _{1})}$ indicated in Subsection 4.1, it follows
that $R_{\mathbf{x}}^{(\tau _{1})}/R_{1}^{(\tau _{1})}\overset{P}{%
\rightarrow }\pi _{\mathbf{x}}^{(1)}(\boldsymbol{\theta }_{1}^{\ast })/[1-\pi _{%
\mathbf{0}}^{(1)}(\boldsymbol{\theta }_{1}^{\ast })]$, $R_{\mathbf{x}%
}^{(A_{i},\tau _{1})}/M_{i}^{(\tau _{1})}\overset{P}{\rightarrow }\pi _{%
\mathbf{x}}^{(A_{i})}(\boldsymbol{\theta }_{1}^{\ast })$ and $M_{i}^{(\tau
_{1})}/R_{1}^{(\tau _{1})}\overset{P}{\rightarrow }1/\{(N-n)[1-\pi _{\mathbf{%
0}}^{(1)}(\boldsymbol{\theta }_{1}^{\ast })]\}$. Therefore, the first two
summands of the last term of the double inequality (\ref{umle-double-ineq})
converges to zero in probability, In addition, since $R_{1}^{(\tau
_{1})}/\tau _{1}\overset{P}{\rightarrow }1-\pi _{\mathbf{0}}^{(1)}(\boldsymbol{%
\theta }_{1}^{\ast })$, and from well known results of large deviations
theory (see Varadhan, 2008), we have that for the binomial probability 
$L_{12}\left( \tau _{1},\boldsymbol{\theta }_{1}^{\ast }\right)$:
{\setlength\arraycolsep{1pt}%
\begin{eqnarray*}
\frac{\ln \left[ L_{12}\left( \tau _{1},\boldsymbol{\theta }_{1}^{\ast
}\right) \right] }{R_{1}^{(\tau _{1})}}&=&-\frac{\tau _{1}-M^{(\tau _{1})}}{%
R_{1}^{(\tau _{1})}}\left\{ \frac{R_{1}^{(\tau _{1})}}{\tau _{1}-M^{(\tau
_{1})}}\ln \left[ \frac{R_{1}^{(\tau _{1})}/(\tau _{1}-M^{(\tau _{1})})}{%
1-\pi _{\mathbf{0}}^{(1)}(\boldsymbol{\theta }_{1}^{\ast })}\right] +\frac{\tau
_{1}-R_{1}^{(\tau _{1})}}{\tau _{1}-M^{(\tau _{1})}}\right. \\
&&\times\! \left. \ln \left[ \frac{(\tau _{1}-R_{1}^{(\tau _{1})})/(\tau
_{1}-M^{(\tau _{1})})}{\pi _{\mathbf{0}}^{(1)}(\boldsymbol{\theta }_{1}^{\ast })}%
\right] \right\} +\frac{\tau _{1}-M^{(\tau _{1})}}{R_{1}^{(\tau _{1})}}%
o_{p}(1) \\
&\overset{P}{\rightarrow}&-\frac{[1-\pi _{\mathbf{0}}^{(1)}(\boldsymbol{%
\theta }_{1}^{\ast })]\ln (1)+\pi _{\mathbf{0}}^{(1)}(\boldsymbol{\theta }%
_{1}^{\ast })\ln (1)}{1-\pi _{\mathbf{0}}^{(1)}(\boldsymbol{\theta }_{1}^{\ast })}=0,
\end{eqnarray*}}
and for the multinomial probability $L_{MULT}\left( \tau _{1}\right)$:
{\setlength\arraycolsep{0pt}%
\begin{eqnarray*}
\frac{\ln \left[ L_{MULT}\left( \tau _{1}\right) \right] }{R_{1}^{(\tau
_{1})}} &=&-\frac{\tau _{1}}{R_{1}^{(\tau _{1})}}\left\{ \sum_{i=1}^{n}%
\frac{M_{i}^{(\tau _{1})}}{\tau _{1}}\ln \!\!\left[ \frac{M_{i}^{(\tau
_{1})}/\tau _{1}}{1/N}\right] +\frac{\tau _{1}-M^{(\tau _{1})}}{\tau _{1}}%
\ln \left[ \frac{(\tau _{1}-M^{(\tau _{1})})/\tau _{1}}{1-n/N}\right]
\right\} \\
+\frac{\tau _{1}}{R_{1}^{(\tau _{1})}}o_{p}(1) &\overset{P}{\rightarrow }%
&-\!\left\{ \sum\nolimits_{i=1}^{n}\frac{1}{N}\ln (1)\!+\!(1-n/N)\ln (1)\right\}
\!/\!\left\{ [1\!-\!\pi _{\mathbf{0}}^{(1)}(\boldsymbol{\theta }_{1}^{\ast
})](1\!-\! n/N)\right\} =0.
\end{eqnarray*}}%
The previous results imply that the last term of the double inequality (\ref%
{umle-double-ineq}) converges to zero in probability, and consequently so
does the middle term.

Thus,{\setlength\arraycolsep{1pt}%
\begin{eqnarray*}
&&\sum_{\mathbf{x}\in \Omega -\{\mathbf{0}\}}\frac{\pi _{\mathbf{x}}^{(1)}(%
\boldsymbol{\theta }_{1}^{\ast })}{1-\pi _{\mathbf{0}}^{(1)}(\boldsymbol{\theta }%
_{1}^{\ast })}\ln \left\{ \frac{\pi _{\mathbf{x}}^{(1)}\left( \boldsymbol{\hat{%
\theta}}_{1}^{(U)}\right) /\left[ 1-\pi _{\mathbf{0}}^{(1)}\left( \boldsymbol{%
\hat{\theta}}_{1}^{(U)}\right) \right] }{\pi _{\mathbf{x}}^{(1)}(\boldsymbol{%
\theta }_{1}^{\ast })/\left[ 1-\pi _{\mathbf{0}}^{(1)}(\boldsymbol{\theta }%
_{1}^{\ast })\right] }\right\} \\
&&+\frac{1}{(N-n)\left[ 1-\pi _{\mathbf{0}}^{(1)}(\boldsymbol{\theta }_{1}^{\ast
})\right] }\sum_{i=1}^{n}\sum_{\mathbf{x}\in \Omega _{-i}}\pi _{\mathbf{x}%
}^{(A_{i})}(\boldsymbol{\theta }_{1}^{\ast })\ln \left[ \frac{\pi _{\mathbf{x}%
}^{(A_{i})}\left( \boldsymbol{\hat{\theta}}_{1}^{(U)}\right) }{\pi _{\mathbf{x}%
}^{(A_{i})}(\boldsymbol{\theta }_{1}^{\ast })}\right] \\
&=&\sum_{\mathbf{x}\in \Omega -\{\mathbf{0}\}}\frac{R_{\mathbf{x}}^{(\tau
_{1})}}{R_{1}^{(\tau _{1})}}\ln \left\{ \frac{\pi _{\mathbf{x}}^{(1)}\left( 
\boldsymbol{\hat{\theta}}_{1}^{(U)}\right) /\left[ 1-\pi _{\mathbf{0}%
}^{(1)}\left( \boldsymbol{\hat{\theta}}_{1}^{(U)}\right) \right] }{R_{\mathbf{x}%
}^{(\tau _{1})}/R_{1}^{(\tau _{1})}}\right\} \\
&&+\sum_{i=1}^{n}\frac{M_{i}^{(\tau _{1})}}{R_{1}^{(\tau _{1})}}\sum_{%
\mathbf{x}\in \Omega _{-i}}\frac{R_{\mathbf{x}}^{(A_{i},\tau _{1})}}{%
M_{i}^{(\tau _{1})}}\ln \left[ \frac{\pi _{\mathbf{x}}^{(A_{i})}\left( 
\boldsymbol{\hat{\theta}}_{1}^{(U)}\right) }{R_{\mathbf{x}}^{(A_{i},\tau
_{1})}/M_{i}^{(\tau _{1})}}\right] \\
&&+\sum_{\mathbf{x}\in \Omega -\{\mathbf{0}\}}\left[ \frac{\pi _{\mathbf{x}%
}^{(1)}(\boldsymbol{\theta }_{1}^{\ast })}{1-\pi _{\mathbf{0}}^{(1)}(\boldsymbol{%
\theta }_{1}^{\ast })}-\frac{R_{\mathbf{x}}^{(\tau _{1})}}{R_{1}^{(\tau
_{1})}}\right] \ln \left\{ \frac{\pi _{\mathbf{x}}^{(1)}\left( \boldsymbol{\hat{%
\theta}}_{1}^{(U)}\right) /\left[ 1-\pi _{\mathbf{0}}^{(1)}\left( \boldsymbol{%
\hat{\theta}}_{1}^{(U)}\right) \right] }{R_{\mathbf{x}}^{(\tau
_{1})}/R_{1}^{(\tau _{1})}}\right\} \\
&&+\sum_{i=1}^{n}\sum_{\mathbf{x}\in \Omega_{-i}}\left[\frac{\pi_{\mathbf{x}}^{(A_{i})}
(\boldsymbol{\theta}_{1}^{\ast})}{(N-n)\left[1-\pi_{\mathbf{0}}^{(1)}(\boldsymbol{\theta}
_{1}^{\ast})\right]}-\frac{R_{\mathbf{x}}^{(A_{i})}}{R_{1}^{(\tau_{1})}}\right]\ln
\left[\frac{\pi_{\mathbf{x}}^{(A_{i})}\left(\boldsymbol{\hat{\theta}}_{1}^{(U)}\right)}
{R_{\mathbf{x}}^{(A_{i},\tau _{1})}/M_{i}^{(\tau_{1})}}\right] \\
&&+\sum_{\mathbf{x}\in \Omega -\{\mathbf{0}\}}\frac{\pi _{\mathbf{x}}^{(1)}(%
\boldsymbol{\theta}_{1}^{\ast})}{1-\pi_{\mathbf{0}}^{(1)}(\boldsymbol{\theta}%
_{1}^{\ast })}\ln\left\{\frac{R_{\mathbf{x}}^{(\tau _{1})}/R_{1}^{(\tau _{1})}}
{\pi _{\mathbf{x}}^{(1)}(\boldsymbol{\theta }_{1}^{\ast })/\left[ 1-\pi _{\mathbf{0}}^{(1)}
(\boldsymbol{\theta }_{1}^{\ast })\right] }\right\} \\
&&+\frac{1}{(N-n)\left[ 1-\pi _{\mathbf{0}}^{(1)}(\boldsymbol{\theta }_{1}^{\ast
})\right] }\sum_{i=1}^{n}\sum_{\mathbf{x}\in \Omega _{-i}}\pi _{\mathbf{x}%
}^{(A_{i})}(\boldsymbol{\theta}_{1}^{\ast})\ln\left[ \frac{R_{\mathbf{x}}^{(A_{i},
\tau_{1})}/M_{i}^{(\tau _{1})}}{\pi _{\mathbf{x}%
}^{(A_{i})}(\boldsymbol{\theta }_{1}^{\ast })}\right] \overset{P}{\rightarrow }0
\end{eqnarray*}}
as $\ln \left\{ \left[ \pi _{\mathbf{x}}^{(1)}\left( \boldsymbol{\hat{\theta}}%
_{1}^{(U)}\right)\! /\!\left( 1-\pi _{\mathbf{0}}^{(1)}\left( \boldsymbol{\hat{\theta%
}}_{1}^{(U)}\right) \right) \right] /\left[ R_{\mathbf{x}}^{(\tau
_{1})}/R_{1}^{(\tau _{1})}\right] \right\}$ and $\ln \Big[\pi _{\mathbf{x}%
}^{(A_{i})}\left( \boldsymbol{\hat{\theta}}_{1}^{(U)}\right) /\Big( R_{\mathbf{x%
}}^{(A_{i},\tau _{1})}/$ $\!M_{i}^{(\tau _{1})}\Big)
\Big]$ are bounded as $%
\tau _{1}\rightarrow \infty $ (otherwise the middle term of the inequality (%
\ref{umle-double-ineq}) would not converge to zero). Finally, condition (%
\textbf{5}) implies that for any $\delta _{1}>0$ we have that $\Pr \left\{
\left\Vert \boldsymbol{\hat{\theta}}_{1}^{(U)}-\boldsymbol{\theta }_{1}^{\ast
}\right\Vert \leq \delta _{1}\right\} \rightarrow 1$, that is, $\boldsymbol{\hat{%
\theta}}_{1}^{(U)}\overset{P}{\rightarrow }\boldsymbol{\theta }_{1}^{\ast }$.

Straightforward results of the previous one are the following: $\pi _{%
\mathbf{x}}^{(1)}\left( \boldsymbol{\hat{\theta}}_{1}^{(U)}\right) \overset{P}{%
\rightarrow }\pi _{\mathbf{x}}^{(1)}(\boldsymbol{\theta }_{1}^{\ast })$, $%
\mathbf{x\in \Omega }$, and $\pi _{\mathbf{x}}^{(A_{i})}\left( \boldsymbol{\hat{%
\theta}}_{1}^{(U)}\right) \overset{P}{\rightarrow }\pi _{\mathbf{x}%
}^{(A_{i})}(\boldsymbol{\theta }_{1}^{\ast })$, $\mathbf{x\in \Omega }_{-i}$, $%
i=1,\ldots ,n$, as $\pi _{\mathbf{x}}^{(1)}(\boldsymbol{\theta }_{1})$ and $\pi
_{\mathbf{x}}^{(A_{i})}(\boldsymbol{\theta }_{1})$ are assumed to be continuous
functions of $\boldsymbol{\theta }_{1}$.

With respect to $\hat{\tau}_{1}^{(U)}$, from expression (\ref{umlet1}) we
have that the difference between $\hat{\tau}_{1}^{(U)}$ and $\left( M^{(\tau
_{1})}+ R_{1}^{(\tau _{1})}\right)\! /\! \left[ 1-(1-n/N)\pi _{%
\mathbf{0}}^{(1)}\left( \boldsymbol{\hat{\theta}}_{1}^{(U)}\right)\right]$ is
less than $1$. Thus, $\Big\{\hat{\tau}_{1}^{(U)}-\Big( M^{(\tau_{1})}+
R_{1}^{(\tau _{1})}\Big) /$ $\left[ 1-(1-n/N)\pi _{\mathbf{0}}^{(1)}\left( \boldsymbol{\hat{\theta}}%
_{1}^{(U)}\right) \right] \Big\} /$ $\tau _{1}$ $=\hat{\tau}%
_{1}^{(U)}/\tau _{1}-\left[ \left( M^{(\tau _{1})}+R_{1}^{(\tau
_{1})}\right) /\tau _{1}\right] /\Big[ 1-(1-n/N)$ $\pi _{\mathbf{0}%
}^{(1)}\left( \boldsymbol{\hat{\theta}}_{1}^{(U)}\right) \Big] \overset{P}{%
\rightarrow }0$, and since the second term of the last difference converges
to $1$ in probability so does $\hat{\tau}_{1}^{(U)}/\tau _{1}$.

\subsubsection{Consistency of the CMLE}

By the definition of the CMLE $\boldsymbol{\hat{\theta}}_{1}^{(C)}$, we have that%
{\setlength\arraycolsep{0pt}%
\begin{eqnarray*}
\frac{\ln [L_{11}(\boldsymbol{\hat{\theta}}_{1}^{(C)})L_{0}(\boldsymbol{\hat{\theta}}%
_{1}^{(C)})]}{R_{1}^{(\tau _{1})}} &=&\sum_{\mathbf{x}\in \Omega -\{\mathbf{%
0}\}}\frac{R_{\mathbf{x}}^{(\tau _{1})}}{R_{1}^{(\tau _{1})}}\ln \left[ 
\frac{\pi _{\mathbf{x}}^{(1)}\left( \boldsymbol{\hat{\theta}}_{1}^{(C)}\right) }{%
1-\pi _{\mathbf{0}}^{(1)}\left( \boldsymbol{\hat{\theta}}_{1}^{(C)}\right) }%
\right] \\
&&+\sum_{i=1}^{n}\frac{M_{i}^{(\tau _{1})}}{R_{1}^{(\tau _{1})}}\sum_{%
\mathbf{x}\in \Omega _{-i}}\frac{R_{\mathbf{x}}^{(A_{i},\tau _{1})}}{%
M_{i}^{(\tau _{1})}}\ln \left[ \pi _{\mathbf{x}}^{(A_{i})}\left( \boldsymbol{%
\hat{\theta}}_{1}^{(C)}\right) \right]+C \\
&\geq &\sum_{\mathbf{x}\in \Omega -\{\mathbf{0}\}}\frac{R_{\mathbf{x}%
}^{(\tau _{1})}}{R_{1}^{(\tau _{1})}}\ln \!\!\left[ \frac{\pi _{\mathbf{x}%
}^{(1)}\left( \boldsymbol{\theta }_{1}^{\ast }\right) }{1-\pi _{\mathbf{0}%
}^{(1)}\left( \boldsymbol{\theta }_{1}^{\ast }\right) }\right] \\
&&+\sum_{i=1}^{n}%
\frac{M_{i}^{(\tau _{1})}}{R_{1}^{(\tau _{1})}}\!\sum_{\mathbf{x}\in \Omega
_{-i}}\!\frac{R_{\mathbf{x}}^{(A_{i},\tau _{1})}}{M_{i}^{(\tau _{1})}}\ln\!\! %
\left[ \pi _{\mathbf{x}}^{(A_{i})}\left( \boldsymbol{\theta }_{1}^{\ast }\right) %
\right]+C \\
&=&\frac{\ln [L_{11}(\boldsymbol{\theta }_{1}^{\ast })L_{0}(\boldsymbol{\theta }%
_{1}^{\ast })]}{R_{1}^{(\tau _{1})}},
\end{eqnarray*}}
where $C$ depends only on observable variables.

Using the same procedure as that used in the case of the UMLE $\boldsymbol{\hat{%
\theta}}_{1}^{(U)}$ we will get the double inequality (\ref{umle-double-ineq}%
) but in terms of $\boldsymbol{\hat{\theta}}_{1}^{(C)}$ instead of $\boldsymbol{\hat{%
\theta}}_{1}^{(U)}$ and without the terms $\ln\! \left[ L_{MULT}\!\left( \tau
_{1}\right) \right]\! /R_{1}^{(\tau _{1})}$ and $\ln \left[ L_{12}\left( \tau
_{1},\boldsymbol{\theta }_{1}^{\ast }\right) \right] /R_{1}^{(\tau _{1})}$.
Consequently, we will also have that $\boldsymbol{\hat{\theta}}_{1}^{(C)}\overset%
{P}{\rightarrow }\boldsymbol{\theta }_{1}^{\ast }$, $\pi _{\mathbf{x}%
}^{(1)}\left( \boldsymbol{\hat{\theta}}_{1}^{(C)}\right) \overset{P}{\rightarrow 
}\pi _{\mathbf{x}}^{(1)}(\boldsymbol{\theta }_{1}^{\ast })$, $\mathbf{x\in
\Omega }$, $\pi _{\mathbf{x}}^{(A_{i})}\left( \boldsymbol{\hat{\theta}}%
_{1}^{(C)}\right) \overset{P}{\rightarrow }\pi _{\mathbf{x}}^{(A_{i})}(%
\boldsymbol{\theta }_{1}^{\ast })$, $\mathbf{x\in \Omega }_{-i}$, $i=1,\ldots ,n$%
, and $\hat{\tau}_{1}^{(C)}/\tau _{1}\overset{P}{\rightarrow }1$, where the
last result is obtained by using expression (\ref{cmlet1}) and the same
arguments as those used to prove that $\hat{\tau}_{1}^{(U)}/\tau _{1}\overset%
{P}{\rightarrow }1$.

\subsection{Asymptotic distributions of the UMLEs and CMLEs of $\protect\tau%
_{1}$ and $\boldsymbol{\protect\theta }_{1}^{\ast }$}

\subsubsection{Asymptotic multivariate normal distribution of the UMLE of
$\left(\protect\tau_{1},\boldsymbol{{\protect\theta}}_{1}^{\ast }\right)$}

We will prove the asymptotic multivariate normal distribution of $\left(\tau
_{1}^{-1/2}\hat{\tau}_{1}^{(U)},\tau_{1}^{1/2}\boldsymbol{\hat{\theta}}_{1}^{(U)}
\right)$ by proving that this estimator satisfies the conditions of Theorem 1. 
Condition (i) was already proved in the previous section. From expression 
(\ref{umlet1}) it follows that $\left( \hat{\tau}_{1}^{(U)},\boldsymbol{\hat{\theta}}%
_{1}^{(U)}\right) $ satisfies condition (ii). Finally, by the definition of
the UMLEs we have that condition (iii) is also satisfied. Thus, \ by Theorem
1, $\left[ \tau _{1}^{-1/2}\left( \hat{\tau}_{1}^{(U)}-\tau _{1}\right)
,\tau _{1}^{1/2}\left( \boldsymbol{\hat{\theta}}_{1}^{(U)}-\boldsymbol{\theta }%
_{1}^{\ast }\right) \right] \overset{D}{\rightarrow }N_{q_{1}+1}\left( 
\mathbf{0},\mathbf{\Sigma }_{1}\right) $. This result implies that $\tau
_{1}^{-1/2}\!\left( \hat{\tau}_{1}^{(U)}\!-\!\tau_{1}\right)$ $\overset{D}{%
\rightarrow }N(0,\sigma _{1U}^{2})$ and $\tau_{1}^{1/2}\left( \mathbf{\hat{\theta}}
_{1}^{(U)}-\mathbf{\theta }_{1}^{\ast}\right) \overset{D}{\rightarrow }N_{q_{1}}
\left( \mathbf{0},\mathbf{\Sigma }_{1_{22}}\right) $, where
{\setlength\arraycolsep{0pt}
\begin{eqnarray}
\sigma _{1U}^{2} &=&\frac{1-n/N}{1-(1-n/N)\pi _{\mathbf{0}}^{(1)}(\boldsymbol{%
\theta }_{1}^{\ast })}\left\{ \pi _{\mathbf{0}}^{(1)}(\boldsymbol{\theta }%
_{1}^{\ast })+\frac{1-n/N}{1-(1-n/N)\pi _{\mathbf{0}}^{(1)}(\boldsymbol{\theta }%
_{1}^{\ast })}\left[ \nabla \pi _{\mathbf{0}}^{(1)}\left( \boldsymbol{\theta }%
_{1}^{\ast }\right) \right] ^{\prime }\right. \notag  \\ 
\label{sigma1U2} && \times\mathbf{\Sigma }_{1_{22}}
\left[ \nabla \pi _{\mathbf{0}}^{(1)}\left( \boldsymbol{\theta }_{1}^{\ast
}\right) \right]\! \Bigg\},
\end{eqnarray}}
\begin{equation}
\label{sigma122}\mathbf{\Sigma }_{1_{22}}=\left\{ \Sigma _{1_{22}}^{-1}-\frac{1-n/N}{\pi _{%
\mathbf{0}}^{(1)}(\boldsymbol{\theta }_{1}^{\ast })\left[ 1-(1-n/N)\pi _{\mathbf{%
0}}^{(1)}(\boldsymbol{\theta }_{1}^{\ast })\right] }\left[ \nabla \pi _{\mathbf{0%
}}^{(1)}\left( \boldsymbol{\theta }_{1}^{\ast }\right) \right] \left[ \nabla \pi
_{\mathbf{0}}^{(1)}\left( \boldsymbol{\theta }_{1}^{\ast }\right) \right]
^{\prime }\right\} ^{-1},  
\end{equation}
$\nabla \pi_{\mathbf{0}}^{(1)}\left( \boldsymbol{\theta }_{1}^{\ast
}\right) $ is the gradient of $\pi _{\mathbf{0}}^{(1)}\left( \boldsymbol{\theta }%
_{1}\right) $ evaluated at $\boldsymbol{\theta }_{1}^{\ast }$ and $\Sigma
_{1_{22}}^{-1}$ is the $q_{1}\times q_{1}$ submatrix of $\mathbf{\Sigma }%
_{1}^{-1}$ obtained by removing its first row and first column.

\subsubsection{Asymptotic multivariate normal distribution of the CMLE $%
\boldsymbol{\hat{\protect\theta}}_{1}^{(U)}$ and asymptotic normal distribution
of the CMLE $\hat{\protect\tau}_{1}^{(C)}$}

The CMLE $\left(\tau_1^{-1/2}\hat{\tau}_{1}^{(C)},\tau_1^{1/2}\boldsymbol{\hat{\theta}}
_{1}^{(C)}\right)$ does not have an asymptotic multivariate normal distribution
since this estimator does not satisfy condition (iii) of Theorem 1. To see this, 
notice that by (\ref{derivlnL11L0}) it follows that $\partial \left\{ \ln \left[
L_{11}\left( \boldsymbol{\hat{\theta}}_{1}^{(C)}\right)L_{0}\left( \boldsymbol
{\hat{\theta}}_{1}^{(C)}\right) \right]\right\} /\partial \theta _{j}^{(1)}=0$. 
Therefore, 
{\setlength\arraycolsep{1pt}
\begin{eqnarray*}
\tau _{1}^{-1/2}\frac{\partial }{\partial \theta _{j}^{(1)}}l_{(1)}\left( 
\hat{\tau}_{1}^{(C)},\boldsymbol{\hat{\theta}}_{1}^{(C)}\right) &=&\tau
_{1}^{-1/2}\frac{\partial }{\partial \theta _{j}^{(1)}}\ln \left[
L_{12}\left( \hat{\tau}_{1}^{(C)},\boldsymbol{\hat{\theta}}_{1}^{(C)}\right) %
\right] \\
&=&\tau _{1}^{-1/2}\frac{\partial \pi _{\mathbf{0}}^{(1)}\left( \boldsymbol{\hat{%
\theta}}_{1}^{(C)}\right) }{\partial \theta _{j}^{(1)}}\!\left[ \frac{\hat{\tau%
}_{1}^{(C)}\!-\!M^{(\tau _{1})}\!-\!R_{1}^{(\tau _{1})}}{\pi _{\mathbf{0}%
}^{(1)}\left( \boldsymbol{\hat{\theta}}_{1}^{(C)}\right) }-\frac{R_{1}^{(\tau
_{1})}}{1\!-\!\pi _{\mathbf{0}}^{(1)}\left( \boldsymbol{\hat{\theta}}%
_{1}^{(C)}\right) }\right]\!\! .
\end{eqnarray*}}%
By using expression (\ref{cmlet1}) and after some algebraic steps we get that%
{\setlength\arraycolsep{1pt}
\begin{eqnarray}
& & \tau _{1}^{-1/2}\frac{\partial }{\partial \theta _{j}^{(1)}}l_{(1)}\left( 
\hat{\tau}_{1}^{(C)},\boldsymbol{\hat{\theta}}_{1}^{(C)}\right)=\left[ \frac{%
\partial \pi _{\mathbf{0}}^{(1)}\left( \boldsymbol{\hat{\theta}}%
_{1}^{(C)}\right) }{\partial \theta _{j}^{(1)}}\right] \notag \\
& & \times \left\{ \frac{\hat{\tau}_{1}^{(C)}-\left( M^{(\tau _{1})}+
R_{1}^{(\tau _{1})}\right) /\left[1-(1-n/N)\pi _{\mathbf{0}}^{(1)}
\left( \boldsymbol{\hat{\theta}}_{1}^{(C)}\right) \right] }{\tau _{1}^{1/2}
\pi_{\mathbf{0}}^{(1)}\left(\boldsymbol{\hat{\theta}}_{1}^{(C)}\right)}\right. \notag \\
&&+\left. \tau _{1}^{1/2}\frac{\left( M^{(\tau _{1})}/\tau _{1}\right)
\left( 1-n/N\right) \left[ 1-\pi _{\mathbf{0}}^{(1)}\left( \boldsymbol{\hat{%
\theta}}_{1}^{(C)}\right) \right] -\left( R_{1}^{(\tau _{1})}/\tau
_{1}\right) (n/N)}{\left[ 1-\left( 1-n/N\right) \pi _{\mathbf{0}%
}^{(1)}\left( \boldsymbol{\hat{\theta}}_{1}^{(C)}\right) \right] \left[ 1-\pi _{%
\mathbf{0}}^{(1)}\left( \boldsymbol{\hat{\theta}}_{1}^{(C)}\right) \right] }%
\right\}.   \label{cmle-no-cond-iii}
\end{eqnarray}}%
From (\ref{cmlet1}) and the fact that $\pi _{\mathbf{0}}^{(1)}\left( \boldsymbol{%
\hat{\theta}}_{1}^{(C)}\right) \overset{P}{\rightarrow }\pi _{\mathbf{0}%
}^{(1)}(\boldsymbol{\theta }_{1}^{\ast })$, it follows that the order of
magnitude of the first term in the curly brackets of (\ref{cmle-no-cond-iii}%
) is $O_{p}\left( \tau _{1}^{-1/2}\right) $. On the other hand, since $%
M^{(\tau _{1})}/\tau _{1}=n/N+O_{p}\left( \tau _{1}^{-1/2}\right) $, $%
R_{1}^{(\tau _{1})}/\tau _{1}=(1-n/N)\left[ 1-\pi _{\mathbf{0}}^{(1)}(%
\boldsymbol{\theta }_{1}^{\ast })\right] +O_{p}\left( \tau _{1}^{-1/2}\right) $
and, as we will show in the next paragraph, $\boldsymbol{\hat{\theta}}_{1}^{(C)}=%
\boldsymbol{\theta }_{1}^{\ast }+O_{p}\left( \tau _{1}^{-1/2}\right) $, it
follows that the order of the second term in the curly brackets of (\ref%
{cmle-no-cond-iii}) is $O_{p}\left( 1\right) $; therefore (\ref%
{cmle-no-cond-iii}) does not converge to zero in probability.

Nevertheless, although $\left[ \tau _{1}^{-1/2}\left( \hat{\tau}%
_{1}^{(C)}-\tau _{1}\right) ,\tau _{1}^{1/2}\left( \boldsymbol{\hat{\theta}}%
_{1}^{(C)}-\boldsymbol{\theta }_{1}^{\ast }\right) \right] $ does not have an
asymptotic multivariate normal distribution, $\tau _{1}^{1/2}\left( \boldsymbol{%
\hat{\theta}}_{1}^{(C)}-\boldsymbol{\theta }_{1}^{\ast }\right) $ does have. To
prove this, we will show that conditions (i) and (ii) of Theorem 2 are
satisfied. In the previous section we proved that $\boldsymbol{\hat{\theta}}%
_{1}^{(C)}$ satisfies condition (i), and from (\ref{derivlnL11L0}) we have
that $\boldsymbol{\hat{\theta}}_{1}^{(C)}$ satisfies condition (ii). Thus by
Theorem 2, $\tau _{1}^{1/2}\left( \boldsymbol{\hat{\theta}}_{1}^{(C)}\!-\boldsymbol{%
\theta }_{1}^{\ast }\right)\!\!\overset{D}{\rightarrow }$ $N_{q_{1}}\left(\mathbf{%
0},\mathbf{\Psi }_{1}\right)$.

Now, $\tau _{1}^{-1/2}\left( \hat{\tau}_{1}^{(C)}-\tau _{1}\right) $ has
also an asymptotic normal distribution because in addition that $\boldsymbol{\hat{%
\theta}}_{1}^{(C)}$ satisfies conditions (i) and (ii) of Theorem 2, $\hat{%
\tau}_{1}^{(C)}$ satisfies condition (iii). Thus by Theorem 2, $\tau
_{1}^{-1/2}\left( \hat{\tau}_{1}^{(C)}-\tau _{1}\right)\overset{D}{%
\rightarrow}$ $N(0,\sigma _{1C}^{2})$, where $\sigma _{1C}^{2}$ is given by (%
\ref{sigma12}).

It is worth noting that the asymptotic marginal distributions of $\tau _{1}^{-1/2}
\Big(\hat{\tau}_{1}^{(C)}-\tau_{1}\Big)$ and $\tau _{1}^{1/2}\left(\boldsymbol{\hat{\theta}}
_{1}^{(C)}\!-\boldsymbol{\theta }_{1}^{\ast}\right)$ are not the same as those of
$\tau _{1}^{-1/2}\Big(\hat{\tau}_{1}^{(U)}-\tau_{1}\Big)$ and $\tau _{1}^{1/2}\left
(\boldsymbol{\hat{\theta}}_{1}^{(U)}\!-\boldsymbol{\theta}_{1}^{\ast}\right)$. To show
this, we will firstly prove that
\begin{equation}
\mathbf{\Psi }_{1}^{-1}=\mathbf{\Sigma }_{1_{22}}^{-1}-\frac{1-n/N}{\pi _{%
\mathbf{0}}^{(1)}\left( \boldsymbol{\theta }_{1}^{\ast }\right) \left[ 1-\pi _{%
\mathbf{0}}^{(1)}\left( \boldsymbol{\theta }_{1}^{\ast }\right) \right] }\left[
\nabla \pi _{\mathbf{0}}^{(1)}\left( \boldsymbol{\theta }_{1}^{\ast }\right) %
\right] \left[ \nabla \pi _{\mathbf{0}}^{(1)}\left( \boldsymbol{\theta }%
_{1}^{\ast }\right) \right] ^{\prime },  \label{psi-1sigma-1}
\end{equation}%
where $\mathbf{\Psi }_{1}^{-1}$ is \ the $q_{1}\times q_{1}$ matrix defined
in the statement of Theorem 2 and $\mathbf{\Sigma }_{1_{22}}^{-1}$ is the $%
q_{1}\times q_{1}$ submatrix of the matrix $\mathbf{\Sigma }_{1}^{-1}$,
defined in the statement of Theorem 1, obtained by removing its first row
and first column. Since $\tilde{\pi}_{\mathbf{x}}^{(1)}\left( \boldsymbol{\theta 
}_{1}^{\ast }\right) =\pi _{\mathbf{x}}^{(1)}\left( \boldsymbol{\theta }%
_{1}^{\ast }\right) /\left[ 1-\pi _{\mathbf{0}}^{(1)}\left( \boldsymbol{\theta }%
_{1}^{\ast }\right) \right] $, it follows that 
\begin{eqnarray*}
\frac{\partial \tilde{\pi}_{\mathbf{x}}^{(1)}\left( \boldsymbol{\theta }%
_{1}^{\ast }\right) }{\partial \theta _{j}^{(1)}} &=&\frac{\left[ \partial
\pi _{\mathbf{x}}^{(1)}\left( \boldsymbol{\theta }_{1}^{\ast }\right) /\partial
\theta _{j}^{(1)}\right] \left[ 1-\pi _{\mathbf{0}}^{(1)}\left( \boldsymbol{%
\theta }_{1}^{\ast }\right) \right] +\pi _{\mathbf{x}}^{(1)}\left( \boldsymbol{%
\theta }_{1}^{\ast }\right) \left[ \partial \pi _{\mathbf{0}}^{(1)}\left( 
\boldsymbol{\theta }_{1}^{\ast }\right) /\partial \theta _{j}^{(1)}\right] }{%
\left[ 1-\pi _{\mathbf{0}}^{(1)}\left( \boldsymbol{\theta }_{1}^{\ast }\right) %
\right] ^{2}} \\
&=&\frac{1}{1-\pi _{\mathbf{0}}^{(1)}\left( \boldsymbol{\theta }_{1}^{\ast
}\right) }\frac{\partial \pi _{\mathbf{x}}^{(1)}\left( \boldsymbol{\theta }%
_{1}^{\ast }\right) }{\partial \theta _{j}^{(1)}}+\frac{\pi _{\mathbf{x}%
}^{(1)}\left( \boldsymbol{\theta }_{1}^{\ast }\right) }{\left[ 1-\pi _{\mathbf{0}%
}^{(1)}\left( \boldsymbol{\theta }_{1}^{\ast }\right) \right] ^{2}}\frac{%
\partial \pi _{\mathbf{0}}^{(1)}\left( \boldsymbol{\theta }_{1}^{\ast }\right) }{%
\partial \theta _{j}^{(1)}}.
\end{eqnarray*}%
Then%
{\setlength\arraycolsep{0pt}
\begin{eqnarray*}
&&(1-n/N)\left[ 1-\pi _{\mathbf{0}}^{(1)}\left( \boldsymbol{\theta }_{1}^{\ast
}\right) \right] \sum_{\mathbf{x}\in \Omega -\{\mathbf{0}\}}\frac{1-\pi _{%
\mathbf{0}}^{(1)}\left( \boldsymbol{\theta }_{1}^{\ast }\right) }{\pi _{\mathbf{x%
}}^{(1)}\left( \boldsymbol{\theta }_{1}^{\ast }\right) }\frac{\partial \tilde{\pi%
}_{\mathbf{x}}^{(1)}\left( \boldsymbol{\theta }_{1}^{\ast }\right) }{\partial
\theta _{i}^{(1)}}\frac{\partial \tilde{\pi}_{\mathbf{x}}^{(1)}\left( 
\boldsymbol{\theta }_{1}^{\ast }\right) }{\partial \theta _{j}^{(1)}} \\
&=&(1-n/N)\left[ 1-\pi _{\mathbf{0}}^{(1)}\left( \boldsymbol{\theta }_{1}^{\ast
}\right) \right] \sum_{\mathbf{x}\in \Omega -\{\mathbf{0}\}}\frac{1-\pi _{%
\mathbf{0}}^{(1)}\left( \boldsymbol{\theta }_{1}^{\ast }\right) }{\pi _{\mathbf{x%
}}^{(1)}\left( \boldsymbol{\theta }_{1}^{\ast }\right) }\left[ \frac{1}{1-\pi _{%
\mathbf{0}}^{(1)}\left( \boldsymbol{\theta }_{1}^{\ast }\right) }\frac{\partial
\pi _{\mathbf{x}}^{(1)}\left( \boldsymbol{\theta }_{1}^{\ast }\right) }{\partial
\theta _{i}^{(1)}}\right. \\
&& +\left.\frac{\pi _{\mathbf{x}}^{(1)}\left( \boldsymbol{\theta }%
_{1}^{\ast }\right) }{\left[ 1-\pi _{\mathbf{0}}^{(1)}\left( \boldsymbol{\theta }%
_{1}^{\ast }\right) \right] ^{2}}\frac{\partial \pi _{\mathbf{0}%
}^{(1)}\left( \boldsymbol{\theta }_{1}^{\ast }\right) }{\partial \theta
_{i}^{(1)}}\right]\left[ \frac{1}{1-\pi _{\mathbf{0}}^{(1)}\left( \boldsymbol{\theta }%
_{1}^{\ast }\right) }\frac{\partial \pi _{\mathbf{x}}^{(1)}\left( \boldsymbol{%
\theta }_{1}^{\ast }\right) }{\partial \theta _{j}^{(1)}}+\frac{\pi _{%
\mathbf{x}}^{(1)}\left( \boldsymbol{\theta }_{1}^{\ast }\right) }{\left[ 1-\pi _{%
\mathbf{0}}^{(1)}\left( \boldsymbol{\theta }_{1}^{\ast }\right) \right] ^{2}}%
\frac{\partial \pi _{\mathbf{0}}^{(1)}\left( \boldsymbol{\theta }_{1}^{\ast
}\right) }{\partial \theta _{j}^{(1)}}\right]  \\
&=&(1-n/N)\sum_{\mathbf{x}\in \Omega -\{\mathbf{0}\}}\frac{1}{\pi _{\mathbf{x%
}}^{(1)}\left( \boldsymbol{\theta }_{1}^{\ast }\right) }\frac{\partial \pi _{%
\mathbf{x}}^{(1)}\left( \boldsymbol{\theta }_{1}^{\ast }\right) }{\partial
\theta _{i}^{(1)}}\frac{\partial \pi _{\mathbf{x}}^{(1)}\left( \boldsymbol{%
\theta }_{1}^{\ast }\right) }{\partial \theta _{j}^{(1)}}+\frac{1-n/N}{1-\pi
_{\mathbf{0}}^{(1)}\left( \boldsymbol{\theta }_{1}^{\ast }\right) }\frac{%
\partial \pi _{\mathbf{0}}^{(1)}\left( \boldsymbol{\theta }_{1}^{\ast }\right) }{%
\partial \theta _{j}^{(1)}} \\
&&\times \sum_{\mathbf{x}\in \Omega -\{\mathbf{0}\}}\frac{%
\partial \pi _{\mathbf{x}}^{(1)}\left( \boldsymbol{\theta }_{1}^{\ast }\right) }{%
\partial \theta _{i}^{(1)}}
+\frac{1-n/N}{1-\pi _{\mathbf{0}}^{(1)}\left( \boldsymbol{\theta }_{1}^{\ast
}\right) }\frac{\partial \pi _{\mathbf{0}}^{(1)}\left( \boldsymbol{\theta }%
_{1}^{\ast }\right) }{\partial \theta _{i}^{(1)}}\sum_{\mathbf{x}\in \Omega
-\{\mathbf{0}\}}\frac{\partial \pi _{\mathbf{x}}^{(1)}\left( \boldsymbol{\theta }%
_{1}^{\ast }\right) }{\partial \theta _{j}^{(1)}}+\frac{1-n/N}{\left[ 1-\pi
_{\mathbf{0}}^{(1)}\left( \boldsymbol{\theta }_{1}^{\ast }\right) \right] ^{2}}  \\
&&\times\frac{\partial \pi _{\mathbf{0}}^{(1)}\left( \boldsymbol{\theta }_{1}^{\ast
}\right) }{\partial \theta _{i}^{(1)}}\frac{\partial \pi _{\mathbf{0}%
}^{(1)}\left( \boldsymbol{\theta }_{1}^{\ast }\right) }{\partial \theta
_{j}^{(1)}} \sum_{\mathbf{x}\in \Omega -\{\mathbf{0}\}}\pi _{\mathbf{x}%
}^{(1)}\left( \boldsymbol{\theta }_{1}^{\ast }\right) \\
&=&(1-n/N)\sum_{\mathbf{x}\in \Omega -\{\mathbf{0}\}}\frac{1}{\pi _{\mathbf{x%
}}^{(1)}\left( \boldsymbol{\theta }_{1}^{\ast }\right) }\frac{\partial \pi _{%
\mathbf{x}}^{(1)}\left( \boldsymbol{\theta }_{1}^{\ast }\right) }{\partial
\theta _{i}^{(1)}}\frac{\partial \pi _{\mathbf{x}}^{(1)}\left( \boldsymbol{%
\theta }_{1}^{\ast }\right) }{\partial \theta _{j}^{(1)}}-\frac{2(1-n/N)}{%
1-\pi _{\mathbf{0}}^{(1)}\left( \boldsymbol{\theta }_{1}^{\ast }\right) }\frac{%
\partial \pi _{\mathbf{0}}^{(1)}\left( \boldsymbol{\theta }_{1}^{\ast }\right) }{%
\partial \theta _{i}^{(1)}}\frac{\partial \pi _{\mathbf{0}}^{(1)}\left( 
\boldsymbol{\theta }_{1}^{\ast }\right) }{\partial \theta _{j}^{(1)}} \\
&&+\frac{1-n/N}
{1-\pi _{\mathbf{0}}^{(1)}\left( \boldsymbol{\theta }_{1}^{\ast }\right) }\frac{%
\partial \pi _{\mathbf{0}}^{(1)}\left( \boldsymbol{\theta }_{1}^{\ast }\right) }{%
\partial \theta _{i}^{(1)}}\frac{\partial \pi _{\mathbf{0}}^{(1)}\left( 
\boldsymbol{\theta }_{1}^{\ast }\right) }{\partial \theta _{j}^{(1)}}
\end{eqnarray*}}
{\setlength\arraycolsep{0pt}
\begin{eqnarray*}
&=&(1-n/N)\sum_{\mathbf{x}\in \Omega -\{\mathbf{0}\}}\frac{1}{\pi _{\mathbf{x%
}}^{(1)}\left( \boldsymbol{\theta }_{1}^{\ast }\right) }\frac{\partial \pi _{%
\mathbf{x}}^{(1)}\left( \boldsymbol{\theta }_{1}^{\ast }\right) }{\partial
\theta _{i}^{(1)}}\frac{\partial \pi _{\mathbf{x}}^{(1)}\left( \boldsymbol{%
\theta }_{1}^{\ast }\right) }{\partial \theta _{j}^{(1)}}-\frac{1-n/N}{1-\pi
_{\mathbf{0}}^{(1)}\left( \boldsymbol{\theta }_{1}^{\ast }\right) }\frac{%
\partial \pi _{\mathbf{0}}^{(1)}\left( \boldsymbol{\theta }_{1}^{\ast }\right) }{%
\partial \theta _{i}^{(1)}}\frac{\partial \pi _{\mathbf{0}}^{(1)}\left( 
\boldsymbol{\theta }_{1}^{\ast }\right) }{\partial \theta _{j}^{(1)}} \\
&=&(1-n/N)\sum_{\mathbf{x}\in \Omega }\frac{1}{\pi _{\mathbf{x}}^{(1)}\left( 
\boldsymbol{\theta }_{1}^{\ast }\right) }\frac{\partial \pi _{\mathbf{x}%
}^{(1)}\left( \boldsymbol{\theta }_{1}^{\ast }\right) }{\partial \theta
_{i}^{(1)}}\frac{\partial \pi _{\mathbf{x}}^{(1)}\left( \boldsymbol{\theta }%
_{1}^{\ast }\right) }{\partial \theta _{j}^{(1)}}-\frac{1-n/N}{\pi _{\mathbf{%
0}}^{(1)}\left( \boldsymbol{\theta }_{1}^{\ast }\right) \left[ 1-\pi _{\mathbf{0}%
}^{(1)}\left( \boldsymbol{\theta }_{1}^{\ast }\right) \right] }  \\
&&\times\frac{\partial
\pi _{\mathbf{0}}^{(1)}\left( \boldsymbol{\theta }_{1}^{\ast }\right) }{\partial
\theta _{i}^{(1)}}\frac{\partial \pi _{\mathbf{0}}^{(1)}\left( \boldsymbol{%
\theta }_{1}^{\ast }\right) }{\partial \theta _{j}^{(1)}}.
\end{eqnarray*}}%
Therefore, from the definitions of $\mathbf{\Psi }_{1}^{-1}$ and $\Sigma
_{1_{22}}^{-1}$ we have that%
\begin{equation*}
\left[ \mathbf{\Psi }_{1}^{-1}\right] _{i,j}=\left[ \Sigma _{1_{22}}^{-1}%
\right] _{i,j}-\frac{1-n/N}{\pi _{\mathbf{0}}^{(1)}\left( \boldsymbol{\theta }%
_{1}^{\ast }\right) \left[ 1-\pi _{\mathbf{0}}^{(1)}\left( \boldsymbol{\theta }%
_{1}^{\ast }\right) \right] }\frac{\partial \pi _{\mathbf{0}}^{(1)}\left( 
\boldsymbol{\theta }_{1}^{\ast }\right) }{\partial \theta _{i}^{(1)}}\frac{%
\partial \pi _{\mathbf{0}}^{(1)}\left( \boldsymbol{\theta }_{1}^{\ast }\right) }{%
\partial \theta _{j}^{(1)}},
\end{equation*}%
and (\ref{psi-1sigma-1}) is proved. 

From (\ref{sigma1U2}) and (\ref{psi-1sigma-1}) it follows that $\sigma
_{1C}^{2}\!\neq\!\sigma _{1U}^{2}$, and hence $%
\tau _{1}^{-1/2}\!\left(\hat{\tau}_{1}^{(U)}\!-\!\tau _{1}\right) $ and $\tau
_{1}^{-1/2}\!\Big(\hat{\tau}_{1}^{(C)}\!-$ $\tau _{1}\Big) $ do not have the same
asymptotic normal distribution. In addition, (\ref{sigma122}) and (\ref%
{psi-1sigma-1}) imply that $\mathbf{\Psi }_{1}\neq \mathbf{\Sigma }_{1_{22}}$%
, and consequently that $\tau _{1}^{1/2}\left( \boldsymbol{\hat{\theta}}%
_{1}^{(U)}-\boldsymbol{\theta }_{1}^{\ast }\right) $ and $\tau _{1}^{1/2}\left( 
\boldsymbol{\hat{\theta}}_{1}^{(C)}-\boldsymbol{\theta }_{1}^{\ast }\right) $ do not
have the same asymptotic normal distribution. Notice also that even though the
asymptotic marginal distributions of the UMLEs and CMLEs of $\tau_1$ and 
$\boldsymbol\theta_1^{\ast}$ are not the same, from (\ref{psi-1sigma-1}) it follows
that if $n/N$ were small enough so that $1-(1-n/N)\pi_{\mathbf{0}}^{(1)}\left
(\boldsymbol{\theta}_{1}^{\ast }\right)\approx 1-\pi_{\mathbf{0}}^{(1)}\left
(\boldsymbol{\theta}_{1}^{\ast }\right)$, then $\mathbf{\Psi }_{1}\approx
\mathbf{\Sigma }_{1_{22}}$ and their asymptotic marginal distributions
would be very similar to each other.

\subsection{Asymptotic properties of unconditional and conditional maximum 
likelihood estimators of $(\protect\tau_{2},\boldsymbol{\protect\theta}_{2}^{\ast})$}

The unconditional and conditional maximum likelihood estimators of $(\tau
_{2},\boldsymbol{\theta }_{2}^{\ast })$ are exactly the same as those used in
capture-recapture studies. Sanathanan (1972) assumed conditions similar to 
\textbf{(1)}-\textbf{(4)} and \textbf{(6)} and proved the following results:

\begin{description}
\item[(i)] $\boldsymbol{\hat{\theta}}_{2}^{(U)}\overset{P}{\rightarrow }\boldsymbol{%
\theta }_{2}^{\ast }$ and $\boldsymbol{\hat{\theta}}_{2}^{(C)}\overset{P}{%
\rightarrow }\boldsymbol{\theta }_{2}^{\ast }$ as $\tau _{2}\rightarrow \infty $.

\item[(ii)] $\hat{\tau}_{2}^{(U)}/\tau _{2}\overset{P}{\rightarrow }1$ and $%
\hat{\tau}_{2}^{(C)}/\tau _{2}\overset{P}{\rightarrow }1$ as $\tau _{2}\rightarrow
\infty$.

\item[(iii)] $\left[ \tau _{2}^{-1/2}\!\left( \hat{\tau}_{2}^{(U)}\!-\tau
_{2}\right)\! ,\tau _{2}^{1/2}\!\left( \boldsymbol{\hat{\theta}}_{2}^{(U)}\!-\boldsymbol{%
\theta }_{2}^{\ast }\right)\! \right]\! \overset{D}{\rightarrow }\! N_{q_{2}}\left( 
\mathbf{0},\mathbf{\Sigma }_{2}\right) $ and $\Big[ \tau _{2}^{-1/2}\!\left( 
\hat{\tau}_{2}^{(C)}\!-\tau _{2}\right)\! ,\tau _{2}^{1/2}\Big( \boldsymbol{\hat{%
\theta}}_{2}^{(C)}\!-$ $\boldsymbol{\theta}_{2}^{\ast}\Big)\! \Big]$ $\overset{D}{%
\rightarrow }N_{q_{2}}\left( \mathbf{0},\mathbf{\Sigma }_{2}\right) $ as $%
\tau _{2}\rightarrow \infty$,

where $\mathbf{\Sigma }_{2}$ is the inverse of the $(q_{2}+1)\times
(q_{2}+1) $ matrix $\mathbf{\Sigma }_{2}^{-1}$ defined by%
{\setlength\arraycolsep{0pt}
\begin{eqnarray*}
\left[ \mathbf{\Sigma }_{2}^{-1}\right] _{1,1} &=&\left[ 1-\pi _{\mathbf{0}%
}^{(2)}(\boldsymbol{\theta }_{2}^{\ast })\right] /\pi _{\mathbf{0}}^{(2)}(%
\boldsymbol{\theta }_{2}^{\ast }), \\
\left[ \mathbf{\Sigma }_{2}^{-1}\right] _{1,j+1} &=&\left[ \mathbf{\Sigma }%
_{2}^{-1}\right] _{j+1,1}=-\left[ 1/\pi _{\mathbf{0}}^{(2)}(\boldsymbol{\theta }%
_{2}^{\ast })\right] \left[ \partial \pi _{\mathbf{0}}^{(2)}(\boldsymbol{\theta }%
_{2}^{\ast })/\partial \theta _{j}^{(2)}\right] \text{, \ }j=1,\ldots ,q_{2},
\\
\left[ \mathbf{\Sigma }_{2}^{-1}\right] _{i+1,j+1} &=&\left[ \mathbf{\Sigma }%
_{2}^{-1}\right] _{j+1,i+1}=\sum_{\mathbf{x}\in \Omega }\left[
1/\pi _{\mathbf{x}}^{(2)}(\boldsymbol{\theta }_{2}^{\ast })\right]\! \left[
\partial \pi _{\mathbf{x}}^{(2)}(\boldsymbol{\theta }_{2}^{\ast })/\partial
\theta _{i}^{(2)}\right]\! \left[ \partial \pi _{\mathbf{x}}^{(2)}(\boldsymbol{%
\theta }_{2}^{\ast })/\partial \theta _{j}^{(2)}\right] , \\
&&\hspace{45ex}i,j=1,\ldots ,q_{2},
\end{eqnarray*}}%
and which is assumed to be a non-singular matrix.
\end{description}

Because the proofs of these results are exactly the same as those given by
Sanathanan (1972), we will omit them. It is worth noting that unlike the
CMLE $\left(\tau_{1}^{-1/2}\hat{\tau}_{1}^{(C)},\tau_{1}^{1/2}
\boldsymbol{\hat{\theta}}_{1}^{(C)}\right)$, the estimator $\left(\tau _{2}^{-1/2}
\hat{\tau}_{2}^{(C)},\tau _{2}^{1/2}\boldsymbol{\hat{\theta}}%
_{2}^{(C)}\right) $ does have an asymptotic multivariate normal distribution.

The previous results imply that $\tau _{2}^{-1/2}\!\left( \hat{\tau}%
_{2}^{(U)}\!-\tau _{2}\right)\! \overset{D}{\rightarrow }\! N(0,\sigma _{2}^{2})$
and $\tau _{2}^{-1/2}\!\left( \hat{\tau}_{2}^{(C)}\!-\tau _{2}\right)\! \overset{D}%
{\rightarrow }\! N(0,\sigma _{2}^{2})$, where%
\begin{eqnarray*}
\sigma _{2}^{2} &=&\frac{1}{1-\pi _{\mathbf{0}}^{(2)}(\boldsymbol{\theta }%
_{2}^{\ast })}\left\{ \pi _{\mathbf{0}}^{(2)}(\boldsymbol{\theta }_{2}^{\ast })+%
\frac{1}{1-\pi _{\mathbf{0}}^{(2)}(\boldsymbol{\theta }_{2}^{\ast })}\left[
\nabla \pi _{\mathbf{0}}^{(2)}\left( \boldsymbol{\theta }_{2}^{\ast }\right) %
\right] ^{\prime }\left[ \Sigma _{2_{22}}^{-1}\right. \right. \\
&&\left. \left. -\frac{1}{\pi _{\mathbf{0}}^{(2)}(\boldsymbol{\theta }_{2}^{\ast
})\left[ 1-\pi _{\mathbf{0}}^{(2)}(\boldsymbol{\theta }_{2}^{\ast })\right] }%
\left[ \nabla \pi _{\mathbf{0}}^{(2)}\left( \boldsymbol{\theta }_{2}^{\ast
}\right) \right] \left[ \nabla \pi _{\mathbf{0}}^{(2)}\left( \boldsymbol{\theta }%
_{2}^{\ast }\right) \right] ^{\prime }\right] ^{-1}\left[ \nabla \pi _{%
\mathbf{0}}^{(2)}\left( \boldsymbol{\theta }_{2}^{\ast }\right) \right] \right\},
\end{eqnarray*}%
where $\nabla \pi _{\mathbf{0}}^{(2)}\left( \boldsymbol{\theta }_{2}^{\ast
}\right) $ is the gradient of $\pi _{\mathbf{0}}^{(2)}\left( \boldsymbol{\theta }%
_{2}\right) $ evaluated at $\boldsymbol{\theta }_{2}^{\ast }$ and $\Sigma
_{2_{22}}^{-1}$ is the $q_{2}\times q_{2}$ submatrix of $\mathbf{\Sigma }%
_{2}^{-1}$ obtained by removing its first row and first column.

\begin{sloppypar} \subsection{Consistency and asymptotic normality of the unconditional and
 conditional maximum likelihood  estimators of $\protect\tau =\protect\tau %
_{1}+\protect\tau _{2}$ } \end{sloppypar}

The UMLE and CMLE of $\tau =\tau _{1}+\tau _{2}$ were defined in Subsection
3.3.3 by $\hat{\tau}^{(U)}=\hat{\tau}_{1}^{(U)}+\hat{\tau}_{2}^{(U)}$ and $%
\hat{\tau}^{(C)}=\hat{\tau}_{1}^{(C)}+\hat{\tau}_{2}^{(C)}$. From
assumptions \textbf{A} and \textbf{B} and the previous results we have that $%
\hat{\tau}^{(U)}/\tau =(\tau _{1}/\tau )\left( \hat{\tau}_{1}^{(U)}/\tau
_{1}\right) +(\tau _{2}/\tau )\left( \hat{\tau}_{2}^{(U)}/\tau _{2}\right) 
\overset{P}{\rightarrow }\alpha _{1}\times 1+\alpha _{2}\times 1=1$, as $%
\tau _{1}\rightarrow \infty $ and $\tau _{2}\rightarrow \infty $. Similarly, 
$\hat{\tau}^{(C)}/\tau \overset{P}{\rightarrow }1$ as $\tau _{1}\rightarrow
\infty $ and $\tau _{2}\rightarrow \infty $. Furthermore, $\tau
^{-1/2}\!\left(\hat{\tau}^{(U)}-\tau \right)\!=\!(\tau /\tau _{1})^{-1/2}\tau
_{1}^{-1/2}\!\left(\hat{\tau}_{1}^{(U)}\!-\tau _{1}\right)$ $+(\tau /\tau
_{2})^{-1/2}\tau _{2}^{-1/2}\ \left( \hat{\tau}_{2}^{(U)}-\tau _{2}\right) 
\overset{D}{\rightarrow }N(0,\sigma^{2}_U)$, as $\tau _{1}\rightarrow
\infty $ and $\tau _{2}\rightarrow \infty $, where $\sigma^{2}_U=
\alpha _{1}\sigma _{1U}^{2}+\alpha_{2}\sigma_{2}^{2}$. Likewise, 
$\tau^{-1/2}\left( \hat{\tau}^{(C)}-\tau \right) \overset{D}{\rightarrow }%
N(0,\sigma^{2}_C)$, where $\sigma^{2}_C=\alpha _{1}\sigma _{1C}^{2}+\alpha_{2}
\sigma_{2}^{2}$.

\vspace{3ex}

\section{Estimation of the matrices $\mathbf{\Sigma }_{k}^{-1}$ and $\mathbf{%
\Psi }_{1}^{-1}$}

Although estimates of $\mathbf{\Sigma }_{k}^{-1}$, $k=1,2$, and $\mathbf{%
\Psi }_{1}^{-1}$ can be obtained by replacing the parameters $\boldsymbol{\theta 
}_{k}^{\ast }$ and $\boldsymbol{\theta }_{1}^{\ast }$ by their respective
estimates in the expressions for these matrices, this procedure requires the 
computation of sums of $2^{n}$ terms. This is not a problem if $n$ is small, but 
if $n$ is large enough, say greater than or equal to $20$, the number of these 
terms is very large and the calculation of the estimates of $\mathbf{\Sigma }_
{k}^{-1}$ and $\mathbf{\Psi }_{1}^{-1}$ could be computationally expensive.

A procedure that requires a much smaller number of calculations is based on
estimates of the vectors $\mathbf{V}_{t}^{(k)}$, $t=1,\ldots ,\tau _{k}$, $%
k=1,2$. Vectors $\mathbf{V}_{t}^{(1)}$s were defined in the proofs of
Theorems 1 and 2, whereas vectors $\mathbf{V}_{t}^{(2)}$s are defined in
Sanathanan (1972) and we will give their definition later in this section.
As was shown in the proofs of Theorem 1 and 2, the vectors $\mathbf{V}%
_{t}^{(1)}$s are independent and equally distributed with mean vector equal
to the vector zero and covariance matrix equal to $\mathbf{\Sigma }_{1}^{-1}$
in the case of Theorem 1, and $\mathbf{\Psi }_{1}^{-1}$ in the case of
Theorem 2. The same result holds in the case of the vectors $\mathbf{V}%
_{t}^{(2)}$s, but the covariance matrix is $\mathbf{\Sigma }_{2}^{-1}$.
Therefore, the sample covariance matrix of the vectors $\mathbf{V}_{t}^{(k)}$%
s is an estimate of their covariance matrix ( $\mathbf{\Sigma }_{k}^{-1}$ or 
$\mathbf{\Psi }_{1}^{-1}$) based only on $\tau _{k}$ observations.

To implement this procedure we need to estimate the $\mathbf{V}_{t}^{(k)}$s
(they are unknown because depend on $\boldsymbol{\theta}_{k}^{\ast }$ and $%
\tau _{k}$). In the case of the $\mathbf{V}_{t}^{(1)}$s defined in Theorem 1
an estimate $\mathbf{\hat{V}}_{t}^{(1)}$ of $\mathbf{V}_{t}^{(1)}$ could be
obtained by replacing $\boldsymbol{\theta }_{1}^{\ast }$ by $\boldsymbol{\hat{\theta}%
}_{1}^{(U)}$ in the expression for $\mathbf{V}_{t}^{(1)}$, and $\tau _{1}$
could be estimated by $\hat{\tau}_{1}^{(U)}$. In the case of the $\mathbf{V}%
_{t}^{(1)}$s defined in Theorem 2 estimates of $\mathbf{V}_{t}^{(1)}$ could
be obtained by replacing $\boldsymbol{\theta }_{1}^{\ast }$ by $\boldsymbol{\hat{%
\theta}}_{1}^{(C)}$ in the expression for $\mathbf{V}_{t}^{(1)}$, and $\tau
_{1}$ could be estimated by $\hat{\tau}_{1}^{(C)}$. Estimates of $\mathbf{V}%
_{t}^{(2)}$s and $\tau _{2}$ could be obtained as in the case of Theorem 1,
and in this situation both UMLE and CMLE could be used. Thus, once $\hat{\tau%
}_{k}$ and the vectors $\mathbf{\hat{V}}_{t}^{(k)}$s are obtained, their
sample covariance matrix can be computed and used as an estimate of $\mathbf{%
\Sigma }_{k}^{-1}$ or $\mathbf{\Psi }_{1}^{-1}$.

The vectors $\mathbf{V}_{t}^{(2)}=[V_{t,1}^{(2)},\ldots
,V_{t,q_{2}+1}^{(2)}]^{\prime }$, $t=1,\ldots ,\tau _{2}$, are defined as
follows:

\begin{description}
\item[(a)] $V_{t,1}^{(2)}=1$ and $V_{t,j+1}^{(2)}=[\pi _{\mathbf{x}}^{(2)}(%
\boldsymbol{\theta }_{2}^{\ast })]^{-1}\partial \pi _{\mathbf{x}}^{(2)}(\boldsymbol{%
\theta }_{2}^{\ast })/\partial \theta _{j}^{(2)}$, $j=1,\ldots ,q_{2}$, if
the vector $X_{t}^{(2)}$ of link-indicator variables associated with the $t$%
-th element in $U_{2}$ equals the vector $\mathbf{x}\in \Omega -\{\mathbf{0}%
\}$;

\item[(b)] $V_{t,1}^{(2)}=-\left[ 1-\pi _{\mathbf{0}}^{(2)}(\boldsymbol{\theta }%
_{2}^{\ast })\right] /$ $\pi _{\mathbf{0}}^{(2)}(\boldsymbol{\theta }_{2}^{\ast
})$ and $V_{t,j+1}^{(2)}=[\pi _{\mathbf{0}}^{(2)}(\boldsymbol{\theta }_{2}^{\ast
})]^{-1}\partial \pi _{\mathbf{0}}^{(2)}(\boldsymbol{\theta }%
_{2}^{\ast })/\partial \theta _{j}^{(2)}$, $j=1,\ldots ,q_{2}$, if the
vector $X_{t}^{(2)}$ of link-indicator variables associated with the $t$-th
element in $U_{2}$ equals the vector $\mathbf{0}\in \Omega $.

\end{description}

\section{Conclusions}

Whenever we want to apply the results that we have obtained in this research to
an actual situation we need to determine whether or not the assumed
conditions are reasonably well satisfied by those observed in the actual
scenario. In particular \ we have assumed that the numbers $M_{i}$s of
people found in the sampled sites follow a multinomial distribution with
homogeneous cell probabilities and that the $M_{i}$s go to infinity while
the number of sites $n$ in the sample and $N$ in the frame are fixed. These
assumptions imply that in the actual scenario the $M_{i}$s should be
relatively large and not very variable. However, we do not know how large
they should be so that the results can be safely used. Therefore, Monte
Carlo studies are required to assess the reliability of the asymptotic
results under different scenarios with finite samples and populations. In
addition, although we have assumed a general parametric model for the
link-probabilities which allows the possibility that the parameter depends or
not on the sampled sites, the model precludes that the probabilities
depend on the $M_{i}$s as they go to infinity. Furthermore, this assumption
assures that the estimators of $\tau _{1}$ and $\tau _{2}$ be independent
and not only conditionally independent given the $M_{i}$s.

An alternative asymptotic framework to the one considered in this work is to
assume that the numbers of sites $n$ in the sample and $N$ in the frame go to 
infinity whereas the $M_{i}$s are fixed. However, this would involve dealing 
with multinomial distributions with infinite numbers of cells. An approach that 
could be used to derive asymptotic properties of estimators under this framework 
is the one considered by Rao (1958) who derived asymptotic properties of a maximum 
likelihood estimator of a parameter on which depend the cell probabilities of a 
multinomial distribution with infinite number of cells. However, this is a topic 
of a future research.

\section*{Acknowledgements}

 This research was partially supported by Grant PIFI-2013-25-73-1.4.3-8 
 from the Secretar{\'\i}a de Educaci\'on P\'ublica to Universidad Aut\'onoma de 
 Sinaloa.

\nocite{FelixThompson04}
\nocite{FelixMonjardin06}
\nocite{coullagresti99}
\nocite{agresti02}
\nocite{Felixetal15}
\nocite{Kalton09}
\nocite{Sanathanan72}
\nocite{Magnanietal05}
\nocite{Varadhan08}
\nocite{Spreen92}
\nocite{ThompsonFrank00}
\nocite{JohnstonSabin10}
\nocite{Birch64}
\nocite{Rao73}
\nocite{Rao58}
\nocite{Bishopetal75}
\nocite{Feller68}

\bibliographystyle{achicago}
\bibliography{ref-lts-asymptotics-rt}

\end{document}